\documentclass{amsart}

\usepackage{amscd}
\usepackage{bbm}
\usepackage{mathrsfs}
\usepackage{eucal}
\usepackage{faktor}
\usepackage{xfrac}
\usepackage{braket}
\usepackage{relsize} 
\usepackage{nccmath}
\usepackage{mathtools}
\usepackage{comment}
\usepackage{verbatim}
\usepackage{wasysym}
\usepackage{graphicx}
\usepackage{caption}
\usepackage{import}
\usepackage{latexsym}
\usepackage{amsfonts}
\usepackage{amssymb}
\usepackage{amsmath}
\usepackage{mathtext}
\usepackage{cite}
\usepackage{enumerate}
\usepackage{float}
\usepackage{amsthm}
\usepackage{upgreek}
\usepackage{nicefrac}
\usepackage{multicol}
\usepackage{rotating}
\usepackage{array}
\usepackage{blkarray}

\usepackage{etoolbox}
\makeatletter
\pretocmd{\chapter}{\addtocontents{toc}{\protect\addvspace{15\p@}}}{}{}
\pretocmd{\section}{\addtocontents{toc}{\protect\addvspace{5\p@}}}{}{}
\makeatother
\let\oldtocsection=\tocsection
\let\oldtocsubsection=\tocsubsection
\let\oldtocsubsubsection=\tocsubsubsection
\renewcommand{\tocsection}[2]{\hspace{0em}\oldtocsection{#1}{#2}}
\renewcommand{\tocsubsection}[2]{\hspace{1.8em}\oldtocsubsection{#1}{#2}}
\renewcommand{\tocsubsubsection}[2]{\hspace{4.4em}\oldtocsubsubsection{#1}{#2}}

\usepackage[svgnames]{xcolor}
\definecolor{linkcolor}{HTML}{e88d67} 
\definecolor{citecolor}{HTML}{e88d67} 
\definecolor{urlcolor}{HTML}{e88d67} 
\definecolor{myNewColorA}{HTML}{fec3a6}
\definecolor{myNewColorB}{HTML}{ffaf80}
\definecolor{myNewColorC}{HTML}{fb8f67}
\definecolor{seagreen}{HTML}{337180}
\definecolor{mseagreen}{HTML}{369673}
\definecolor{darksalmon}{HTML}{e88d67}
\definecolor{silver}{HTML}{bbbbbb}
\definecolor{flowerblue}{HTML}{4e77fc}
\definecolor{tomato}{HTML}{ff6347}
\definecolor{orange}{HTML}{f2b13d}
\definecolor{darkgray}{HTML}{939393}

\usepackage[
  bookmarks = true,
  bookmarksopen = true,
  colorlinks = , 
  urlcolor = urlcolor, 
  linkcolor = linkcolor, 
  citecolor = citecolor, 
  linktoc = all,
  pagebackref = true
]{hyperref} 
\mathtoolsset{showonlyrefs,showmanualtags}

\DeclareMathOperator{\PI}{PI}
\DeclareMathOperator{\PPII}{PII}

\DeclareMathOperator{\diag}{diag}

\DeclareMathOperator{\cchar}{char}

\newcommand{\brackets}[1]{\left( #1 \right)}

\newcommand{\abs}[1]{\left| #1 \right|}

\newcommand{\Painleve}{Painlev{\'e} }

\theoremstyle{plain}
\newtheorem{thm}{Theorem}[]

\newtheorem{lem}{Lemma}[]

\newtheorem{prop}{Proposition}[]

\theoremstyle{definition}
\newtheorem{defn}{Definition}[]

\newtheorem{exmp}{Example}[]

\theoremstyle{remark}
\newtheorem{rem}{Remark}

\usepackage{chngcntr}
\counterwithin*{thm}{section} 
\counterwithin*{lem}{section} 
\counterwithin*{claim}{section} 
\counterwithin*{prop}{section} 
\counterwithin*{corl}{section} 
\counterwithin*{defn}{section} 
\counterwithin*{exmp}{section} 
\counterwithin*{ex}{section} 
\counterwithin*{rem}{section} 

\numberwithin{equation}{section}

\usepackage{apptools}
\AtAppendix{\counterwithin{thm}{section}}
\AtAppendix{\counterwithin{lem}{section}}
\AtAppendix{\counterwithin{claim}{section}}
\AtAppendix{\counterwithin{prop}{section}}
\AtAppendix{\counterwithin{corl}{section}}
\AtAppendix{\counterwithin{defn}{section}}
\AtAppendix{\counterwithin{exmp}{section}}
\AtAppendix{\counterwithin{ex}{section}}
\AtAppendix{\counterwithin{rem}{section}}

\allowdisplaybreaks
\usepackage{tikz}

\begin{document}

    \title[]{Non-Abelian discrete Toda chains and related lattices}

    \author{Irina Bobrova}
    \noindent\address{\noindent 
    \phantom{\hspace{-5.5mm}}
    Max-Planck-Institute for Mathematics in the Sciences, 04103, Leipzig, Germany
    \newline
    Laboratoire de Math{\'e}matiques de Reims, CNRS UMR 9008,
    Universit{\'e} de Reims-Champagne-Ardenne,
    F-51687, Reims, France
    \newline
    IITP (Institute for Information Transmission Problems), RAS, Moscow, 127051, Russia
    }
    \email{ia.bobrova94@gmail.com}
    
    \author{Vladimir Retakh}
    \noindent\address{\noindent 
    \phantom{\hspace{-5.5mm}}
    Department of Mathematics,
    Rutgers University,
    Piscataway, New Jersey, 08854, USA
    }
    \email{vretakh@math.rutgers.edu}
    
    \author{Vladimir Rubtsov}
    \noindent\address{\noindent 
    \phantom{\hspace{-5.5mm}}
    Maths Department,
    Univ. Angers, 
    CNRS, 
    LAREMA, 
    SFR MATHSTIC, 
    F-49000, Angers, France
    \newline
    IGAP (Institute for Geometry and Physics), 34136, Trieste, Italy
    \newline
    IITP (Institute for Information Transmission Problems), RAS, Moscow, 127051, Russia
    }
    \email{volodya@univ-angers.fr}
    
    \author{Georgy Sharygin}
    \noindent\address{\noindent 
    \phantom{\hspace{-5.5mm}}
    Department of Mathematics and Mechanics,
    Moscow State (Lomonosov) University,
    Leninskie Gory, d. 1,
    Moscow, 119991, Russia
    }
    \email{sharygin@itep.ru}

    \subjclass{Primary 37K20. Secondary 16B99, 16W25} 
    \keywords{
    non-commutative lattices,
    Toda equations,
    hexagonal Boussinesq lattice,
    Korteweg-de Vries lattice,
    modified Korteweg-de Vries lattice,
    sine-Gordon lattice,
    Somos sequences,
    $q$-Painlevé equations,
    quasideterminants, 
    Lax pairs}
    
    \maketitle 

    \begin{abstract}
    We have derived a non-abelian analog for the two-dimensional discrete Toda lattice which possesses solutions in terms of quasideterminants and admits Lax pairs of different forms. Its connection with non-abelian analogs for several well-known (1+1) and one-dimensional lattices is discussed. In particular, we consider a non-commutative analog of the scheme: discrete Toda equations $\rightarrow$ Somos-$N$ sequences $\rightarrow$ discrete Painlevé equations. 
    \end{abstract}

    
    \section*{Introduction}

In the present paper, we consider a non-abelian analog of the well-known two-dimensional discrete Toda equation, also called the Hirota or KP equation. In the commutative case, this equation may be regarded as the master equation, since it can be reduced to many other important integrable systems. It is also known that Toda lattices may be interpreted as the determinant Jacobi identity, if one presents their solutions in a determinant form. By using a non-abelian Jacobi identity for quasideterminants \cite{gel1991determinants}, we have derived the non-commutative two-dimensional discrete Toda lattice, the \ref{eq:2ddToda_0}, differ from \cite{nimmo2006non} (see Remark \ref{rem:nim}). Our analog, the \ref{eq:2ddToda_0} equation, possesses Lax pairs of different types and solutions in terms of quasideterminants. Note that the form of the \ref{eq:2ddToda_0} is more suitable for the further reductions to (1+1)-dimensional and one-dimensional systems. In particular, we are interested in an extension to the non-commutative case of the following scheme
\begin{align}
    \text{discrete Toda equations}
    &&
    \overset{\text{\cite{hone2017reductions}}}{\longrightarrow}&
    &&
    \text{Somos-$N$ equations}
    &&
    \overset{\text{\cite{hone2014discrete}}\,\,}{\longrightarrow}&
    &&
    \text{discrete Painlev\'e equations}.
\end{align}

The paper is separated into three parts: \textbf{(1)} non-abelian analogs of the Toda lattices, \textbf{(2)} reductions of the \ref{eq:2ddToda_0} to (1+1)-dimensional systems and \textbf{(3)} to one-dimensional lattices. As a result, we have discovered connections between the \ref{eq:2ddToda_0}, its continuous analog and one-dimensional discrete and continuous lattices. Extending well-known reductions of the two-dimensional abelian discrete Toda equation to (1+1)-dimensional systems \cite{date1983method} to the non-commutative case, non-abelian analogs of the hexagonal Boussinesq, Korteweg-de Vries (KdV), modified Korteweg-de Vries (mKdV), and sine-Gordon lattices were obtained together with their Lax pairs. Let us stress that some non-abelian discrete systems were discussed in the papers \cite{volkov1995quantum}, \cite{bobenko2002integrable}, \cite{adler2009discrete}. Plane-wave reductions of the \ref{eq:2ddToda_0} and its scalar Lax pair lead to a non-commutative analog of the autonomous Somos-$N$ like sequences and the corresponding isospectral pairs. Note that  commutative Somos-$N$ like sequences are associated with  $q$PI and $q$PII equations and their hierarchies \cite{hone2014discrete}. We remark that some observations of this paper were were published a little earlier by N. Okubo in \cite{okubo2013discrete}. In order to generalize this connection to the non-abelian case, one needs to get a non-autonomous analog of Somos-$N$ like sequences. Here we suggest a method which in the commutative case is different from those in \cite{hone2014discrete}. This allows us to obtain a non-abelian non-autonomous Somos-$N$ like equation and, thus, non-commutative analogs of the $q$-\Painleve equations. Note that the problem of finding the corresponding Lax pairs for these non-autonomous systems remains open even in the commutative case. 

For the reader convenience, we have gathered main results of the paper and necessary preliminaries in Sections \ref{sec:mainres} and \ref{sec:prel}, respectively, while detailed discussions of the resuls can be found in Sections \ref{sec:main11} -- \ref{sec:main12}, \ref{sec:2ddTLtosys}, and~\ref{sec:Somos_dPainleve}. 

\subsection*{Acknowledgements}
The authors are deeply grateful to V. Adler for invaluable comments. We are also
thankful to A. Hone, A. Zabrodin, and A. Khakimova for discussions. The research of I.B. was partially
supported by the International Laboratory of Cluster Geometry HSE, RF Government grant № 075-15-
2021-608, and by Young Russian Mathematics award; V.Rb. is grateful to Centre Henri Lebesgue, program
ANR-11-LABX-0020-0; G.Sh. was partially supported by the National Key
R\&D Program of China (Grant No. 2020YFE0204200). Contribution of G.Sh. in 
Section \ref{sec:main11} was done under support of the Russian Science Foundation (Project No. 22-11-00272).
Contribution of I.B. and V.Rb. in Section \ref{sec:2ddTLtosys} was done under support of the Russian Science Foundation (Project No. 23-41-00049).

The authors are deeply grateful to the referees for the attention to this paper, as well as for constructive and useful comments and suggestions.

\section{Main results}
\label{sec:mainres}

\subsection{Non-Abelian Toda lattices}
\label{sec:mainres1}
The first part of the paper is devoted to a derivation of a non-abelian analog of the two-dimensional discrete Toda equation and  its Lax pairs of different types (scalar and matrix). We have also discussed their reduction to the one-dimensional Toda lattice and continuous limits to the two- and one-dimensional Toda field lattices. 

As we have already mentioned in the Introduction, abelian Toda lattices may be interpreted as the Jacobi identity for determinants, if one presents solutions of the corresponding Toda lattice in a determinant form. Thanks to a non-abelian analog of the determinant Jacobi identity \cite{gel1991determinants}, this method of a derivation of the Toda lattices can be extended to the non-commutative case. We follow an idea given by Theorem 4.2 \cite{gel1992theory} (see also Theorem 9.7.2 from \cite{gelfand2005quasideterminants}), where the authors had obtained a non-abelian version of the one-dimensional discrete Toda equation. We consider elements of an associative unital division ring $R$ over a field $F$, $\cchar F = 0$ (for more detail see Subsection \ref{sec:ncset}) and often refer to the elements of the ring $R$ as to \textit{functions}. Let us introduce matrix
\begin{align}
    \label{eq:Theta_def_0}
    \Theta_{l, m, n}
    &:= 
    \begin{pmatrix}
    \varphi_{l, m} 
    & \dots & \varphi_{l, m + n - 1}
    \\
    \vdots & \ddots & \vdots
    \\
    \varphi_{l + n - 1, m} 
    & \dots & 
    \varphi_{l + n - 1, m + n - 1}
    \end{pmatrix}
    = \brackets{
    \varphi_{l + i - 1, m + j - 1}
    }_{1 \leq i, j \leq n},
    &
    (l, m, n + 1)
    &\in \mathbb{Z}_{\geq 0}^{3},
\end{align}
where $\varphi_{i, j} \in R$ and set
$\begin{aligned}
    \theta_{l, m, n}
    = |\Theta_{l, m, n}|_{nn}.
\end{aligned}$
Then, one can verify (see Proposition \ref{thm:2ddToda_nc}) that from a non-abelian Jacobi identity for the quasideterminants it follows that the functions $\theta_{l, m, n}$ satisfy the equation
\begin{align}
    \label{eq:2ddToda_0}
    \tag*{2ddTL}
    \theta_{l + 1, m + 1, n}
    &= \theta_{l, m, n + 1}
    + \theta_{l + 1, m, n} \brackets{
    \theta_{l, m, n}^{-1}
    - \theta_{l + 1, m + 1, n - 1}^{-1}
    } \theta_{l, m + 1, n}.
\end{align}
We call the latter equation a non-abelian analog of the two-dimensional discrete Toda lattice and  denote it by \ref{eq:2ddToda_0}. It turns out that equation \ref{eq:2ddToda_0} has Lax pairs of different types (scalar, matrix $2\times2$, and matrix semi-infinite), which we have presented in Propositions \ref{thm:2ddToda_scalsys}, \ref{thm:2ddToda_matsys}, and \ref{thm:2ddToda_matsys_n}. In particular, a scalar Lax pair is given by the system
\begin{gather}
    \label{eq:2ddToda_scalsys_0}
    \left\{
    \begin{array}{rcl}
         \psi_{l + 1, m, n}
         &=& \psi_{l, m, n + 1}
         + a_{l, m, n} \psi_{l, m, n},
         \\[2mm]
         \psi_{l, m - 1, n + 1}
         &=& \psi_{l, m, n + 1}
         + b_{l, m, n} \psi_{l, m, n},
    \end{array}
    \right.
    \\[3mm]
    \label{eq:2ddToda_abdef_0}
    \begin{aligned}
    a_{l, m, n}
    &= \theta_{l + 1, m, n} \theta_{l, m, n}^{-1},
    &&&&&
    b_{l, m, n}
    &= \theta_{l, m - 1, n + 1} \theta_{l, m, n}^{-1}.
    \end{aligned}
\end{gather}
This systems can be easily rewritten in the matrix form in terms of either $2\times2$ or semi-infinite matrices. Note that these pairs generalize well-known Lax pairs \cite{krichever1997quantum}, \cite{zabrodin1997hirota} for the abelian two-dimensional discrete Toda equation.

\begin{rem}
\label{rem:nim}
In the paper \cite{nimmo2006non}, the author have introduced different form of the Hirota lattice called a non-abelian Hirota-Miwa equation. Namely, it arises as a compatibility condition of the linear system
\begin{align}
    &&
    \varphi_{, i} - \varphi_{, j}
    + U_{ij} \, \varphi
    &= 0,
    &
    i,j
    &\in \mathbb{N}
    ,
    &&
\end{align}
where $\varphi \in {R}$ and $U_{ij} \in {R}^{\times}$ (for $i \neq j$). The final form of the compatibility condition is
\begin{align}
    &&
    U_{ij}
    + U_{jk}
    + U_{ki}
    &= 0,
    &
    U_{ij, k} \, U_{ki}
    &= U_{ki, j} \, U_{ij}
    .
    &&
\end{align}
\end{rem}

It is natural to expect that the \ref{eq:2ddToda_0} somehow reduces to its continuous analog. Indeed, regarding a continuous limit as in Proposition \ref{thm:2dToda_scalsys}, one can arrive at the equation
\begin{align}
    \label{eq:2dToda_0}
    \tag*{2dTL}
    \brackets{
    \theta_{n, x} \theta_n^{-1}
    }_y
    &= \theta_{n + 1} \theta_n^{-1}
    - \theta_n \theta_{n - 1}^{-1}
    ,
\end{align}
which is a well-known non-abelian two-dimensional Toda field equation \cite{mikhailov1981reduction} ( see also \cite{sall1982darboux}). This limit can be extended for the Lax pairs and quasideterminant solutions as well as for the equations (see Subsection~\ref{sec:2dToda}). 

Going further to other Toda lattices, one can also get analogs of the one-dimensional discrete and continuous Toda equations. They are given, respectively, by the lattices
\begin{align}
    \label{eq:1ddToda_0}
    \tag*{1ddTL}
    \theta_{m + 2, n}
    &= \theta_{m, n + 1}
    + \theta_{m + 1, n} \brackets{
    \theta_{l, m, n}^{-1}
    - \theta_{m + 2, n - 1}^{-1}
    } \theta_{m + 1, n},
\end{align}
\begin{align}
    \label{eq:1dToda_0}
    \tag*{1dTL}
    \brackets{
    \theta_{n}' \theta_n^{-1}
    }'
    &= \theta_{n + 1} \theta_n^{-1}
    - \theta_n \theta_{n - 1}^{-1}
    .
\end{align}
These equations coincide with those given, for instance, in the paper \cite{gelfand2005quasideterminants}. As for the \ref{eq:2dToda_0}, Lax pairs and quasideterminant solutions may be reduced from the \ref{eq:2ddToda_0}. 

\begin{rem}
The 1d Toda lattice \ref{eq:1dToda_0} is related to the non-abelian Volterra lattice 
\begin{align}
    \label{eq:1dV_nc_0}
    \tag*{1dVL}
    \gamma_{n}'
    &= \gamma_{n + 1} \gamma_{n}
    - \gamma_{n} \gamma_{n - 1}
\end{align}
by the Miura map given in Proposition \ref{thm:1dVLto1dToda}.
\end{rem}

Summing up the discussion above, we formulate the first main result of the paper, a detailed proof of which is given in Subsections \ref{sec:2ddToda} -- \ref{sec:2dToda} and \ref{sec:1ddToda} -- \ref{sec:1dToda}.

\begin{thm}
\label{thm:2ddto1d}
There exist maps connecting the equations, their scalar and matrix linear problems {\rm(}in terms of $2\times2$ and semi-infinite matrices{\rm)}, and quasideterminant solutions of the non-abelian Toda equations 
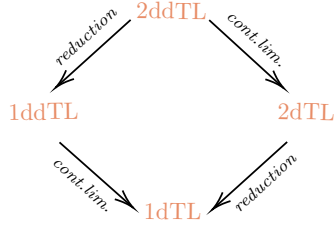
\begin{figure}[H]
    \centering
    \scalebox{.9}{\tikzset{every picture/.style={line width=0.75pt}} 

\begin{tikzpicture}[x=0.75pt,y=0.75pt,yscale=-1,xscale=1]

\draw    (367.23,60.62) -- (407.5,97.43) ;
\draw [shift={(408.98,98.78)}, rotate = 222.42] [color={rgb, 255:red, 0; green, 0; blue, 0 }  ][line width=0.75]    (10.93,-3.29) .. controls (6.95,-1.4) and (3.31,-0.3) .. (0,0) .. controls (3.31,0.3) and (6.95,1.4) .. (10.93,3.29)   ;
\draw    (323.2,60.5) -- (282.62,97.19) ;
\draw [shift={(281.14,98.53)}, rotate = 317.88] [color={rgb, 255:red, 0; green, 0; blue, 0 }  ][line width=0.75]    (10.93,-3.29) .. controls (6.95,-1.4) and (3.31,-0.3) .. (0,0) .. controls (3.31,0.3) and (6.95,1.4) .. (10.93,3.29)   ;
\draw    (283.34,129.18) -- (323.2,164.56) ;
\draw [shift={(324.69,165.89)}, rotate = 221.6] [color={rgb, 255:red, 0; green, 0; blue, 0 }  ][line width=0.75]    (10.93,-3.29) .. controls (6.95,-1.4) and (3.31,-0.3) .. (0,0) .. controls (3.31,0.3) and (6.95,1.4) .. (10.93,3.29)   ;
\draw    (410.91,128.96) -- (370.34,165.64) ;
\draw [shift={(368.86,166.99)}, rotate = 317.88] [color={rgb, 255:red, 0; green, 0; blue, 0 }  ][line width=0.75]    (10.93,-3.29) .. controls (6.95,-1.4) and (3.31,-0.3) .. (0,0) .. controls (3.31,0.3) and (6.95,1.4) .. (10.93,3.29)   ;

\draw (325.,48.72) node [anchor=north west][inner sep=0.75pt]  [rotate=-0.65] [align=left] {\rm\ref{eq:2ddToda_0}};
\draw (404.26,105.19) node [anchor=north west][inner sep=0.75pt]  [rotate=-1.33] [align=left] {\rm\ref{eq:2dToda_0}};
\draw (380.69,51.98) node [anchor=north west][inner sep=0.75pt]  [font=\scriptsize,rotate=-42.57] [align=left] {cont.lim.};
\draw (253.02,105.9) node [anchor=north west][inner sep=0.75pt]  [rotate=-359.16] [align=left] {\rm\ref{eq:1ddToda_0}};
\draw (276.74,82.32) node [anchor=north west][inner sep=0.75pt]  [font=\scriptsize,rotate=-317.35] [align=left] {reduction};
\draw (283.05,134.04) node [anchor=north west][inner sep=0.75pt]  [font=\scriptsize,rotate=-41.75] [align=left] {cont.lim.};
\draw (377.49,164.2) node [anchor=north west][inner sep=0.75pt]  [font=\scriptsize,rotate=-317.27] [align=left] {reduction};
\draw (329.29,161.15) node [anchor=north west][inner sep=0.75pt]  [rotate=-0.65] [align=left] {\rm\ref{eq:1dToda_0}};

\end{tikzpicture}}
\end{figure}
\end{thm}

\medskip
\subsection{Non-Abelian analogs of some discrete (1 + 1) systems}
\label{sec:mainres2}
The second part of the paper, Section~\ref{sec:2ddTLtosys}, is devoted to reductions of the \ref{eq:2ddToda_0} to different non-abelian analogs of some discrete $(1 + 1)$ systems. We consider  KdV-like and periodic reductions and, as a result, derived non-abelian analogs of  hexagonal discrete Boussinesq equation, discrete Korteweg-de Vries equation and its modified analog, and discrete sine-Gordon lattice. Under commutative reduction, all our analogs coincide with well-known abelian lattices (see, e.g. \cite{date1983method}). We also obtain Lax pairs for these analogs. 

The KdV-like reductions are given by a linear change of the indices $(l, m, n) \mapsto (f(l, m, n), g(l, m, n))$ and lead, in particular, to the discrete Boussinesq and KdV lattices. Regarding such maps for the equations \ref{eq:2ddToda_0}, we arrive to the following non-abelian analogs of hexagonal Boussinesq equation
\begin{align}
    \label{eq:dB_nc_0}
    \tag*{dBSQ}
    \theta_{m, n - 1} \theta_{m - 1, n - 1}^{-1}
    - \theta_{m - 1, n} \theta_{m - 1, n - 1}^{-1}
    &= \theta_{m + 1, n + 1} \theta_{m, n + 1}^{-1}
    - \theta_{m + 1, n + 1} \theta_{m + 1, n}^{-1}
\end{align}
and to the KdV equation
\begin{align}
    \label{eq:dKdV_nc_0}
    \tag*{dKdV}
    &&
    v_{m + 1, n + 1}
    - v_{m, n}
    &= v_{m, n + 1}^{-1}
    - v_{m + 1, n}^{-1},
    &
    v_{m, n}
    &= \theta_{m, n} \theta_{m - 1, n}^{-1}
    .
    &&
\end{align}
These two lattices and their Lax pairs are discussed in Subsection \ref{sec:kdvred}. 

In Subsection \ref{sec:perred}, we consider a periodic reduction $(n \mod 2)$ for the \ref{eq:2ddToda_0}. This leads us to non-abelian discrete sine-Gordon equation
\begin{align}
    \label{eq:dsG_nc_0}
    \tag*{$\text{dsG}_1$}
    \theta_{l + 1, m + 1}
    \theta_{l, m + 1}^{-1}
    - \theta_{l, m}^{-1} \theta_{l, m + 1}^{-1}
    &= \theta_{l + 1, m} \theta_{l, m}^{-1}
    - \theta_{l + 1, m}
    \theta_{l + 1, m + 1}
    .
\end{align}
If one supplements the periodic reduction with a KdV-like reduction, then one can derive the fllowing non-abelian discrete modified KdV equation:
\begin{align}
    \label{eq:dmKdV_nc_1_0}
    \tag*{$\text{dmKdV}_1$}
    \theta_{l + 1, m + 1}^{-1} 
    \theta_{l + 1, m}
    - \theta_{l + 1, m + 1}^{-1} \theta_{l, m + 1}
    &= \theta_{l, m + 1}
    \theta_{l, m}^{-1} 
    - \theta_{l + 1, m}
    \theta_{l, m}^{-1}.
\end{align}
We also present Lax pairs for the above lattices.
Note that in the commutative case these two lattices are connected by the change of variables $\theta_{l, m} \mapsto \theta_{l, m}^{-1}$ for even $m$ and $\theta_{l, m} \mapsto \theta_{l, m}$ for odd $m$ \cite{quispel1991integrable}, which is also applicable in the non-abelian setting. 

Motivated by the well-known link between the dKdV and dmKdV lattices (see, e.g. \cite{nijhoff1995discrete}),
\begin{align}
    &&&&
    {\rm dKdV}
    &&
    \longleftrightarrow
    &&
    {\rm dpKdV}
    &&
    \overset{\text{Miura}}{\longrightarrow}
    &&
    {\rm dmKdV}
    ,
    &&&&
\end{align}
where a dpKdV is a potential KdV lattice \cite{wiersma1987lattice}, we want to establish the same connection for the \ref{eq:dKdV_nc_0} and \ref{eq:dmKdV_nc_1_0} equations. However, an extension of this scheme to a non-abelian case leads to one more analog of the dmKdV lattice. We will briefly describe it now (see also Subsection \ref{sec:pkdv}).

It turns out that equation \ref{eq:dKdV_nc_0} has the following potential form (see Theorem \ref{thm:dKdVtodpKdV_nc})
\begin{align}
    \label{eq:dpKdV_nc_1_0}
    \tag*{dpKdV}
    \brackets{
    u_{m, n + 1} - u_{m + 1, n}
    } \brackets{
    u_{m, n} - u_{m + 1, n + 1}
    }
    &= p \, q.
\end{align}
By generalizing commutative Miura map to the non-abelian case (see Proposition \ref{thm:pKdVtomKdV}), the latter lattice maps to the following non-abelian version of the dmKdV equation
\begin{align}
    \label{eq:dmKdV_nc_2_0}
    \tag*{$\text{dmKdV}_2$}
    \theta_{m, n + 1} \theta_{m + 1, n + 1}^{-1}
    - \theta_{m + 1, n} \theta_{m + 1, n + 1}^{-1}
    = \theta_{m + 1, n} \theta_{m, n}^{-1}
    - \theta_{m, n + 1} \theta_{m, n}^{-1}
    .
\end{align}
To the authors' knowledge, it cannot be obtained directly from the \ref{eq:2ddToda_0} and it is not equivalent to the \ref{eq:dmKdV_nc_1_0}. Note that in the commutative case there is no difference between the \ref{eq:dmKdV_nc_1_0} and \ref{eq:dmKdV_nc_2_0} lattices. Using the same change between dmKdV and dsG equations, one can derive a one more non-abelian analog of the discrete sine-Gordon equation:
\begin{align}
    \label{eq:dsG_nc_2_0}
    \tag*{$\text{dsG}_2$}
    \theta_{m, n + 1}^{-1} \theta_{m + 1, n + 1}
    - \theta_{m + 1, n} \theta_{m + 1, n + 1}
    = \theta_{m + 1, n} \theta_{m, n}^{-1}
    - \theta_{m, n + 1}^{-1} \theta_{m, n}^{-1}
    .
\end{align}
The systems \ref{eq:dpKdV_nc_1_0}, \ref{eq:dmKdV_nc_2_0}, and \ref{eq:dsG_nc_2_0} are integrable in the sense of existence of Lax pairs. 

\medskip
Summarizing the results, we get the following
\begin{thm}
The 2d discrete non-abelian Toda lattice {\rm\ref{eq:2ddToda_0}} and its Lax pair may be reduced to the given non-abelian $(1+1)$ discrete systems and the corresponding Lax pairs, i.e. the following diagram holds
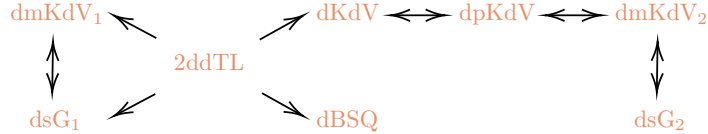
\begin{figure}[H]
    \centering
    \scalebox{.9}{\tikzset{every picture/.style={line width=0.75pt}} 

\begin{tikzpicture}[x=0.75pt,y=0.75pt,yscale=-1,xscale=1]

\draw    (359.69,82.85) -- (382.45,70.17) ;
\draw [shift={(384.2,69.2)}, rotate = 150.88] [color={rgb, 255:red, 0; green, 0; blue, 0 }  ][line width=0.75]    (10.93,-3.29) .. controls (6.95,-1.4) and (3.31,-0.3) .. (0,0) .. controls (3.31,0.3) and (6.95,1.4) .. (10.93,3.29)   ;
\draw    (360.97,112.85) -- (382.45,124.83) ;
\draw [shift={(384.2,125.8)}, rotate = 209.14] [color={rgb, 255:red, 0; green, 0; blue, 0 }  ][line width=0.75]    (10.93,-3.29) .. controls (6.95,-1.4) and (3.31,-0.3) .. (0,0) .. controls (3.31,0.3) and (6.95,1.4) .. (10.93,3.29)   ;
\draw    (435,69.2) -- (461.6,69.2) ;
\draw [shift={(463.6,69.2)}, rotate = 180] [color={rgb, 255:red, 0; green, 0; blue, 0 }  ][line width=0.75]    (10.93,-3.29) .. controls (6.95,-1.4) and (3.31,-0.3) .. (0,0) .. controls (3.31,0.3) and (6.95,1.4) .. (10.93,3.29)   ;
\draw [shift={(433,69.2)}, rotate = 0] [color={rgb, 255:red, 0; green, 0; blue, 0 }  ][line width=0.75]    (10.93,-3.29) .. controls (6.95,-1.4) and (3.31,-0.3) .. (0,0) .. controls (3.31,0.3) and (6.95,1.4) .. (10.93,3.29)   ;
\draw    (521.6,69.2) -- (549.2,69.2) ;
\draw [shift={(519.6,69.2)}, rotate = 0] [color={rgb, 255:red, 0; green, 0; blue, 0 }  ][line width=0.75]    (10.93,-3.29) .. controls (6.95,-1.4) and (3.31,-0.3) .. (0,0) .. controls (3.31,0.3) and (6.95,1.4) .. (10.93,3.29)   ;
\draw [shift={(551.2,69.2)}, rotate = 180] [color={rgb, 255:red, 0; green, 0; blue, 0 }  ][line width=0.75]    (10.93,-3.29) .. controls (6.95,-1.4) and (3.31,-0.3) .. (0,0) .. controls (3.31,0.3) and (6.95,1.4) .. (10.93,3.29)   ;
\draw    (582.63,85.8) -- (582.25,110.68) ;
\draw [shift={(582.22,112.68)}, rotate = 270.88] [color={rgb, 255:red, 0; green, 0; blue, 0 }  ][line width=0.75]    (10.93,-3.29) .. controls (6.95,-1.4) and (3.31,-0.3) .. (0,0) .. controls (3.31,0.3) and (6.95,1.4) .. (10.93,3.29)   ;
\draw [shift={(582.66,83.8)}, rotate = 90.88] [color={rgb, 255:red, 0; green, 0; blue, 0 }  ][line width=0.75]    (10.93,-3.29) .. controls (6.95,-1.4) and (3.31,-0.3) .. (0,0) .. controls (3.31,0.3) and (6.95,1.4) .. (10.93,3.29)   ;
\draw    (279.56,70.14) -- (301.76,81.98) ;
\draw [shift={(277.8,69.2)}, rotate = 28.08] [color={rgb, 255:red, 0; green, 0; blue, 0 }  ][line width=0.75]    (10.93,-3.29) .. controls (6.95,-1.4) and (3.31,-0.3) .. (0,0) .. controls (3.31,0.3) and (6.95,1.4) .. (10.93,3.29)   ;
\draw    (279.52,124.86) -- (301.17,113.35) ;
\draw [shift={(277.75,125.8)}, rotate = 332.01] [color={rgb, 255:red, 0; green, 0; blue, 0 }  ][line width=0.75]    (10.93,-3.29) .. controls (6.95,-1.4) and (3.31,-0.3) .. (0,0) .. controls (3.31,0.3) and (6.95,1.4) .. (10.93,3.29)   ;
\draw    (243.83,85.8) -- (243.45,110.68) ;
\draw [shift={(243.42,112.68)}, rotate = 270.88] [color={rgb, 255:red, 0; green, 0; blue, 0 }  ][line width=0.75]    (10.93,-3.29) .. controls (6.95,-1.4) and (3.31,-0.3) .. (0,0) .. controls (3.31,0.3) and (6.95,1.4) .. (10.93,3.29)   ;
\draw [shift={(243.86,83.8)}, rotate = 90.88] [color={rgb, 255:red, 0; green, 0; blue, 0 }  ][line width=0.75]    (10.93,-3.29) .. controls (6.95,-1.4) and (3.31,-0.3) .. (0,0) .. controls (3.31,0.3) and (6.95,1.4) .. (10.93,3.29)   ;

\draw (310.15,89.2) node [anchor=north west][inner sep=0.75pt]   [align=left] {\rm\ref{eq:2ddToda_0}};
\draw (389.95,121) node [anchor=north west][inner sep=0.75pt]   [align=left] {\rm\ref{eq:dB_nc_0}};
\draw (389.95,61.2) node [anchor=north west][inner sep=0.75pt]   [align=left] {\rm\ref{eq:dKdV_nc_0}};
\draw (469.95,61.2) node [anchor=north west][inner sep=0.75pt]   [align=left] {\rm\ref{eq:dpKdV_nc_1_0}};
\draw (557.35,61.2) node [anchor=north west][inner sep=0.75pt]   [align=left] {\rm\ref{eq:dmKdV_nc_2_0}};
\draw (567.95,121) node [anchor=north west][inner sep=0.75pt]   [align=left] {\rm\ref{eq:dsG_nc_2_0}};
\draw (218.55,61.2) node [anchor=north west][inner sep=0.75pt]   [align=left] {\rm\ref{eq:dmKdV_nc_1_0}};
\draw (229.15,121) node [anchor=north west][inner sep=0.75pt]   [align=left] {\rm\ref{eq:dsG_nc_0}};

\end{tikzpicture}}
\end{figure}
\end{thm}
It would be interesting to consider different continuous limits of the above discrete systems. We leave this issue for our forthcoming papers.

\medskip
\subsection{Non-Abelian Somos-\texorpdfstring{$N$}{N} like sequences and \texorpdfstring{$q$}{q}-Painlevé equations}
\label{sec:mainres3}
The last part of the paper is also devoted to  reductions of \ref{eq:2ddToda_0}, but in the one-dimensional case. By a plane-wave reduction, equation \ref{eq:2ddToda_0} may be reduced to an autonomous Somos-$N$ like equation together with its Lax pair (see Proposition \ref{thm:TodatoSomosN}). A resulting non-abelian analog of the Somos-$N$ like equation can be written as
\begin{gather}
    \label{eq:SomosN_nc_0}
    \begin{aligned}
    y_{M + N}
    &= \alpha^2 \, y_{M + r}
    + y_{M + N - s} \, 
    \brackets{
    y_M^{-1}
    - \alpha^2 \, y_{M + N - r}^{-1} 
    } \,
    y_{M + s},
    &&&
    N 
    &\in \mathbb{N}_{> 3},
    &
    1 
    &\leq r < s \leq \left[\tfrac{N}{2}\right],
    \end{aligned}
\end{gather}
where $\alpha$ is an arbitrary abelian parameter. Since the equation is one-dimensional and autonomous, its linearization is given by an isospectral Lax pair. We present them together with several explicit examples in Subsection \ref{sec:ncsomos}.

It is known that there exists a connection between Somos-$N$ like sequences and $q$-\Painleve equations \cite{hone2014discrete}. An implementation of the non-autonomous parameter in commutative Somos-$N$ sequences is described in Subsection \ref{sec:somoslikeseq}. Note that the method introduced in \cite{hone2014discrete} is not suitable in the non-abelian case. In order to extend this link to the non-commutative case, we have first implement a non-autonomous parameter in the \ref{eq:2ddToda_0} and then use the plane-wave reduction. It allows us, on the hand, to obtain a non-autonomous Somos-$N$ like equation, but, on the other hand, the corresponding Lax pair remains unknown. Let us stress that one-dimensional non-autonomous systems no longer have an isospectral Lax pair and one needs to use a Zakharov-Shabat type pair. 

A non-autonomous non-abelian Somos-$N$ equation derived in Subsection \ref{sec:qPainleve} is given by
\begin{gather}
    \label{eq:SomosN_nc_alM_0}
    \begin{aligned}
    \begin{aligned}
    y_{M + N}
    &= \alpha_M \, y_{M + r}
    + y_{M + s} \, 
    \brackets{
    y_M^{-1}
    - y_{M + N - r}^{-1} \, \alpha_{M - r} 
    } \,
    y_{M + N - s},
    &
    \alpha_M
    &= \alpha \, q^M,
    \end{aligned}
    \\
    \begin{aligned}
    N 
    &\in \mathbb{N}_{> 3},
    &
    1 
    &\leq r < s \leq \left[\tfrac{N}{2}\right].
    \end{aligned}
    \end{aligned}
\end{gather}
Here $\alpha$ is a non-abelian constant parameter. Regarding the analogs of the changes of variables in the paper \cite{hone2014discrete}, this sequence turns into either the $q$PI or $q$PII hierarchy. Namely, let $r = 1$ and $s = 2$, then the changes of variables
\begin{align}
    u_{M}
    &= y_{M + 3} y_{M + 2}^{-1},
    &
    u_{M}
    &= y_{M + 4} y_{M + 2}^{-1}
\end{align}
for even $N = 2 n$ and odd $N = 2 n + 1$ in \eqref{eq:SomosN_nc_alM_0} lead to non-abelian analogs of $q$PI and $q$PII hierarchies, respectively,
\begin{gather}
    \tag*{$q\text{PI}[n]$}
    \label{eq:qPIn_nc_0}
    \begin{aligned}
        u_{M + N - 3} \,
        u_{M + N - 4}
        - u_{M - 1} \, u_{M - 2}
        &= \prod_{k = 0}^{N - 4}
        \brackets{
        \alpha_M \, u_{M + k - 1}^{-1}
        - u_{M + k}^{-1} \, \alpha_{M - 1}
        },
    \end{aligned}
    \\[0mm]
    \tag*{$q\text{PII}[n]$}
    \label{eq:qPIIn_nc_0}
    \begin{aligned}
        \prod_{k = 1 - [\frac{N}{2}]}^0
        u_{M - 2 k - 1}
        - u_{M - 2} \, \prod_{k = 2 - [\frac{N}{2}]}^0
        u_{M - 2 k - 3}
        &= \alpha_M
        - \prod_{k = 0}^{[\frac{N}{2}] - 2}
        u_{M + 2k - 2}^{-1} \,\, 
        \alpha_{M - 1}
        \prod_{k = 2 - [\frac{N}{2}]}^{0}
        u_{M - 2k - 3}.
    \end{aligned}
\end{gather}
A few explicit examples of the above equations and their connection with abelian $q$-\Painleve equations can be found in Subsection \ref{sec:qPainleve}.

\section{Preliminaries}
\label{sec:prel}
\subsection{Hirota's equations}
\label{sec:hirotaeq}

The famous Hirota equation \cite{hirota1981discrete}
\begin{align}
    \label{eq:Hirota}
    &&
    \tau_{l + 1, m, n + 1} \, \tau_{l, m + 1, n}
    &= \tau_{l + 1, m, n} \, \tau_{l, m + 1, n + 1}
    + \tau_{l + 1, m + 1, n} \, \tau_{l, m, n + 1},
    &
    \tau_{l, m, 0}
    &= 1,
    &&
\end{align}
where $\brackets{l, m, n + 1} \in \mathbb{Z}_{\geq 0}^3$, is the most general discrete equation from which one can obtain many other important discrete, semi-discrete, and continuous systems. In particular, this equation is a discrete generalization of the two-dimensional Toda field equation that was introduced in \cite{hirota1993difference} (see eq. (35)):
\begin{align}
    \label{eq:2ddToda}
    &&
    \tau_{l + 1, m + 1, n} \, \tau_{l, m, n}
    &= \tau_{l + 1, m + 1, n - 1} \, \tau_{l, m, n + 1}
    + \tau_{l + 1, m, n} \, \tau_{l, m + 1, n}.
    &&
\end{align}
This equation is connected with \eqref{eq:Hirota} by the map
\begin{align}
    \label{eq:2ddTodatoHirota}
    (l, m, n)
    &\mapsto (l, - m + 1, n + m)
    = (
    \tilde{l}, 
    \tilde{m} + 1, 
    \tilde{n}
    ) 
    ,
\end{align}
where we omitted the sign $\tilde{\phantom{\tau}}$ in \eqref{eq:Hirota}. The Hirota equation can be written in different forms and is also known as a discrete analog of the Kadomtsev–Petviashvili equation \cite{zabrodin1997hirota}.

In order to derive \eqref{eq:Hirota}, R. Hirota used the method of constructing equations whose solutions are expressed via discrete solitons. Moreover,  Hirota equation \eqref{eq:Hirota} is a consequence of the following scalar linear problem
\begin{gather}
    \label{eq:Hirota_scalsys}
    \left\{
    \begin{array}{rcl}
         \psi_{l + 1, m, n}
         - \psi_{l, m, n + 1}
         &=& a_{l, m, n} \, \psi_{l, m, n},
         \\[2mm]
         \psi_{l, m + 1, n}
         - \psi_{l, m, n + 1}
         &=& b_{l, m, n} \, \psi_{l, m, n},
         \\[2mm]
         \psi_{l + 1, m, n}
         - \psi_{l, m + 1, n}
         &=& c_{l, m, n}
         \, \psi_{l, m, n},
    \end{array}
    \right.
    \\[3mm]
    \begin{gathered}
    a_{l, m, n}
    = \tau_{l, m, n} \tau_{l + 1, m, n + 1} (\tau_{l + 1, m, n} \tau_{l, m, n + 1})^{-1},
    \qquad
    b_{l, m, n}
    = \tau_{l, m, n} \tau_{l, m + 1, n + 1} (\tau_{l, m + 1, n} \tau_{l, m, n + 1})^{-1},
    \\[2mm]
    c_{l, m, n}
    = \tau_{l, m, n} \tau_{l + 1, m + 1, n} (\tau_{l + 1, m, n} \tau_{l, m + 1, n})^{-1}.
    \end{gathered}
\end{gather}
We would like to stress that it is enough to consider any two equations of system \eqref{eq:Hirota_scalsys} and the third equation can be recovered from the chosen couple. Let us choose the first two equations, since they can be rewritten in terms of the function $\theta_{l, m, n} = \tau_{l, m, n} \tau_{l, m, n - 1}^{-1}$, and write the linear problem for the 2d discrete Toda lattice \eqref{eq:2ddToda}, by using  \eqref{eq:2ddTodatoHirota}. Then we obtain the system
\begin{gather}
    \label{eq:2ddToda_scalsys_}
    \left\{
    \begin{array}{rcl}
         \psi_{l + 1, m, n}
         &=& \psi_{l, m, n + 1} + a_{l, m, n} \, \psi_{l, m, n},
         \\[2mm]
         \psi_{l, m - 1, n + 1}
         &=& \psi_{l, m, n + 1}
         + b_{l, m, n} \, \psi_{l, m, n},
    \end{array}
    \right.
    \\[3mm]
    \label{eq:2ddToda_abdef_}
    \begin{gathered}
    a_{l, m, n}
    = \theta_{l + 1, m, n} \theta_{l, m, n}^{-1},
    \qquad
    b_{l, m, n}
    = \theta_{l, m - 1, n + 1} \theta_{l, m, n}^{-1}
    .
    \end{gathered}
\end{gather}
Its compatibility condition can be written as
\begin{align}
    T_{- 2, 3} \brackets{
    \psi_{l + 1, m, n}
    }
    &= \psi_{l + 1, m - 1, n + 1}
    = T_{1} \brackets{
    \psi_{l, m - 1, n + 1}
    },
\end{align}
where $T_{\pm i}$, $i = 1, 2, 3$ is a notation for the shift operators
\begin{align}
    &&
    T_{\pm 1} \brackets{f_{l, m, n}}
    &= f_{l \pm 1, m, n},
    &
    T_{\pm 2} \brackets{f_{l, m, n}}
    &= f_{l, m \pm 1, n},
    &
    T_{\pm 3} \brackets{f_{l, m, n}}
    &= f_{l, m, n \pm 1},
    &&
\end{align}
leads to the system for the functions $a_{l, m, n}$ and $b_{l, m, n}$
\begin{align}
\left\{
    \begin{array}{rcl}
        b_{l, m, n + 1} + a_{l, m - 1, n + 1}
        &=& a_{l, m, n + 1} + b_{l + 1, m, n},
        \\[2mm]
        a_{l, m - 1, n + 1} b_{l, m, n}
        &=& b_{l + 1, m, n} a_{l, m, n}
    \end{array}
    \right.
\end{align}
that, after substitution of \eqref{eq:2ddToda_abdef_} into it, reduces to the equation
\begin{align}
    \label{eq:2ddToda_}
    \theta_{l + 1, m + 1, n} \theta_{l, m + 1, n}^{-1}
    = \theta_{l, m, n + 1} \theta_{l, m + 1, n}^{-1}
    + \theta_{l + 1, m, n} \theta_{l, m, n}^{-1} 
    - \theta_{l + 1, m, n} \theta_{l + 1, m + 1, n - 1}^{-1}.
\end{align}
Section \ref{sec:2ddToda} is devoted to a non-abelian generalisation for \eqref{eq:2ddToda_}. 

Let us show how the latter is related to \eqref{eq:2ddToda}. 
Recall that $\theta_{l, m, n} = \tau_{l, m, n} \, \tau_{l, m, n - 1}^{-1}$. Then, \eqref{eq:2ddToda_} can be rewritten in terms of $\tau_{l, m, n}$ as follows
\begin{gather}
     \big(\theta_{l + 1, m + 1, n} - \theta_{l, m, n + 1} \big)
    \theta_{l, m + 1, n}^{-1}
    = \theta_{l + 1, m, n} 
    \brackets{\theta_{l, m, n}^{-1} - \theta_{l + 1, m + 1, n - 1}^{-1}},
    \\[3mm]
    \begin{multlined}
    \brackets{
    \tau_{l + 1, m + 1, n} \tau_{l, m, n}
    - \tau_{l, m, n + 1} \tau_{l + 1, m + 1, n - 1}
    } \, \tau_{l + 1, m, n}^{-1} \tau_{l, m + 1, n}^{-1}
    \hspace{4cm}
    \\[1mm]
    = \brackets{
    \tau_{l + 1, m + 1, n - 1} \tau_{l, m, n - 1}
    - \tau_{l, m, n} \tau_{l + 1, m + 1, n - 2}
    } \, \tau_{l + 1, m, n - 1}^{-1} \tau_{l, m + 1, n - 1}^{-1}
    .
    \end{multlined}
\end{gather}
Since the l.h.s. is differ from the r.h.s. by a shift of $n$, we obtain the equation
\begin{align}
    &&
    \tau_{l + 1, m + 1, n} \tau_{l, m, n}
    &= \tau_{l, m, n + 1} \tau_{l + 1, m + 1, n - 1}
    + \beta_{l, m, n} \, \tau_{l + 1, m, n} \tau_{l, m + 1, n}, 
    &
    \beta_{l, m, n + 1}
    &= \beta_{l, m, n},
    &&
\end{align}
which coincides with \eqref{eq:2ddToda} when $\beta_{l, m, n} = 1$. 

Note that both \eqref{eq:2ddToda} and \eqref{eq:2ddToda_scalsys_} can contain parameters useful for the reductions. In particular, the following change
\begin{align}
    \label{eq:scalrule}
    &&
    \tau_{l, m, n}
    &\mapsto \tau_{l, m, n} \, \alpha^{(l + m) n} \, \beta^{l m}
    ,
    &
    \psi_{l, m, n}
    &\mapsto \lambda_1^{-l} \lambda_2^{m} \lambda_3^{-n} \, \psi_{l, m, n}
    &&
\end{align}
brings the equation and system to the form
\begin{align}
    \label{eq:2ddToda_al}
    &&
    \tau_{l + 1, m + 1, n} \, \tau_{l, m, n}
    &= \alpha^2 \, \tau_{l + 1, m + 1, n - 1} \, \tau_{l, m, n + 1}
    + \beta \, \tau_{l + 1, m, n} \, \tau_{l, m + 1, n},
    &&
\end{align}
\begin{gather}
    \label{eq:2ddToda_scalsys_al}
    \left\{
    \begin{array}{rcl}
         \lambda_1 \, \psi_{l + 1, m, n}
         &=& \lambda_3 \, \psi_{l, m, n + 1} + a_{l, m, n} \, \psi_{l, m, n},
         \\[2mm]
         \lambda_2 \lambda_3 \, \psi_{l, m - 1, n + 1}
         &=& \lambda_3 \, \psi_{l, m, n + 1}
         + b_{l, m, n} \, \psi_{l, m, n},
    \end{array}
    \right.
    \\[3mm]
    \label{eq:2ddToda_abdef_al}
    \begin{gathered}
    a_{l, m, n}
    = \alpha^{-1} \, \theta_{l + 1, m, n} \theta_{l, m, n}^{-1},
    \qquad
    b_{l, m, n}
    = \alpha \, \theta_{l, m - 1, n + 1} \theta_{l, m, n}^{-1}
    .
    \end{gathered}
\end{gather}
\begin{rem}
Regarding $\theta_{l, m, n}$, the map \eqref{eq:scalrule} turns into $\theta_{l, m, n} \mapsto \theta_{l, m, n} \alpha^{l + m}$.
\end{rem}

\begin{rem}
\label{rem:theta_qlmn}
One more useful map is given by the scaling 
\begin{align}
    \label{eq:Toda_change}
    \tau_{l, m, n}
    \mapsto 
    \alpha^{(l + m) \, n} \beta^{l \, m} \, q^{\frac12 k_1 l (l - 1) n} \, 
    q^{\frac12 k_2 m (m - 1) n} \, 
    q^{\frac12 k_3 l \, n \, (l + n)} 
    \, 
    \tau_{l, m, n},
\end{align}
where $\alpha$, $\beta$, $q$, $k_i$ are constants.
It changes \eqref{eq:2ddToda} as
\begin{align}
    \label{eq:2ddToda_albet}
    &&
    \tau_{l + 1, m + 1, n} \, \tau_{l, m, n}
    &= \alpha_{l, m, n} \, \tau_{l + 1, m + 1, n - 1} \, \tau_{l, m, n + 1}
    + \beta \, \tau_{l + 1, m, n} \, \tau_{l, m + 1, n},
    &
    \alpha_{l, m, n}
    &= \alpha^2 \, q^{k_1 l + k_2 m + k_3 n}.
    &&
\end{align}
Here we have an additional parameter $q$, depending on $l$, $m$, and $n$. In other words, it is non-autonomous.
Note that this map can be rewritten in terms of the function $\theta_{l, m, n}$ as 
\begin{align}
    \label{eq:2ddToda_albet_th}
    \theta_{l, m, n}
    \mapsto 
    {\alpha}^{l + m} \, 
    q^{\frac12 k_1 l (l - 1)} \, 
    q^{\frac12 k_2 m (m - 1)} \, 
    q^{\frac12 k_3 l (2 n + l - 1)} 
    \, 
    \theta_{l, m, n}.
\end{align}
\end{rem}

\subsection{Somos-like sequences}
\label{sec:somoslikeseq}
In this subsection we give a brief overview of the Somos-$N$ type equation
\begin{align}
    \label{eq:SomosN_can}
    &&
    x_{M + N} \, x_{M}
    &= \alpha^2 \, x_{M + N - r} \, x_{M + r} 
    + \beta \, x_{M + N - s} \, x_{M + s},
    &&&
    N 
    &\in \mathbb{N}_{>3},
    &
    1 
    &\leq r < s \leq \left[\tfrac{N}{2}\right]
    &&
\end{align}
and its relations to the two-dimensional Toda lattice and discrete \Painleve equations. Here the parameters $\alpha$, $\beta$ are arbitrary constants and, in the initial sequence, they are equal to one. Note that the equation \eqref{eq:SomosN_can} is also known as the two-term version of the Gale–Robinson sequence (see eq. (1.10) in \cite{fomin2002laurent}). 

Regarding a plane-wave reduction, one may bring the 2d discrete Toda equation \eqref{eq:2ddToda_al} to the Somos-$N$ sequence \eqref{eq:SomosN_can}. Indeed, taking the change
\begin{align}
    &&
    x_M
    &:= \tau_{l, m, n},
    &
    \varphi_{M}
    &:= \psi_{l, m, n},
    &
    M 
    &= (N - s) \, l + s \, m + r \, n,
    &&
\end{align}
the 2d discrete Toda \eqref{eq:2ddToda_al} turns into \eqref{eq:SomosN_can}, while the Lax pair \eqref{eq:2ddToda_scalsys_al} with $\lambda_1 = 1$, $\lambda_3 = \mu$, and $\lambda_2 = \lambda \, \mu^{-1}$ becomes (cf. with Proposition 2.1 in \cite{hone2017reductions})
\begin{gather}
    \label{eq:SomosN_scalsys}
    \left\{
    \begin{array}{rcl}
         \varphi_{M + N - s}
         &=& \mu \, \varphi_{M + r} 
         + a_{M} \, \varphi_{M},
         \\[2mm]
         \lambda \, \varphi_{M}
         &=& \mu \, \varphi_{M + s}
         + b_{M} \, \varphi_{M + s - r},
    \end{array}
    \right.
    \\[3mm]
    \begin{aligned}
        a_{M}
        &= \alpha^{-1} \,\, 
        x_{M + N - s} x_{M + N - s - r}^{-1} \, 
        x_{M}^{-1}
        x_{M - r}
        ,
        &&&&&
        b_{M}
        &= \alpha \,\, 
        x_M x_{M - r}^{-1} \, 
        x_{M + s - r}^{-1}
        x_{M + s - 2 r}
        .
    \end{aligned}
\end{gather}
Here $\lambda$ and $\mu$ are spectral parameters. Note that one can introduce $y_M = x_{M} x_{M - r}^{-1}$ and then $\theta_{l, m, n} \mapsto y_M$. The system \eqref{eq:SomosN_scalsys} can be written in the matrix form
\begin{align}
    \left\{
    \begin{array}{rcl}
         \mathcal{L}_M \, \Phi_M
         &=& \lambda \, \Phi_M,
         \\[2mm]
         \Phi_{M + 1}
         &=& \mathcal{M}_M \, \Phi_M,
    \end{array}
    \right.
\end{align}
where $\mathcal{L}_M = \mathcal{L}_M (\mu)$ and $\mathcal{M}_M = \mathcal{M}_M(\mu)$, by using the vector-function $\Phi_M = \begin{pmatrix} \varphi_{M + N - s - 1} & \dots & \varphi_{M + 1} & \varphi_M \end{pmatrix}$. Its compatibility condition 
\begin{align}
    \mathcal{L}_{M + 1} \, \mathcal{M}_M
    &= \mathcal{M}_M \, \mathcal{L}_M
\end{align}
gives rise to
\begin{align}
    \brackets{
    x_{M + N} x_M - \alpha^2 x_{M + N - r} x_{M + r}
    } \, x_{M + N - s}^{-1} x_{M + s}^{-1}
    &= \brackets{
    x_{M + N - r} x_{M - r} - \alpha^2 x_{M + N - 2 r} x_{M}
    } \, x_{M + N - s - r}^{-1} x_{M + s - r}^{-1}.
\end{align}
Here we do not present an explicit form of the matrices $\mathcal{L}_M$, $\mathcal{M}_M$. This fact follows from the reduced compatibility condition of system \eqref{eq:2ddToda_scalsys_al}. 

Similarly to the 2d discrete Toda case, the l.h.s. is differ from the r.h.s. by a shift of $M$. So, we may consider equation
\begin{align}
    &&
    x_{M + N} x_M
    &= \alpha^2 x_{M + N - r} x_{M + r} 
    + \beta_M \, x_{M + N - s} x_{M + s},
    &
    \beta_{M + r}
    &= \beta_M,
    &&
\end{align}
which coincides with \eqref{eq:SomosN_can} if $\beta_M $ is an autonomous constant.

\medskip
In \cite{hone2014discrete}, A. Hone and R. Inoue  studied a connection between Somos-$N$ type sequences and discrete \Painleve equations.
It turns out that the Somos-$4$ and Somos-$5$ sequences are connected with $q$PI and $q$PII equations, respectively. Moreover, one can derive a hierarchy of discrete equations containing either  $q$PI or $q$PII as the first member of the hierarchy. We briefly discuss these results below.

Consider the Somos-4 equation with $\alpha = \beta = 1$
\begin{gather}
    \label{eq:Somos4_can}
    \begin{aligned}
    x_{M + 4} \, x_{M}
    &= x_{M + 3} \, x_{M + 1} 
    + x_{M + 2}^2.
    \end{aligned}
\end{gather}
Regarding the so-called $Y$-variables
$y_M = x_{M + 2} \,\, x_{M + 1}^{-2} \,\, x_M$,
\eqref{eq:Somos4_can} becomes
\begin{align}
    \label{eq:qPI_0}
    y_{M + 2} \,\, y_{M + 1}^2 \,\, y_{M}
    &= 1 + y_{M + 1}.
\end{align}
Note that the quantity
$ \mathcal{Z}_M
    := y_{M + 2} \,\, y_{M + 1}^{2} \,\, y_{M} 
    \brackets{
    1 + y_{M + 1}
    }^{-1}$
satisfies the condition
\begin{align}
    \mathcal{Z}_{M + 2} \,\,
    \mathcal{Z}_{M + 1}^{-2} \,\,
    \mathcal{Z}_M 
    &= 1.
\end{align}
Taking logarithms of the latter, it turns into a linear difference equation whose solution is given by $\log \mathcal{Z}_M = \alpha \, M + \beta$ and one gets $\mathcal{Z}_M = \beta \, q^{M}$ for $\beta$, $q \in \mathbb{C}^*$. Therefore, the general solution of the Somos-$4$ equation \eqref{eq:qPI_0} is given by a solution of the second-order recurrence
\begin{align}
    \label{eq:qPI_1}
    y_{M + 2} \,\, y_{M + 1}^2 \,\, y_{M}
    &= \beta \, q^{M} \, \, \brackets{1 + y_{M + 1}}.
\end{align}
The gauge transformation 
$y_M = \beta^{-1} \alpha_M \, u_M$ leads to the equation
\begin{align}
    &&&&
    u_{M + 2} \, u_{M + 1}^2 \, u_{M}
    &= \alpha_{M + 1} \, u_{M + 1}
    + \beta,
    &
    \alpha_{M + 2} \, \alpha_{M + 1}^2 \, \alpha_{M}
    &= \beta^4 q^M.
    &&&&
\end{align}
A solution of the recurrence $\alpha_{M + 2} \, \alpha_{M + 1}^2 \, \alpha_{M} = \beta^4 q^M$ can be written as $\alpha_M = \pm \alpha \, q^M$. After the shift $M \mapsto M + 1$, we arrive at the canonical form of the $q$PI equation
\begin{align}
    \label{eq:qPI}
    u_{M + 1} \, u_{M}^2 \, u_{M - 1}
    &= \alpha \, q^{M} \, u_{M} + \beta
    .
\end{align}
In the paper \cite{sakai1914problem}, this equation is labeled as $q\text{-P}(A_7)$ (see eq. (2.44) therein). 

Similarly, one can reduce the Somos-$5$ equation
\begin{gather}
    \label{eq:Somos5_can}
    \begin{aligned}
    x_{M + 5} \, x_{M}
    &= x_{M + 4} \, x_{M + 1} 
    + x_{M + 3} \, x_{M + 2}
    \end{aligned}
\end{gather}
to the $q$PII equation of the form
\begin{align}
    \label{eq:qPII}
    u_{M + 1} \, u_{M} \, u_{M - 1}
    &= \alpha \, q^M \, u_{M} + \beta,
\end{align}
due to the existence of the function $\mathcal{Z}_M$ satisfying the condition 
\begin{align}
    \mathcal{Z}_{M + 3} \,\,
    \mathcal{Z}_{M + 2}^{-1} \mathcal{Z}_{M + 1}^{-1} \,\,
    \mathcal{Z}_M 
    &= 1.
\end{align} 
The equation \eqref{eq:qPII} can be rewritten in the asymmetric form (eq. (6.4) in \cite{ramani1996discrete} or eq. (26) in \cite{kruskal2000asymmetric})
\begin{align}
    \label{eq:asymqPII}
    \brackets{
    u_{M + 1} u_{M}
    - 1
    } \, \brackets{
    u_{M} u_{M - 1} 
    - 1
    }
    &= \frac{\alpha q^{-M} \, u_{M}}{\beta q^M \, u_{M} + \gamma}
\end{align}
then it coincides with $q\text{-P}(A_6)$ from \cite{sakai1914problem} (see eq. (2.43) therein).

Considering more general Somos-type sequences \eqref{eq:SomosN_can} with $r = 1$ and $s = 2$ and $\alpha = \beta = 1$, one is able to obtain a higher-order generalization of $q$PI and $q$PII equations. For this purpose, let us look at the family of bilinear systems
\begin{align}
    \label{eq:Tz_sys}
    x_{M + N} \, x_{M}
    &= \mathcal{Z}_M \, \brackets{
    x_{M + N - 1} \, x_{N + 1}
    + x_{M + N - 2} \, x_{M + 2}
    },
    &
    \mathcal{Z}_{M + N - 2} \, 
    \mathcal{Z}_{M + N - 3}^{-1} \mathcal{Z}_{M + 1}^{-1} \,
    \mathcal{Z}_M
    &= 1.
\end{align}
Note that it is always possible to make a gauge transformation to move the non-autonomous coefficient entirely into either the first or the second term in the right-hand side of \eqref{eq:Tz_sys}. Note also that a solution of the $\mathcal{Z}$-system in \eqref{eq:Tz_sys} is given by
\begin{align}
    &&
    \mathcal{Z}_M
    &= \beta_M \, q^M,
    &
    \beta_{M + N - 3}
    &= \beta_M.
    &&
\end{align}
Then, when $N$ is even and $N \geq 4$,  substitution $u_M = x_{M + 2} \, x_{M + 1}^{-2} \, x_{M}$ gives the $q$PI hierarchy
\begin{align}
    \label{eq:qPIn_}
    u_{M + N - 2} \, \brackets{
    \prod_{j = 1}^{N - 3} u_{M + j}
    }^2 \, u_{M}
    &= \mathcal{Z}_M \, \brackets{
    1 + \prod_{k = 1}^{N - 3} u_{M + k}
    },
\end{align}
while for odd $N \geq 5$, the change $u_{M} = x_{M + 3} \, x_{M + 2}^{-1} x_{M + 1}^{-1} \, x_{M}$  yields the $q$PII hierarchy
\begin{align}
    \label{eq:qPIIn_}
    \prod_{j = 0}^{N - 3} u_{M + j}
    &= \mathcal{Z}_M \, \brackets{
    1 + \prod_{k = 1}^{\frac12 (N - 5)} u_{M + 2 k + 1}
    }.
\end{align}

In Section \ref{sec:qPainleve}, we construct non-abelian analogs of the equations discussed above. We suggest another method of introducing non-autonomous constants in Somos-$N$ sequences. It is based on Remark \ref{rem:theta_qlmn}.

\begin{rem}
\label{rem:spSomos7}
Let $N = 7$, $r = 1$, $s = 3$ and $\alpha = \beta = 1$ in \eqref{eq:SomosN_can}. It can be written as
\begin{align}
    x_{M + 7} \, x_{M + 4}^{-1} x_{M + 3}^{-1} \, x_{M}
    &= x_{M + 6} \, x_{M + 4}^{-1} x_{M + 3}^{-1} \, x_{M + 1}
    + 1.
\end{align}
By introducing new variable \cite{hone2014discrete}
\begin{align}
    u_{M}
    &= x_{M + 6} \, x_{M + 4}^{-1} x_{M + 3}^{-1} \, x_{M + 1},
\end{align}
this special case of the Somos-$7$ equation becomes the Gauss recurrence relation of period $5$ (see, e.g., \cite{iwao2012bilinear}):
\begin{align}
    \label{eq:GRP5_can}
    u_{M + 1} \, u_{M - 1}
    &= u_{M} + 1.
\end{align}
\end{rem}

\subsection{Non-Abelian setting}
\label{sec:ncset}

Here we introduce main definitions on which the subsequent sections are based.

Let $R$ be an associative unital division ring over a field $F$ s.t. $\cchar F = 0$. Note that $R$ is a generalization of the ring of rational functions over the field $\mathbb{C}$ to a non-abelian case. 
For the brevity, we identify the unit of the field with the unit of the ring, hoping that this will not lead to misunderstanding. Note that elements of $R$ do not necessarily commute with each other, while the field $F$ belongs to the center $\mathcal{Z}(R)$ of $R$. Below we will often refer to the elements of the ring $R$ as to \textit{functions}.
Note that on $R$ we have the following involution. 

\begin{defn}
\label{def:tau}
An $F$-linear map 
$\tau: R \to R$ 
defined by 
\begin{itemize}
    \item[(a)] 
    $
    \begin{aligned}
    \tau(\alpha)
    &= \alpha,
    &&
    \text{for any}
    &
    \alpha
    &\in F,
    \end{aligned}
    $
    
    \item[(b)] 
    $
    \begin{aligned}
    \tau(f)
    &= f,
    &
    &&
    \text{for any generator $f$ of }
    R,
    \end{aligned}
    $
    
    \item[(c)]
    $
    \begin{aligned}
    \tau(P \, Q)
    &= \tau(Q) \, \tau(P), 
    &&
    \text{for any}
    &
    P, \, Q
    &\in R
    \end{aligned}
    $
\end{itemize}
is called {\it a transposition}.
\end{defn}

We remark that this involution can be extended to the matrices $X = \brackets{x_{i j}}$ over $R$ as follows: 
\begin{align}
    \label{tau_mat}
    \tau(X)
    &= \brackets{
    \tau(x_{j i})
    }.
\end{align}

In addition to the transposition, let us introduce a derivation of $R$.
\begin{defn}
\label{def:der}
Consider an $F$-linear map $\partial : R \to R$ satisfying the properties
\begin{enumerate}
    \item[(a)] 
    $\partial (\alpha) = 0$ for any $\alpha \in F$,
        
    \item[(b)]
    $\partial \left(F \, G\right) 
    = \partial
    \left(F\right) \,\, G 
    + F \,\, \partial \left(G\right)$ 
    for any $F$, $G \in R$,
\end{enumerate}
Such a map we call \textit{a derivation} of the ring $R$ and set $\partial (F) = F'$.
\end{defn}

\begin{rem}
\label{rem:tr_der}
Note that the transposition and the derivation commute with each other:
\begin{align}
    \tau\brackets{
    \partial (F)
    }
    &= \partial \big(
    \tau(F)
    \big)
    .
\end{align}
\end{rem}

One can extend this construction to imitate the case of functions depending on several variables.  
Namely, if $\partial_x$, $\partial_y$ are two commuting derivatives satisfying Definition \ref{def:der} such that there exist elements $x$, $y$ of $\mathcal{Z}(R)$ which $\partial_x(x) = \partial_y(y) = 1$, $\partial_x(y) = \partial_y(x) = 0$. In~this case, we use the notation $\partial_x(F) = F_x$, $\partial_y(F) = F_y$, and so on.

Below we will use non-abelian functions depending on the indices $(l, m, n)$ or variables $x$, $y$. Namely, we remark that functions $f_{l, m, n} \in R$ are dynamical variables at site $(l, m, n) \in \mathbb{Z}^3$, while the function $g_n(x, y) \in R$ is at site $n \in \mathbb{Z}$ and depends on abelian variables $x$, $y \in \mathcal{Z}(R)$. A good example for the function $g_n(x, y)$ that is helpful to keep in mind is matrix-valued functions on a lattice. Note that such a function can be expanded into a Taylor series. We will use this fact for continuous limits. 

\subsection{Quasideterminants}
\label{sec:quasidet}

Following to \cite{gel1991determinants}, \cite{gelfand2005quasideterminants}, we recall the definition of quasideterminants and several their properties. Consider an $n \times n$ matrix $X = (x_{i j})$ whose entries are deﬁned in $R$. Consider also the block decomposition of the matrix $X$ and its inverse $X^{-1} = Y$
\begin{align}
    &&
    X
    &= 
    \begin{blockarray}{ccc}
      J_1 & J_2 &  \\
    \begin{block}{(cc)c}
      X_{I_1 J_1} & X_{I_1 J_2} & I_1  \\
      X_{I_2 J_1} & X_{I_2 J_2} & I_2  \\
    \end{block}
    \end{blockarray},
    &
    X^{-1}
    &= Y = 
    \begin{blockarray}{ccc}
      I_1 & I_2 &  \\
    \begin{block}{(cc)c}
      Y_{J_1 I_1} & Y_{J_1 I_2} & J_1  \\
      Y_{J_2 I_1} & Y_{J_2 I_2} & J_2  \\
    \end{block}
    \end{blockarray},
    &
    |I_1|
    &= |J_1|,
    &
    |I_2|
    &= |J_2|.
    &&
\end{align}
Then, by the block multiplication, we get
\begin{align}
    Y_{J_1 I_1}
    &= \brackets{
    X_{I_1 J_1} 
    - X_{I_1 J_2} \, X_{I_2 J_2}^{-1} \, X_{I_2 J_1}
    }^{-1},
    &
    Y_{J_1 I_2}
    &= - X_{I_1 J_1}^{-1} \, X_{I_1 J_2} \,
    \brackets{
    X_{I_2 J_2} 
    - X_{I_2 J_1} \, X_{I_1 J_1}^{-1} \, X_{I_1 J_2}
    }^{-1},
    \\[1mm]
    Y_{J_2 I_2}
    &= \brackets{
    X_{I_2 J_2} 
    - X_{I_2 J_1} \, X_{I_1 J_1}^{-1} \, X_{I_1 J_2}
    }^{-1},
    &
    Y_{J_2 I_1}
    &= - X_{I_2 J_2}^{-1} \, X_{I_2 J_1} \,
    \brackets{
    X_{I_1 J_1} 
    - X_{I_1 J_2} \, X_{I_2 J_2}^{-1} \, X_{I_2 J_1}
    }^{-1}.
\end{align}
In particular, setting $I_1 = {i}$, $J_1 = {j}$, we obtain that the entries of the inverse of $X$ are given by the recursive expression
\begin{align}
    y_{ji}
    &= \brackets{
    x_{ij}
    - X_{i J_2} \, (X_{I_2 J_2})^{-1} \, X_{I_2 j}
    }^{-1},
    &
    I_2
    &= \{ 1, \dots, n\} \setminus \{ i \},
    &
    J_2
    &= \{ 1, \dots, n\} \setminus \{ j \}.
\end{align}

Let $r_i^j$ represent the $i$-th row of $X$ with the $j$-th element removed, $c_j^i$ be the $j$-th column of $X$ with the $i$-th element removed, and $X^{ij}$ be the submatrix obtained by removing the $i$-th row and the $j$-th column from~$X$. Then there are $n^2$ quasideterminants, denoted as $|X|_{ij}$ for $i, j = 1, 2, \dots , n$, if all of its inverse $(X^{ij})^{-1}$ exist.
\begin{defn} 
\textit{A quasideterminant} $|X|_{ij}$ is defined recursively as
\begin{align}
    |X|_{ij}
    &= x_{i j}
    - r_{i}^j \, (X^{ij})^{-1} \, c_j^i
    = x_{ij}
    - \sum_{k \neq i, l \neq j} 
    x_{i l} \brackets{(X^{ij})^{-1}}_{l k} x_{k j}
    .
\end{align}
\end{defn}
Sometimes we will denote the quasideterminant by
\begin{align}
    |X|_{ij}
    &= 
    \begin{vmatrix}
    x_{11} & \dots & x_{1 j} & \dots & x_{1n}
    \\
    \vdots & & \vdots & & \vdots 
    \\
    x_{i1} & \dots & \boxed{x_{i j}} & \dots & x_{in}
    \\
    \vdots & & \vdots & & \vdots 
    \\
    x_{n1} & \dots & x_{n j} & \dots & x_{nn}
    \end{vmatrix}
    .
\end{align}
Since the quasideterminants of the matrix $X$ are the inverses of the entries of matrix $Y = X^{-1}$:
\begin{align}
    &&
    y_{ij}^{-1}
    &= |X|_{ji},
    &
    i, j
    &= 1, \dots, n,
    &&
\end{align}
the quasideterminant $|X|_{ij}$ can be rewritten as
\begin{align}
    |X|_{ij}
    &= x_{ij}
    - \sum_{k \neq i, l \neq j}
    x_{i l} \,
    \brackets{|X^{ij}|_{kl}}^{-1} \,
    x_{k j}.
\end{align}

Note that in the case of a commutative ring $R$, the quasideterminant is simply given by
\begin{align}
    \label{eq:quasidet_det}
    &&
    &&
    |X|_{ij}
    &= (-1)^{i + j} \frac{\det X}{\det X^{ij}}, 
    &&&
    \text{for any}
    &&
    i, j
    &= 1, 2, \dots, n.
    &&
    &&
\end{align}
Thus, the quasideterminants may be regarded as non-commutative analogs of the ratio of a determinant to one of its principal minors.

\begin{exmp}
\phantom{}

\begin{itemize}
    \item Consider a $2 \times 2$-matrix $X$ with generic entries $x_{ij}$, $i, j = 1, 2$. Then it has four quasideterminants:
    \begin{align}
        |X|_{11}
        &= x_{11}
        - x_{12} \, x_{22}^{-1} \, x_{21}
        ,
        &
        |X|_{12}
        &= x_{12}
        - x_{11} \, x_{21}^{-1} \, x_{22}
        ,
        \\[1mm]
        |X|_{21}
        &= x_{21}
        - x_{22} \, x_{12}^{-1} \, x_{11}
        ,
        &
        |X|_{22}
        &= x_{22}
        - x_{21} \, x_{11}^{-1} \, x_{12}
        .
    \end{align}
    
    \item In the case of $3 \times 3$-matrix $X$ with generic entries $x_{ij}$, $i, j = 1, 2, 3$, there are nine quasideterminants. For example,
    \begin{align}
        |X|_{11}
        = x_{11}
        &- x_{12} \, \brackets{
        x_{22} 
        - x_{23} \, x_{33}^{-1} \, x_{32}
        }^{-1} x_{21}
        - x_{12} \, \brackets{
        x_{32}
        - x_{33} \, x_{23}^{-1} \, x_{22}
        }^{-1} \, x_{31}
        \\[1mm]
        &- \, x_{13} \brackets{
        x_{23}
        - x_{22} \, x_{32}^{-1} \, x_{33}
        }^{-1} \, x_{21}
        - x_{13} \, \brackets{
        x_{33} 
        - x_{32} \, x_{22}^{-1} \, x_{23}
        }^{-1} \, x_{31}.
    \end{align}
\end{itemize}
\end{exmp}

Below we list several important properties of quasideterminants.
\begin{prop}
A quasideterminant is invariant under a permutation of the rows or columns.
\end{prop}
\begin{exmp}
\begin{align}
    \begin{vmatrix}
        x_{11} & x_{12} & x_{13} \\
        x_{21} & x_{22} & x_{23} \\
        x_{31} & \boxed{x_{32}} & x_{33}
    \end{vmatrix}
    &= 
    \begin{vmatrix}
        x_{21} & x_{22} & x_{23} \\
        x_{11} & x_{12} & x_{13} \\
        x_{31} & \boxed{x_{32}} & x_{33}
    \end{vmatrix}
    = 
    \begin{vmatrix}
        x_{12} & x_{11} & x_{13} \\
        x_{22} & x_{21} & x_{23} \\
        \boxed{x_{32}} & x_{31} & x_{33}
    \end{vmatrix}
    .
\end{align}
\end{exmp}

\begin{prop}
\phantom{}
\begin{enumerate}
        \item[\rm{(a)}]
        Let the matrix $B$ be obtained from $A$ by multiplying the $i$-th row on the left by $\lambda$, then
        \begin{align}
            |B|_{k j}
            &= \left\{
            \begin{array}{cc}
                 \lambda |A|_{ij},
                 &  k = i,
                 \\[1mm]
                 |A|_{kj},
                 &  k \neq i.
            \end{array}
            \right.
        \end{align}
        
        \item[\rm{(b)}]
        Let the matrix $C$ be obtained from $A$ by multiplying the $j$-th column on the right by $\mu$, then
        \begin{align}
            |C|_{i l}
            &= \left\{
            \begin{array}{cc}
                |A|_{ij} \mu,
                 &  l = j,
                 \\[1mm]
                 |A|_{il},
                 &  l \neq j.
            \end{array}
            \right.
        \end{align}
        
        \item[\rm{(c)}]
        Let the matrix $D$ be obtained from $A$ by adding to some row {\rm(}resp. column{\rm)} of $A$ its $k$-th row {\rm(}reps. column{\rm)}, then for any $i \neq k$ {\rm(}resp. $j \neq k${\rm)}
        \begin{align}
            |D|_{ij}
            &= |A|_{ij}.
        \end{align}
\end{enumerate}
\end{prop}

\begin{prop}
The quasiminors of the matrix $X$ are related by the homological relations
\begin{align} 
    \label{eq:homrel_col}
    |X|_{ij} \, \brackets{|X^{il}|_{kj}}^{-1}
    &= - |X|_{il} \, \brackets{|X^{ij}|_{kl}}^{-1},
    \\[2mm]
    \label{eq:homrel_row}
    \brackets{|X^{kj}|_{il}}^{-1} \, |X|_{ij}
    &= - \brackets{|X^{ij}|_{kl}}^{-1} \, |X|_{kj}.
\end{align}
\end{prop}

\begin{exmp}
\begin{align}
    \begin{vmatrix}
        x_{11} & x_{12} \\
        x_{31} & \boxed{x_{32}}
    \end{vmatrix}^{-1} \,
    \begin{vmatrix}
        x_{11} & x_{12} & x_{13} \\
        x_{21} & x_{22} & x_{23} \\
        x_{31} & x_{32} & \boxed{x_{33}}
    \end{vmatrix}
    &= -
    \begin{vmatrix}
        x_{11} & x_{12} \\
        x_{21} & \boxed{x_{22}}
    \end{vmatrix}^{-1} \,
    \begin{vmatrix}
        x_{11} & x_{12} & x_{13} \\
        x_{21} & x_{22} & \boxed{x_{23}} \\
        x_{31} & x_{32} & x_{33}
    \end{vmatrix}
    .
\end{align}
\end{exmp}

\begin{prop}
Consider the following $(n - 1) \times (n - 1)$-submatrices of $X${\rm:} 
\begin{itemize}
    \item 
    the matrix $A$ obtained from $X$ by deleting its $n$-th row and $n$-th column{\rm;}
    
    \item 
    the matrix $B$ obtained from $X$ by deleting its $1$-st row and $n$-th column{\rm;}
    
    \item 
    the matrix $C$ obtained from $X$ by deleting its $n$-th row and $1$-st column{\rm;}
    
    \item 
    the matrix $D$ obtained from $X$ by deleting its $1$-st row and $1$-st column{\rm.}
\end{itemize}
Then, a non-commutative analog of the determinant Jacobi identity holds
\begin{align}
    \label{eq:Jacobi_nc}
    |X|_{nn}
    &= |D|_{n - 1, n - 1}
    - |B|_{n - 1, 1} \, \, |A|_{1, 1}^{-1} \, \, |C|_{1, n - 1}.
\end{align}
\end{prop}

\begin{rem}
    Using homological relations from \cite{gel1991determinants}, \eqref{eq:Jacobi_nc} can be rewritten in the form
    \begin{align}
    \label{eq:Jacobi_nc_}
    |X|_{nn}
    &= |D|_{n - 1, n - 1}
    - |B|_{n - 1, 1} \, \, |A|_{n - 1, 1}^{-1} \, \, |C|_{n - 1, n - 1}.
    \end{align}
\end{rem}

Schematically, \eqref{eq:Jacobi_nc} can be presented as
\begin{figure}[H]
    \centering
    \scalebox{1.}{\tikzset{every picture/.style={line width=0.75pt}} 

\begin{tikzpicture}[x=0.75pt,y=0.75pt,yscale=-1,xscale=1]

\draw    (50,20.5) -- (50,89.86) ;
\draw    (120.57,20.5) -- (120.57,89.86) ;

\draw    (149.86,20.5) -- (149.86,89.86) ;
\draw    (220.43,20.5) -- (220.43,89.86) ;

\draw    (249.86,20.5) -- (249.86,89.86) ;
\draw    (320.43,20.5) -- (320.43,89.86) ;

\draw    (340.29,20.5) -- (340.29,89.86) ;
\draw    (410.86,20.5) -- (410.86,89.86) ;

\draw    (429.29,20.5) -- (429.29,89.86) ;
\draw    (499.86,20.5) -- (499.86,89.86) ;

\draw    (155.71,33.4) -- (216,33.4) ;
\draw    (164.29,25.47) -- (164.43,87.04) ;
\draw    (254.71,33.4) -- (315,33.4) ;
\draw    (306.06,24.5) -- (306.21,86.07) ;
\draw    (344.94,76.36) -- (405.22,76.36) ;
\draw    (396.29,24.5) -- (396.43,86.07) ;
\draw    (435.16,76.36) -- (495.44,76.36) ;
\draw    (443.84,24.5) -- (443.98,86.07) ;
\draw   (100,69.5) -- (115.5,69.5) -- (115.5,85) -- (100,85) -- cycle ;
\draw   (199.75,69.5) -- (215.25,69.5) -- (215.25,85) -- (199.75,85) -- cycle ;
\draw   (255.5,69.5) -- (271,69.5) -- (271,85) -- (255.5,85) -- cycle ;
\draw   (346,25) -- (361.5,25) -- (361.5,40.5) -- (346,40.5) -- cycle ;
\draw   (479,25) -- (494.5,25) -- (494.5,40.5) -- (479,40.5) -- cycle ;

\draw (128,51.4) node [anchor=north west][inner sep=0.75pt]    {$=$};
\draw (227.5,48) node [anchor=north west][inner sep=0.75pt]    {$-$};
\draw (326.5,50) node [anchor=north west][inner sep=0.75pt]    {$\cdot $};
\draw (416.5,50) node [anchor=north west][inner sep=0.75pt]    {$\cdot $};
\draw (411.43,15.54) node [anchor=north west][inner sep=0.75pt]    {$^{-1}$};

\end{tikzpicture}}.
\end{figure}
Recall that the classical Jacobi determinant identity reads
\begin{figure}[H]
    \centering
    \scalebox{1.}{\tikzset{every picture/.style={line width=0.75pt}} 

\begin{tikzpicture}[x=0.75pt,y=0.75pt,yscale=-1,xscale=1]

\draw    (150,20.5) -- (150,89.86) ;
\draw    (220.57,20.5) -- (220.57,89.86) ;

\draw    (249.86,20.5) -- (249.86,89.86) ;
\draw    (320.43,20.5) -- (320.43,89.86) ;

\draw    (255.71,33.4) -- (316,33.4) ;
\draw    (264.29,25.47) -- (264.43,87.04) ;
\draw    (59,20.5) -- (59,89.86) ;
\draw    (129.57,20.5) -- (129.57,89.86) ;

\draw    (155.29,33.4) -- (215.57,33.4) ;
\draw    (163.86,25.5) -- (164,87.07) ;
\draw    (155.51,77.9) -- (215.79,77.9) ;
\draw    (206.86,25.93) -- (207,87.5) ;
\draw    (340.43,20.5) -- (340.43,89.86) ;
\draw    (411,20.5) -- (411,89.86) ;

\draw    (345.08,77.86) -- (405.37,77.86) ;
\draw    (396.43,24.5) -- (396.57,86.07) ;
\draw    (440.71,20.47) -- (440.71,89.83) ;
\draw    (511.29,20.47) -- (511.29,89.83) ;

\draw    (529.71,20.47) -- (529.71,89.83) ;
\draw    (600.29,20.47) -- (600.29,89.83) ;

\draw    (445.37,77.9) -- (505.65,77.9) ;
\draw    (535.59,33.4) -- (595.87,33.4) ;
\draw    (586.27,24.47) -- (586.41,86.04) ;
\draw    (453.52,24.5) -- (453.66,86.07) ;

\draw (228,51.4) node [anchor=north west][inner sep=0.75pt]    {$=$};
\draw (418,48.4) node [anchor=north west][inner sep=0.75pt]    {$-$};
\draw (136.07,50) node [anchor=north west][inner sep=0.75pt]    {$\cdot $};
\draw (326.71,50) node [anchor=north west][inner sep=0.75pt]    {$\cdot $};
\draw (517.29,50) node [anchor=north west][inner sep=0.75pt]    {$\cdot $};

\end{tikzpicture}}.
\end{figure}

\section{Non-Abelian 2d Toda equations}
\label{sec:main11}

\subsection{Discrete case}
\label{sec:2ddToda}

By using a non-abelian analog of the determinant Jacobi identity, we will define the non-commutative analog of the 2d discrete Toda lattice. Consider the matrix
\begin{align}
    \label{eq:Theta_def}
    \Theta_{l, m, n}
    &:= 
    \begin{pmatrix}
    \varphi_{l, m} 
    & \dots & \varphi_{l, m + n - 1}
    \\
    \vdots & \ddots & \vdots
    \\
    \varphi_{l + n - 1, m} 
    & \dots & 
    \varphi_{l + n - 1, m + n - 1}
    \end{pmatrix}
    = \brackets{
    \varphi_{l + i - 1, m + j - 1}
    }_{1 \leq i, j \leq n},
    &
    (l, m, n + 1)
    &\in \mathbb{Z}_{\geq 0}^{3},
\end{align}
where $\varphi_{i, j} \in R$. Set
\begin{align}
    \label{eq:theta_def}
    \theta_{l, m, n}
    &= |\Theta_{l, m, n}|_{nn}.
\end{align}
We put $\Theta_{l, m, 0}^{-1}$. Let us also use notation $T_{\pm i}$, $i = 1, 2, 3$ for the shift operators
\begin{align}
    &&
    T_{\pm 1} \brackets{\theta_{l, m, n}}
    &= \theta_{l \pm 1, m, n},
    &
    T_{\pm 2} \brackets{\theta_{l, m, n}}
    &= \theta_{l, m \pm 1, n},
    &
    T_{\pm 3} \brackets{\theta_{l, m, n}}
    &= \theta_{l, m, n \pm 1}.
    &&
\end{align}

\begin{prop}
\label{thm:2ddToda_nc}
Quasideterminant $\theta_{l, m, n}$ defined by \eqref{eq:theta_def} satisfies the non-abelian analog of the two-dimensional discrete Toda system
\begin{align}
    \label{eq:2ddToda_nc}
    \theta_{l + 1, m + 1, n}
    &= \theta_{l, m, n + 1}
    + \theta_{l + 1, m, n} \brackets{
    \theta_{l, m, n}^{-1}
    - \theta_{l + 1, m + 1, n - 1}^{-1}
    } \theta_{l, m + 1, n}.
\end{align}
\end{prop}

\begin{rem}
Making the change \eqref{eq:2ddTodatoHirota} in \eqref{eq:2ddToda_nc}, one gets a non-abelian analog of the Hirota equation \eqref{eq:Hirota}:
\begin{align}
    \label{eq:Hirota_nc}
    \theta_{l + 1, m, n + 1}
    &= \theta_{l, m + 1, n + 1}
    + \theta_{l + 1, m + 1, n} \brackets{
    \theta_{l, m + 1, n}^{-1}
    - \theta_{l + 1, m, n}^{-1}
    } \theta_{l, m, n + 1}.
\end{align}
\end{rem}

\begin{proof}
Before moving on to an arbitrary $n$, let us first consider the case of $n = 2$ to demonstrate key points in the proof.

\textbullet \,\,
Set $n = 2$ in \eqref{eq:2ddToda_nc}. Then, $\theta_{l, m, 3}$, $\theta_{l, m, 2}$, and $\theta_{l, m, 1}$ defined by \eqref{eq:theta_def} are solutions of \eqref{eq:2ddToda_nc}. 
Indeed, from the non-commutative Jacobi identity \eqref{eq:Jacobi_nc} it follows that
\begin{multline}
    \begin{vmatrix}
    \varphi_{l, m} 
    & \varphi_{l, m + 1}  
    & \varphi_{l, m + 2}
    \\[1mm]
    \varphi_{l + 1, m} 
    & \varphi_{l + 1, m + 1}  
    & \varphi_{l + 1, m + 2}
    \\[1mm]
    \varphi_{l + 2, m} 
    & \varphi_{l + 2, m + 1}  
    & \boxed{\varphi_{l + 2, m + 2}}
    \end{vmatrix}
    = 
    \begin{vmatrix}
    \varphi_{l + 1, m + 1}  
    & \varphi_{l + 1, m + 2}
    \\[1mm]
    \varphi_{l + 2, m + 1}  
    & \boxed{\varphi_{l + 2, m + 2}}
    \end{vmatrix}
    - 
    \begin{vmatrix}
    \varphi_{l + 1, m}  
    & \varphi_{l + 1, m + 1}
    \\[1mm]
    \boxed{\varphi_{l + 2, m}}
    & {\varphi_{l + 2, m + 1}}
    \end{vmatrix}
    \\[2mm]
    \, \cdot
    \begin{vmatrix}
    \varphi_{l, m}  
    & \varphi_{l, m + 1}
    \\[1mm]
    \boxed{\varphi_{l + 1, m}}
    & {\varphi_{l + 1, m + 1}}
    \end{vmatrix}^{-1}
    \,
    \begin{vmatrix}
    \varphi_{l, m + 1}  
    & \varphi_{l, m + 2}
    \\[1mm]
    \varphi_{l + 1, m + 1}  
    & \boxed{\varphi_{l + 1, m + 2}}
    \end{vmatrix}.
\end{multline}
In order to rewrite this relation in terms of $\theta_{l,m,n}$ only, we need to move all boxes to the bottom right corner. Note that
\begin{align}
    \begin{vmatrix}
    x_{11}  
    & x_{12} \,\,\, 
    \\[1mm]
    \boxed{x_{21}}
    & {x_{22}} \,\,\,
    \end{vmatrix}
    = x_{21} - x_{22} \, x_{12}^{-1} \, x_{11}
    = - \brackets{
    x_{22} - x_{21} \, x_{11}^{-1} \, x_{12}
    } x_{12}^{-1} \, x_{11}
    = -
    \begin{vmatrix}
    \,\,
    x_{11}  
    & x_{12}
    \\[1mm]
    \,\,
    {x_{21}}
    & \boxed{x_{22}}
    \end{vmatrix}
    \, x_{12}^{-1} \, x_{11}
    .
\end{align}
Then
\begin{multline}
    \begin{vmatrix}
    \varphi_{l + 1, m}  
    & \varphi_{l + 1, m + 1}
    \\[1mm]
    \boxed{\varphi_{l + 2, m}}
    & {\varphi_{l + 2, m + 1}}
    \end{vmatrix}
    \, \cdot \,
    \begin{vmatrix}
    \varphi_{l, m}  
    & \varphi_{l, m + 1}
    \\[1mm]
    \boxed{\varphi_{l + 1, m}}
    & {\varphi_{l + 1, m + 1}}
    \end{vmatrix}^{-1}
    \\
    = 
    \begin{vmatrix}
    \varphi_{l + 1, m}  
    & \varphi_{l + 1, m + 1}
    \\[1mm]
    {\varphi_{l + 2, m}}
    & \boxed{\varphi_{l + 2, m + 1}}
    \end{vmatrix}
    \varphi_{l + 1, m + 1}^{-1} \, \, \,
    \left(\varphi_{l + 1, m} \,
    \varphi_{l, m}^{-1} \,
    \varphi_{l, m + 1}\right) \,\,\,
    \begin{vmatrix}
    \varphi_{l, m}  
    & \varphi_{l, m + 1}
    \\[1mm]
    {\varphi_{l + 1, m}}
    & \boxed{\varphi_{l + 1, m + 1}}
    \end{vmatrix}^{-1}
    .
\end{multline}
Note also that
\begin{align}
    \begin{vmatrix}
    \varphi_{l, m}  
    & \varphi_{l, m + 1}
    \\[1mm]
    {\varphi_{l + 1, m}}
    & \boxed{\varphi_{l + 1, m + 1}}
    \end{vmatrix}
    &= \varphi_{l + 1, m + 1}
    - \varphi_{l + 1, m} \,
    \varphi_{l, m}^{-1} \,
    \varphi_{l, m + 1},
\end{align}
which implies
\begin{multline}
    \begin{vmatrix}
    \varphi_{l + 1, m}  
    & \varphi_{l + 1, m + 1}
    \\[1mm]
    \boxed{\varphi_{l + 2, m}}
    & {\varphi_{l + 2, m + 1}}
    \end{vmatrix}
    \, \cdot \,
    \begin{vmatrix}
    \varphi_{l, m}  
    & \varphi_{l, m + 1}
    \\[1mm]
    \boxed{\varphi_{l + 1, m}}
    & {\varphi_{l + 1, m + 1}}
    \end{vmatrix}^{-1}
    \\
    = 
    \begin{vmatrix}
    \varphi_{l + 1, m}  
    & \varphi_{l + 1, m + 1}
    \\[1mm]
    {\varphi_{l + 2, m}}
    & \boxed{\varphi_{l + 2, m + 1}}
    \end{vmatrix}
    \varphi_{l + 1, m + 1}^{-1} 
    \brackets{
    \varphi_{l + 1, m + 1}
    -
    \begin{vmatrix}
    \varphi_{l, m}  
    & \varphi_{l, m + 1}
    \\[1mm]
    {\varphi_{l + 1, m}}
    & \boxed{\varphi_{l + 1, m + 1}}
    \end{vmatrix}
    }
    \begin{vmatrix}
    \varphi_{l, m}  
    & \varphi_{l, m + 1}
    \\[1mm]
    {\varphi_{l + 1, m}}
    & \boxed{\varphi_{l + 1, m + 1}}
    \end{vmatrix}^{-1}
    .
\end{multline}
Finally, we have
\begin{multline}
    \begin{vmatrix}
    \varphi_{l, m} 
    & \varphi_{l, m + 1}  
    & \varphi_{l, m + 2}
    \\[1mm]
    \varphi_{l + 1, m} 
    & \varphi_{l + 1, m + 1}  
    & \varphi_{l + 1, m + 2}
    \\[1mm]
    \varphi_{l + 2, m} 
    & \varphi_{l + 2, m + 1}  
    & \boxed{\varphi_{l + 2, m + 2}}
    \end{vmatrix}
    = 
    \begin{vmatrix}
    \varphi_{l + 1, m + 1}  
    & \varphi_{l + 1, m + 2}
    \\[1mm]
    \varphi_{l + 2, m + 1}  
    & \boxed{\varphi_{l + 2, m + 2}}
    \end{vmatrix}
    - 
    \begin{vmatrix}
    \varphi_{l + 1, m}  
    & \varphi_{l + 1, m + 1}
    \\[1mm]
    {\varphi_{l + 2, m}}
    & \boxed{\varphi_{l + 2, m + 1}}
    \end{vmatrix}
    \\[2mm]
    \, \cdot
    \brackets{
    \begin{vmatrix}
    \varphi_{l, m}  
    & \varphi_{l, m + 1}
    \\[1mm]
    {\varphi_{l + 1, m}}
    & \boxed{\varphi_{l + 1, m + 1}}
    \end{vmatrix}^{-1}
    - \varphi_{l + 1, m + 1}^{-1}
    }
    \,
    \begin{vmatrix}
    \varphi_{l, m + 1}  
    & \varphi_{l, m + 2}
    \\[1mm]
    \varphi_{l + 1, m + 1}  
    & \boxed{\varphi_{l + 1, m + 2}}
    \end{vmatrix},
\end{multline}
or, in terms of $\theta_{l, m, n}$,
\begin{align}
    \theta_{l, m, 3}
    &= \theta_{l + 1, m + 1, 2}
    - \theta_{l + 1, m, 2} \, 
    \brackets{
    \theta_{l, m, 2}^{-1}
    - \theta_{l + 1, m + 1, 1}^{-1}
    } \,
    \theta_{l, m + 1, 2}.
\end{align}

\medskip
\textbullet \,\,
We proceed now to the case of arbitrary $n$. From the non-commutative Jacobi identity for  quasideterminants we use the relation
\begin{align}
    \label{eq:Theta_in}
    |\Theta_{l, m, n}|_{nn}
    &= |\Theta_{l + 1, m + 1, n - 1}|_{n - 1, n - 1}
    - |\Theta_{l + 1, m, n - 1}|_{n - 1, 1} \cdot
    \brackets{|\Theta_{l, m, n - 1}|_{n - 1, 1}}^{-1} \, 
    |\Theta_{l, m + 1, n - 1}|_{n-1, n-1}.
\end{align}
To rewrite this identity in terms of $\theta_{l, m, n}$ only, we use relation \eqref{eq:homrel_col} between quasiminors.
In particular,
\begin{multline}
    \begin{vmatrix}
    \varphi_{l, m} &
    \dots &
    \varphi_{l, m + n - 2} &
    \varphi_{l, m + n - 1}
    \\
    \vdots & \ddots & \vdots & \vdots
    \\
    \varphi_{l + n - 2, m} &
    \dots &
    \varphi_{l + n - 2, m + n - 2} &
    \varphi_{l + n - 2, m + n - 1}
    \\
    \boxed{\varphi_{l + n - 1, m}} &
    \dots &
    \varphi_{l + n - 1, m + n - 2} &
    \varphi_{l + n - 1, m + n - 1}
    \end{vmatrix}
    = - 
    \begin{vmatrix}
    \varphi_{l, m} &
    \dots &
    \varphi_{l, m + n - 2} &
    \varphi_{l, m + n - 1}
    \\
    \vdots & \ddots & \vdots & \vdots
    \\
    \varphi_{l + n - 2}, m &
    \dots &
    \varphi_{l + n - 2, m + n - 2} &
    \varphi_{l + n - 2, m + n - 1}
    \\
    {\varphi_{l + n - 1, m}} &
    \dots &
    \varphi_{l + n - 1, m + n - 2} &
    \boxed{\varphi_{l + n - 1, m + n - 1}}
    \end{vmatrix}
    \\
    \hspace{20mm}
    \cdot \,
    \begin{vmatrix}
    \varphi_{l, m + 1} &
    \dots &
    \varphi_{l, m + n - 1}
    \\
    \vdots & \ddots & \vdots
    \\
    \varphi_{l + n - 2, m + 1} &
    \dots &
    \boxed{\varphi_{l + n - 2, m + n - 1}}
    \end{vmatrix}^{-1}
    \, 
    \begin{vmatrix}
    \varphi_{l, m} &
    \dots &
    \varphi_{l, m + n - 2}
    \\
    \vdots & \ddots & \vdots
    \\
    \boxed{\varphi_{l + n - 2, m}} &
    \dots &
    \varphi_{l + n - 2, m + n - 2}
    \end{vmatrix},
\end{multline}
or, in terms of $\Theta_{l, m, n}$,
\begin{align}
    |\Theta_{l, m, n}|_{n 1}
    &= - |\Theta_{l, m, n}|_{nn} \, \brackets{
    |\Theta_{l, m + 1, n - 1}|_{n - 1, n - 1}
    }^{-1} \, 
    |\Theta_{l, m, n - 1}|_{n - 1, 1}.
\end{align}
Then
\begin{align}
    |\Theta_{l + 1, m, n - 1}|_{n - 1, 1} \, 
    \brackets{|\Theta_{l, m, n - 1}|_{n - 1, 1}}^{-1}
    &= |\Theta_{l + 1, m, n - 1}|_{n - 1, n - 1} \, \brackets{
    |\Theta_{l + 1, m + 1, n - 2}|_{n - 2, n - 2}
    }^{-1} \, 
    |\Theta_{l + 1, m, n - 2}|_{n - 2, 1}
    \\[1mm]
    &\quad \, \cdot 
    \brackets{|\Theta_{l, m, n - 2}|_{n - 2, 1}}^{-1} \,
    |\Theta_{l, m + 1, n - 2}|_{n - 2, n - 2} \, 
    \brackets{
    |\Theta_{l, m, n - 1}|_{n - 1, n - 1}
    }^{-1}.
\end{align}
Note that, according to the Jacobi identity \eqref{eq:Jacobi_nc}, we have
\begin{align}
    |\Theta_{l + 1, m, n - 2}|_{n - 2, 1} \,
    \brackets{|\Theta_{l, m, n - 2}|_{n - 2, 1}}^{-1} \,
    |\Theta_{l, m + 1, n - 2}|_{n - 2, n - 2}
    &= |\Theta_{l + 1, m + 1, n - 2}|_{n - 2, n - 2}
    - |\Theta_{l, m, n - 1}|_{n - 1, n - 1}
\end{align}
and, therefore,
\begin{align}
    |\Theta_{l + 1, m, n - 1}|_{n - 1, 1} \, 
    \brackets{|\Theta_{l, m, n - 1}|_{n - 1, 1}}^{-1}
    &= |\Theta_{l + 1, m, n - 1}|_{n - 1, n - 1} \, \brackets{
    |\Theta_{l + 1, m + 1, n - 2}|_{n - 2, n - 2}
    }^{-1}
    \\[1mm]
    &\quad \, \cdot 
    \brackets{
    |\Theta_{l + 1, m + 1, n - 2}|_{n - 2, n - 2}
    - |\Theta_{l, m, n - 1}|_{n - 1, n - 1}
    }\, 
    \brackets{
    |\Theta_{l, m, n - 1}|_{n - 1, n - 1}
    }^{-1}.
\end{align}
Hence, \eqref{eq:Theta_in} becomes
\begin{align}
    |\Theta_{l, m, n}|_{nn}
    &= |\Theta_{l + 1, m + 1, n - 1}|_{n - 1, n - 1}
    - |\Theta_{l + 1, m, n - 1}|_{n - 1, n - 1} \, \brackets{
    |\Theta_{l + 1, m + 1, n - 2}|_{n - 2, n - 2}
    }^{-1} \, 
    \\[1mm]
    & \quad \, \cdot \brackets{
    |\Theta_{l + 1, m + 1, n - 2}|_{n - 2, n - 2}
    - |\Theta_{l, m, n - 1}|_{n - 1, n - 1}
    }\, 
    \brackets{
    |\Theta_{l, m, n - 1}|_{n - 1, n - 1}
    }^{-1} \,
    |\Theta_{l, m + 1, n - 1}|_{n-1, n-1},
    \\[2mm]
    |\Theta_{l, m, n}|_{nn}
    &= |\Theta_{l + 1, m + 1, n - 1}|_{n - 1, n - 1}
    - |\Theta_{l + 1, m, n - 1}|_{n - 1, n - 1} \, 
    \\[1mm]
    & \quad \, \cdot \brackets{
    \brackets{
    |\Theta_{l, m, n - 1}|_{n - 1, n - 1}
    }^{-1}
    - \brackets{
    |\Theta_{l + 1, m + 1, n - 2}|_{n - 2, n - 2}
    }^{-1}
    } \,
    |\Theta_{l, m + 1, n - 1}|_{n-1, n-1},
    \\[2mm]
    \theta_{l, m, n}
    &= \theta_{l + 1, m + 1, n - 1}
    - \theta_{l + 1, m, n - 1} \, \brackets{
    \theta_{l, m, n - 1}^{-1}
    - \theta_{l + 1, m + 1, n - 2}^{-1}
    } \, 
    \theta_{l, m + 1, n - 1}.
\end{align}
After the shift $n \mapsto n - 1$, the latter coincides with \eqref{eq:2ddToda_nc}.
\end{proof}

\begin{prop}
\label{thm:2ddToda_scalsys}
The non-abelian 2d discrete Toda equation \eqref{eq:2ddToda_nc} is a consequence of the following scalar linear system
\begin{gather}
    \label{eq:2ddToda_scalsys}
    \left\{
    \begin{array}{rcl}
         \psi_{l + 1, m, n}
         &=& \psi_{l, m, n + 1}
         + a_{l, m, n} \psi_{l, m, n},
         \\[2mm]
         \psi_{l, m - 1, n + 1}
         &=& \psi_{l, m, n + 1}
         + b_{l, m, n} \psi_{l, m, n},
    \end{array}
    \right.
    \\[3mm]
    \label{eq:2ddToda_abdef}
    \begin{aligned}
    a_{l, m, n}
    &= \theta_{l + 1, m, n} \theta_{l, m, n}^{-1},
    &&&&&
    b_{l, m, n}
    &= \theta_{l, m - 1, n + 1} \theta_{l, m, n}^{-1}.
    \end{aligned}
\end{gather}
\end{prop}

\begin{proof}
Let us consider the compatibility condition of system \eqref{eq:2ddToda_scalsys}. Let \eqref{eq:2ddToda_scalsys}$_1$ and \eqref{eq:2ddToda_scalsys}$_2$ be the first and second equations of system \eqref{eq:2ddToda_scalsys}, respectively. Then, on the one hand, we have
\begin{align}
    T_{-2, 3} \brackets{\psi_{l + 1, m, n}}
    &\overset{\,\,\eqref{eq:2ddToda_scalsys}_1}{=}
    \psi_{l, m - 1, n + 2} 
    + a_{l, m - 1, n + 1} \psi_{l, m - 1, n + 1}
    \\[1mm]
    &\overset{\,\,\eqref{eq:2ddToda_scalsys}_2}{=}
    \psi_{l, m, n + 2} 
    + b_{l, m, n + 1} \psi_{l, m, n + 1}
    + a_{l, m - 1, n + 1} \brackets{
    \psi_{l, m, n + 1} 
    + b_{l, m, n} \psi_{l, m, n}
    }
    \\[1mm]
    &\overset{\,\,\phantom{\eqref{eq:2ddToda_scalsys}_2}}{=}
    \psi_{l, m, n + 2}
    + \brackets{
    b_{l, m, n + 1}
    + a_{l, m - 1, n + 1}
    } \psi_{l, m, n + 1}
    + a_{l, m - 1, n + 1} b_{l, m, n} \psi_{l, m, n}.
\end{align}
On the other hand, one can obtain
\begin{align}
    T_{1} \brackets{\psi_{l, m - 1, n + 1}}
    &\overset{\,\,\eqref{eq:2ddToda_scalsys}_2}{=}
    \psi_{l + 1, m, n + 1} 
    + b_{l + 1, m, n} \psi_{l + 1, m, n}
    \\[1mm]
    &\overset{\,\,\eqref{eq:2ddToda_scalsys}_1}{=}
    \psi_{l, m, n + 2} 
    + a_{l, m, n + 1} \psi_{l, m, n + 1}
    + b_{l + 1, m, n} \brackets{
    \psi_{l, m, n + 1} 
    + a_{l, m, n} \psi_{l, m, n}
    }
    \\[1mm]
    &\overset{\,\,\phantom{\eqref{eq:2ddToda_scalsys}_2}}{=}
    \psi_{l, m, n + 2}
    + \brackets{
    a_{l, m, n + 1}
    + b_{l + 1, m, n}
    } \psi_{l, m, n + 1}
    + b_{l + 1, m, n} a_{l, m, n} \psi_{l, m, n}.
\end{align}
Since we assume that the system \eqref{eq:2ddToda_scalsys} is compatible, i.e.
\begin{align}
    T_{-2, 3} \brackets{\psi_{l + 1, m, n}} = T_{1} \brackets{\psi_{l, m - 1, n + 1}},
\end{align}
the following conditions must be satisfied
\begin{align}
    \label{eq:2ddToda_scalsysab}
    \left\{
    \begin{array}{rcl}
        b_{l, m, n + 1} + a_{l, m - 1, n + 1}
        &=& a_{l, m, n + 1} + b_{l + 1, m, n},
        \\[2mm]
        a_{l, m - 1, n + 1} b_{l, m, n}
        &=& b_{l + 1, m, n} a_{l, m, n}.
    \end{array}
    \right.
\end{align}
After substitution of $a_{l, m, n}$ and $b_{l, m, n}$ defined by \eqref{eq:2ddToda_abdef} into \eqref{eq:2ddToda_scalsysab}$_2$, it turns into the identity:
\begin{align}
        a_{l, m - 1, n + 1} \cdot b_{l, m, n}
        &= b_{l + 1, m, n} \cdot a_{l, m, n}
        \\[2mm]
        &\Leftrightarrow
        \qquad
        \theta_{l + 1, m - 1, n + 1} \theta_{l, m - 1, n + 1}^{-1}
        \cdot \theta_{l, m - 1, n + 1} \theta_{l, m, n}^{-1}
        = \theta_{l + 1, m - 1, n + 1} \theta_{l + 1, m, n}^{-1}
        \cdot \theta_{l + 1, m, n} \theta_{l, m, n}^{-1},
\end{align}
while the first equation \eqref{eq:2ddToda_scalsysab}$_1$ coincides with the 2d discrete Toda lattice \eqref{eq:2ddToda_nc}:
\begin{align}
    b_{l, m, n + 1} + a_{l, m - 1, n + 1}
    &= a_{l, m, n + 1} + b_{l + 1, m, n}
    \\[2mm]
    &\Leftrightarrow
    \qquad
    a_{l, m + 1, n}
    = b_{l, m + 1, n} + a_{l, m, n} 
    - b_{l + 1, m + 1, n - 1}
    \\[2mm]
    &\Leftrightarrow
    \qquad
    \theta_{l + 1, m + 1, n} \theta_{l, m + 1, n}^{-1}
    = \theta_{l, m, n + 1} \theta_{l, m + 1, n}^{-1}
    + \theta_{l + 1, m, n} \theta_{l, m, n}^{-1} 
    - \theta_{l + 1, m, n} \theta_{l + 1, m + 1, n - 1}^{-1}
    \\[2mm]
    &\Leftrightarrow
    \qquad
    \theta_{l + 1, m + 1, n}
    = \theta_{l, m, n + 1} 
    + \theta_{l + 1, m, n} \brackets{
    \theta_{l, m, n}^{-1} 
    - \theta_{l + 1, m + 1, n - 1}^{-1}
    } \theta_{l, m + 1, n}.
\end{align}
Therefore, we are done.
\end{proof}

Once we have a scalar system equivalent to an equation, one may rewrite it in matrix form. We are going to rewrite system \eqref{eq:2ddToda_scalsys} by using either the two-component or semi-infinite formalisms. The two-component formalism gives a discrete analog of the Zakharov-Shabat equation, i.e. the zero-curvature condition, while the semi-infinite case leads to a discrete Lax equation. Regarding the first case, one can arrive at the following 
\begin{prop}
\label{thm:2ddToda_matsys}
Let the functions $\theta_{l, m, n}$ be solutions of the non-abelian 2d discrete Toda lattice \eqref{eq:2ddToda_nc}.Then the following matrix linear system
\begin{align}
    \label{eq:2ddToda_matsys}
    &&
    &
    \left\{
    \begin{array}{rcl}
         \Psi_{l, m, n + 1}
         &=& \mathcal{L}_{l, m, n} \Psi_{l, m, n},  
         \\[2mm]
         \Psi_{l, m - 1, n}
         &=& \mathcal{M}_{l, m, n} \Psi_{l, m, n},  
    \end{array}
    \right.
    &
    \Psi_{l, m, n}
    &= 
    \begin{pmatrix}
        \psi_{l, m, n} & \psi_{l, m, n - 1}
    \end{pmatrix}^{T},
    &&
\end{align}
where $\mathcal{L}_{l, m, n}$ and $\mathcal{M}_{l, m, n}$ are $2\times2$ matrices given by
\begin{align}
    \label{eq:2ddToda_pair}
    \begin{aligned}
    \mathcal{L}_{l, m, n}
    &= 
    \begin{pmatrix}
        T_{1, -2} 
        - \theta_{l + 1, m, n} \theta_{l + 1, m + 1, n - 1}^{-1}
        - \theta_{l + 1, m + 1, n} \theta_{l, m + 1, n}^{-1}
        & 
        - \theta_{l + 1, m, n}
        \\[2mm]
        \theta_{l, m + 1, n}^{-1} 
        & 
        0
    \end{pmatrix},
    \\[3mm]
    \mathcal{M}_{l, m, n}
    &= 
    \begin{pmatrix}
        1
        & 
        \theta_{l, m, n}
        \\[2mm]
        - \theta_{l + 1, m, n - 1}^{-1} 
        & 
        T_{1, - 2} - \theta_{l + 1, m, n - 1}^{-1} \theta_{l, m, n}
    \end{pmatrix}
    ,
    \end{aligned}
\end{align}
is compatible, i.e. a discrete analog of the zero-curvature condition holds{\rm:}
\begin{align}
    \label{eq:zrc_2dd}
    \mathcal{L}_{l, m - 1, n} 
    \mathcal{M}_{l, m, n}
    &= \mathcal{M}_{l, m, n + 1}
    \mathcal{L}_{l, m, n}.
\end{align}
\end{prop}

\begin{proof}
Following the two-component formalism, one is able to rewrite the scalar system \eqref{eq:2ddToda_scalsys} in the matrix form. During the proof, we first obtain the matrices $\mathcal{L}_{l, m, n}$, $\mathcal{M}_{l, m, n}$ by using \eqref{eq:2ddToda_scalsys}, then we bring them into the form \eqref{eq:2ddToda_pair} by a gauge transformation, and, finally, we verify that \eqref{eq:zrc_2dd} is equivalent to \eqref{eq:2ddToda_nc}.

\medskip
\textbf{\textbullet \,\, Two-component formalism.}

Let us rewrite components of the vector-functions $\Psi_{l, m, n + 1} = \begin{pmatrix} \psi_{l, m, n + 1} & \psi_{l, m, n}\end{pmatrix}^{T}$ and $\Psi_{l, m - 1, n} = \begin{pmatrix} \psi_{l, m - 1, n} & \psi_{l, m - 1, n - 1}\end{pmatrix}^{T}$ in terms of $\psi_{l, m, n}$ and $\psi_{l, m, n - 1}$, by using the linear scalar system \eqref{eq:2ddToda_scalsys} and the shift operator $T_{1, -2}$. Regarding the $\mathcal{L}_{l, m, n}$ matrix, we have
\begin{align}
    \psi_{l, m, n + 1}
    &\overset{\,\,\eqref{eq:2ddToda_scalsys}_1}{=}
    \psi_{l + 1, m, n}
    - a_{l, m, n} \psi_{l, m, n}
    \overset{\,\,\eqref{eq:2ddToda_scalsys}_2}{=}
    \psi_{l + 1, m - 1, n} 
    - b_{l + 1, m, n - 1} \psi_{l + 1, m, n + 1}
    - a_{l, m, n} \psi_{l, m, n}
    \\[2mm]
    &\overset{\,\,\phantom{\eqref{eq:2ddToda_scalsys}_1}}{=}
    \brackets{T_{1, -2} - a_{l, m, n}} \psi_{l, m, n}
    - b_{l + 1, m, n - 1} \psi_{l + 1, m, n - 1}
    \\[2mm]
    &\overset{\,\,\eqref{eq:2ddToda_scalsys}_1}{=}
    \brackets{T_{1, -2} - a_{l, m, n}} \psi_{l, m, n}
    - b_{l + 1, m, n - 1} \brackets{
    \psi_{l, m, n}
    + a_{l, m, n - 1} \psi_{l, m, n - 1}
    }
    \\[2mm]
    &\overset{\,\,\phantom{\eqref{eq:2ddToda_scalsys}_1}}{=}
    \begin{pmatrix}
        T_{1, -2} - a_{l, m, n} - b_{l + 1, m, n - 1}
        &
        - b_{l + 1, m, n - 1} a_{l, m, n - 1}
    \end{pmatrix}
    \begin{pmatrix}
        \psi_{l, m, n} & \psi_{l, m, n - 1}
    \end{pmatrix}^{T},
    \\[3mm]
    \psi_{l, m, n}
    &\overset{\,\,\phantom{\eqref{eq:2ddToda_scalsys}_1}}{=}
    \begin{pmatrix}
        1
        &
        0
    \end{pmatrix}
    \begin{pmatrix}
        \psi_{l, m, n} & \psi_{l, m, n - 1}
    \end{pmatrix}^{T},
\end{align}
and, therefore,
\begin{align}
    \mathcal{L}_{l, m, n}
    &= 
    \begin{pmatrix}
        T_{1, -2} - a_{l, m, n} - b_{l + 1, m, n - 1}
        &
        - b_{l + 1, m, n - 1} a_{l, m, n - 1}
        \\[2mm]
        1
        &
        0
    \end{pmatrix}
    .
\end{align}

Similarly, for matrix $\mathcal{M}_{l, m, n}$, one can get
\begin{align}
    \psi_{l, m - 1, n}
    &\overset{\,\,\eqref{eq:2ddToda_scalsys}_2}{=}
    \begin{pmatrix}
        1
        &
        b_{l, m, n - 1}
    \end{pmatrix}
    \begin{pmatrix}
        \psi_{l, m, n} & \psi_{l, m, n - 1}
    \end{pmatrix}^{T},
    \\[2mm]
    a_{l, m - 1, n - 1} \psi_{l, m - 1, n - 1}
    &\overset{\,\,\eqref{eq:2ddToda_scalsys}_1}{=}
    \psi_{l + 1, m - 1, n - 1} 
    - \psi_{l, m - 1, n}
    \overset{\,\,\eqref{eq:2ddToda_scalsys}_2}{=}
    T_{1, -2} \brackets{\psi_{l, m, n - 1}}
    - \brackets{
    \psi_{l, m, n} + b_{l, m, n - 1} \psi_{l, m, n - 1}
    }
    \\[2mm]
    &\overset{\,\,\phantom{\eqref{eq:2ddToda_scalsys}_1}}{=}
    \begin{pmatrix}
        - 1
        &
        T_{1, -2} - b_{l, m, n - 1}
    \end{pmatrix}
    \begin{pmatrix}
        \psi_{l, m, n} & \psi_{l, m, n - 1}
    \end{pmatrix}^{T},
\end{align}
or, finally,
\begin{align}
    \mathcal{M}_{l, m, n}
    &= 
    \begin{pmatrix}
        1 
        & 
        b_{l, m, n - 1}
        \\[2mm]
        - a_{l, m - 1, n - 1}^{-1}
        &
        a_{l, m - 1, n - 1}^{-1} \brackets{
        T_{1, -2} - b_{l, m, n - 1}
        }.
    \end{pmatrix}
\end{align}

\medskip
\textbf{\textbullet \,\, Gauge transformation.}
To bring matrices $\mathcal{L}_{l, m, n}$ and $\mathcal{M}_{l, m, n}$ to the form \eqref{eq:2ddToda_pair}, one needs to make a gauge-transformation with the help of a non-singular matrix $G_{l, m, n}$, i.e. we consider the change
\begin{align}
    \hat \Psi_{l, m, n}
    &= G_{l, m, n} \Psi_{l, m, n}
\end{align}
in system \eqref{eq:2ddToda_matsys}. Then, the matrices $\hat{\mathcal{L}}_{l, m, n}$, $\hat{\mathcal{M}}_{l, m, n}$ are defined as follows:
\begin{align}
    \label{eq:gaugerule_2dd}
    &&
    \hat{\mathcal{L}}_{l, m, n}
    &= G_{l, m, n + 1}^{-1} \mathcal{L}_{l, m, n} G_{l, m, n},
    &
    \hat{\mathcal{M}}_{l, m, n}
    &= G_{l, m - 1, n}^{-1} \mathcal{M}_{l, m, n} G_{l, m, n}.
    &&
\end{align}
Let $G_{l, m, n} = \diag\brackets{1, \,\, \theta_{l, m + 1, n - 1}}$, then its inverse is $G_{l, m, n}^{-1} = \diag\brackets{1, \,\, \theta_{l, m + 1, n - 1}^{-1}}$. Making a conjugation in the shifted $m \mapsto m - 1$ matrices $\mathcal{L}_{l, m, n}$ and $\mathcal{M}_{l, m, n}$, we obtain
\begin{align}
    \hat{\mathcal{L}}_{l, m, n}
    &= 
    \begin{pmatrix}
        1 & 0 \\[2mm] 0 & \theta_{l, m + 1, n}^{-1}
    \end{pmatrix}
    \begin{pmatrix}
        \begin{aligned}
        T_{1, -2} 
        &- \theta_{l + 1, m + 1, n} \, \theta_{l, m + 1, n}^{-1}
        \\
        &- \theta_{l + 1, m, n} \theta_{l + 1, m + 1, n - 1}^{-1}
        \end{aligned}
        & 
        - \theta_{l + 1, m, n} \theta_{l, m + 1, n - 1}^{-1}
        \\[2mm]
        1 & 0
    \end{pmatrix}
    \begin{pmatrix}
        1 & 0 \\[2mm] 0 & \theta_{l, m + 1, n - 1}
    \end{pmatrix}
    \\[2mm]
    &= 
    \begin{pmatrix}
        T_{1, -2}
        - \theta_{l + 1, m + 1, n} \theta_{l, m + 1, n}^{-1}
        - \theta_{l + 1, m, n} \theta_{l + 1, m + 1, n - 1}^{-1}
        &
        - \theta_{l + 1, m, n} \theta_{l, m + 1, n - 1}^{-1}
        \\[2mm]
        \theta_{l, m + 1, n}^{-1} 
        & 
        0
    \end{pmatrix}
    \begin{pmatrix}
        1 & 0 \\[2mm] 0 & \theta_{l, m + 1, n - 1}
    \end{pmatrix}
    \\[2mm]
    &= 
    \begin{pmatrix}
        T_{1, -2}
        - \theta_{l + 1, m + 1, n} \theta_{l, m + 1, n}^{-1}
        - \theta_{l + 1, m, n} \theta_{l + 1, m + 1, n - 1}^{-1}
        &
        - \theta_{l + 1, m, n}
        \\[2mm]
        \theta_{l, m + 1, n}^{-1} 
        & 
        0
    \end{pmatrix},
\end{align}
and 
\begin{align}
    \hat{\mathcal{M}}_{l, m, n}
    &= 
    \begin{pmatrix}
        1 & 0 \\[2mm] 0 & \theta_{l, m, n - 1}^{-1}
    \end{pmatrix}
    \begin{pmatrix}
        1
        & 
        \theta_{l, m, n} \theta_{l, m + 1, n - 1}^{-1}
        \\[2mm]
        - \theta_{l, m, n - 1} \theta_{l + 1, m, n - 1}^{-1} 
        & 
        \begin{aligned}
        &\theta_{l, m, n - 1} \theta_{l + 1, m, n - 1}^{-1} 
        \\
        &\brackets{
        T_{1, -2}
        - \theta_{l, m, n} \theta_{l, m + 1, n - 1}^{-1}
        }
        \end{aligned}
    \end{pmatrix}
    \begin{pmatrix}
        1 & 0 \\[2mm] 0 & \theta_{l, m + 1, n - 1}
    \end{pmatrix}
    \\[2mm]
    &= 
    \begin{pmatrix}
        1 
        & 
        \theta_{l, m, n} \theta_{l, m + 1, n - 1}^{-1}
        \\[2mm]
        - \theta_{l + 1, m, n - 1}^{-1}
        &
        \theta_{l + 1, m, n - 1}^{-1}
        \brackets{
        T_{1, -2}
        - \theta_{l, m, n} \theta_{l, m + 1, n - 1}^{-1}
        }
    \end{pmatrix}
    \begin{pmatrix}
        1 & 0 \\[2mm] 0 & \theta_{l, m + 1, n - 1}
    \end{pmatrix}
    \\[2mm]
    &= 
    \begin{pmatrix}
        1 
        & 
        \theta_{l, m, n}
        \\[2mm]
        - \theta_{l + 1, m, n - 1}^{-1}
        &
        \theta_{l + 1, m, n - 1}^{-1}
        \brackets{
        T_{1, -2}
        - \theta_{l, m, n} \theta_{l, m + 1, n - 1}^{-1}
        } \theta_{l, m + 1, n - 1}
    \end{pmatrix}
    \\[2mm]
    &= 
    \begin{pmatrix}
        1 
        & 
        \theta_{l, m, n}
        \\[2mm]
        - \theta_{l + 1, m, n - 1}^{-1}
        &
        \theta_{l + 1, m, n - 1}^{-1}
        T_{1, -2} \brackets{\theta_{l, m + 1, n - 1}}
        - \theta_{l + 1, m, n - 1}^{-1}
        \theta_{l, m, n}
    \end{pmatrix}
    \\[2mm]
    &= 
    \begin{pmatrix}
        1 
        & 
        \theta_{l, m, n}
        \\[2mm]
        - \theta_{l + 1, m, n - 1}^{-1}
        &
        T_{1, -2}
        - \theta_{l + 1, m, n - 1}^{-1}
        \theta_{l, m, n}
    \end{pmatrix}
    .
\end{align}
Below we will omit the signs $\hat{\phantom{u}}$, hoping that it will not lead to misunderstandings.

\medskip
\textbf{\textbullet \,\, Zero-curvature condition.}
The linear system \eqref{eq:2ddToda_matsys} is compatible if
\begin{align}
    \tag{\ref{eq:zrc_2dd}}
    \mathcal{L}_{l, m - 1, n} 
    \mathcal{M}_{l, m, n}
    &= \mathcal{M}_{l, m, n + 1}
    \mathcal{L}_{l, m, n}.
\end{align}
Verifying the discrete analog of the zero-curvature condition, we have
\begin{align}
    l.h.s.\eqref{eq:zrc_2dd}
    &=
    \begin{pmatrix}
        T_{1, -2} - \theta_{l + 1, m, n} \theta_{l, m, n}^{-1}
        &
        - \theta_{l + 1, m, n}
        \\[2mm]
        \theta_{l, m, n}^{-1}
        &
        1
    \end{pmatrix},
    \\[2mm]
    r.h.s.\eqref{eq:zrc_2dd}
    &= 
    \begin{pmatrix}
        T_{1, -2}
        - \theta_{l + 1, m + 1, n} \theta_{l, m + 1, n}^{-1}
        - \theta_{l + 1, m, n} \theta_{l + 1, m + 1, n - 1}^{-1}
        + \theta_{l, m, n + 1} \theta_{l, m + 1, n}^{-1}
        &
        - \theta_{l + 1, m, n}
        \\[2mm]
        \theta_{l + 1, m, n}^{-1} \theta_{l + 1, m + 1, n} \theta_{l, m + 1, n}^{-1}
        + \theta_{l + 1, m + 1, n - 1}^{-1}
        - \theta_{l + 1, m, n}^{-1}
        \theta_{l, m, n + 1} \theta_{l, m + 1, n}^{-1}
        &
        1
    \end{pmatrix},
    \\[4mm]
    l.h.s.\eqref{eq:zrc_2dd}
    &- r.h.s.\eqref{eq:zrc_2dd}
    \\[2mm]
    &= 
    \begin{pmatrix}
        \brackets{
        \theta_{l + 1, m + 1, n}
        - \theta_{l, m, n + 1}
        - \theta_{l + 1, m , n}
        \brackets{
        \theta_{l, m, n}^{-1}
        - \theta_{l + 1, m + 1, n - 1}^{-1}
        } \theta_{l, m + 1, n}
        } \theta_{l, m + 1, n}^{-1}
        &
        0
        \\[2mm]
        - \theta_{l + 1, m, n}^{-1}
        \brackets{
        \theta_{l + 1, m + 1, n}
        - \theta_{l, m, n + 1}
        - \theta_{l + 1, m, n} \brackets{
        \theta_{l, m, n}^{-1}
        - \theta_{l + 1, m + 1, n - 1}^{-1}
        } \theta_{l, m + 1, n}
        } \theta_{l, m + 1, n}^{-1}
        &
        0
    \end{pmatrix}.
\end{align}
Therefore, the condition \eqref{eq:zrc_2d} holds if $\theta_{l, m, n}$ solves the non-abelian 2d discrete Toda lattice \eqref{eq:2ddToda_nc}.
\end{proof}

As we have mentioned before, besides the two-component formalism, one can also introduce a semi-infinite vector-function in order to derive a discrete analog of the Lax equation. Due to the further reductions to the 1d cases, instead of the scalar system \eqref{eq:2ddToda_scalsys}, we prefer to consider the equivalent system
\begin{gather}
    \label{eq:2ddToda_scalsys_2}
    \left\{
    \begin{array}{rcl}
         \psi_{l + 1, m - 1, n}
         &=& \psi_{l, m, n + 1}
         + \brackets{
         a_{l, m - 1, n} + b_{l, m, n}
         } 
         \psi_{l, m, n}
         + a_{l, m - 1, n} b_{l, m, n - 1} \psi_{l, m, n - 1}
         ,
         \\[2mm]
         \psi_{l, m - 1, n}
         &=& \psi_{l, m, n}
         + b_{l, m, n - 1} \psi_{l, m, n - 1},
    \end{array}
    \right.
\end{gather}
where the functions $a_{l, m, n}$ and $b_{l, m, n}$ are defined as before by formula \eqref{eq:2ddToda_abdef}. 
\begin{prop}
\label{thm:2ddToda_matsys_n}
The non-abelian semi-infinite 2d discrete Toda equation \eqref{eq:2ddToda_nc} is equivalent to the system
\begin{align}
    \label{eq:2ddToda_matsys_n}
    &\left\{
    \begin{array}{rcl}
         L_{l, m} \Psi_{l, m} 
         &=& \Psi_{l + 1, m - 1}
         ,
         \\[2mm]
         \Psi_{l, m - 1}
         &=& M_{l, m} \Psi_{l, m},
    \end{array}
    \right.
    &
    \Psi_{l, m}
    &= 
    \begin{pmatrix}
        \psi_{l, m, 1}
        &
        \psi_{l, m, 2}
        &
        \psi_{l, m, 3}
        &
        \dots
    \end{pmatrix}^T
\end{align}
with matrices $L_{l, m}$ and $M_{l, m}$ of the form
\begin{align}
    \label{eq:2ddToda_pair_n}
    \begin{aligned}
    L_{l, m}
    &= 
    \begin{pmatrix}
        a_{l, m - 1, 1} + b_{l, m, 1}
        &
        1
        &&&
        \\[1mm]
        a_{l, m - 1, 2} b_{l, m, 1}
        &
        a_{l, m - 1, 2} + b_{l, m, 2}
        &
        1
        &
        &
        \\[1mm]
        &
        a_{l, m - 1, 3} b_{l, m, 2}
        &
        a_{l, m - 1, 3} + b_{l, m, 3}
        &
        \ddots
        &
        \\[1mm]
        &
        &
        \ddots
        &
        \ddots
    \end{pmatrix}
    \\
    &\equiv \sum_{k \geq 1} E_{k, k + 1}
    + \sum_{k \geq 1} \brackets{a_{l, m - 1, k} + b_{l, m, k}} E_{k, k}
    + \sum_{k \geq 1} a_{l, m - 1, k} b_{l, m, k} E_{k + 1, k}
    ,
    \\[3mm]
    M_{l, m}
    &= 
    \begin{pmatrix}
        1
        &
        &&&
        \\[1mm]
        b_{l, m, 1}
        &
        1
        &
        &
        &
        \\[1mm]
        &
        b_{l, m, 2}
        &
        1
        &
        &
        \\[1mm]
        &
        &
        \ddots
        &
        \ddots
    \end{pmatrix}
    \equiv \sum_{k \geq 1} E_{k, k}
    + \sum_{k \geq 1} b_{l, m, k} E_{k + 1, k}
    .
    \end{aligned}
\end{align}
\end{prop}
\begin{proof}
In order to derive the matrices $L_{l, m}$ and $M_{l, m}$, one needs to introduce the semi-infinite vector-function
\begin{align}
    \Psi_{l, m}
    &= 
    \begin{pmatrix}
        \psi_{l, m, 1}
        &
        \psi_{l, m, 2}
        &
        \psi_{l, m, 3}
        &
        \dots
    \end{pmatrix}^T
    .
\end{align}
Then, the first equation \eqref{eq:2ddToda_scalsys_2}$_1$ gives rise to the matrix $L_{l, m}$, while the second equation \eqref{eq:2ddToda_scalsys_2}$_2$ leads to the matrix~$M_{l, m}$. Note that for any $l$ and $m$ the function $b_{l, m, 0}$ is equal to zero (see the definition of $\Theta_{l, m, n}$). The system \eqref{eq:2ddToda_matsys_n} with matrices \eqref{eq:2ddToda_pair_n} yields the discrete analog of the Lax equation
\begin{align}
    \label{eq:lax_2dd}
    L_{l, m - 1} M_{l, m}
    &= M_{l + 1, m - 1} L_{l, m},
\end{align}
from which we derive the semi-infinite 2d discrete Toda equation \eqref{eq:2ddToda_nc}. 
\end{proof}

\begin{rem}
The matrices $L_{l, m}$ and $M_{l, m}$ are non-abelian generalizations of those presented in \cite{krichever1997quantum} (eq. (3.3) therein),
to which they reduce in the commutative case.
\end{rem}

\begin{rem}
Note that one needs to consider boundary conditions in the case of finite lattices.
\end{rem}

We conclude this subsection with a specialisation of the functions $\varphi_{l, m}$ useful below. Note that the definition of the discrete derivation
\begin{align}
    \varphi_{l + k}
    := \varepsilon^{- k} 
    \sum_{j = 0}^{k}
    (-1)^{j} 
    \binom{k}{j}
    \phi_{l - k}
\end{align}
generalizes to the case of two indices as follows
\begin{align}
    \label{eq:spec_2dd}
    &&
    \varphi_{l + k_1, m + k_2}
    &= \varepsilon_1^{- k_1} \, \varepsilon_2^{- k_2}
    \sum_{j_1 = 0}^{k_1} \sum_{j_2 = 0}^{k_2} (-1)^{j_1 + j_2} 
    \binom{k_1}{j_1} \binom{k_2}{j_2}
    \phi_{l - k_1, m - k_2},
    &
    0 \leq k_1, k_2 \leq n - 1
    &&
\end{align}
and can be regarded as a specialization of the solutions of \eqref{eq:2ddToda_nc}.

\subsection{Continuous case}
\label{sec:2dToda}

We introduce now a non-abelian analog of the two-dimensional Toda field equation, by using the continuous limit of \eqref{eq:2ddToda_nc}. As a result, we obtain its scalar and matrix linear problems and quasideterminant solutions. Here the main change of the variables is given by formula (see, e.g., \cite{hietarinta2016discrete})
\begin{align}
    \label{eq:2ddto2d_th}
    &&
    \theta_n (x, y)
    &= \varepsilon^{l + m} \theta_{l, m, n},
    &
    x
    &= \varepsilon l,
    &
    y
    &= \varepsilon m.
    &&
\end{align}
Making this change in \eqref{eq:2ddToda_nc} and taking the Taylor expansion near $\varepsilon = 0$ for the~functions $\theta_n (x + k_1 \varepsilon, y + k_2 \varepsilon)$, 
\begin{align}
    \theta_{n} (x + k_1 \varepsilon, y + k_2 \varepsilon)
    = \theta_n (x, y)
    &+ \, \varepsilon \brackets{
    k_1 \theta_{n, x} (x, y)
    + k_2 \theta_{n, y} (x, y)
    }
    \\
    &+ \, \tfrac12 \varepsilon^2 \brackets{
    k_1^2 \theta_{n, xx} (x, y)
    + 2 k_1 k_2 \theta_{n, xy} (x, y)
    + k_2^2 \theta_{n, yy} (x, y)
    }
    + O (\varepsilon^3),
\end{align}
where
\begin{align}
    \theta_{n, x} (x, y)
    &= \left. \brackets{
    \dfrac{\partial\phantom{(+ k_1 \varepsilon)}}{\partial(x + k_1 \varepsilon)} \theta_n (x + k_1 \varepsilon, y + k_2 \varepsilon)
    } \right|_{\varepsilon = 0},
    &    
    \theta_{n, y} (x, y)
    &= \left. \brackets{
    \dfrac{\partial\phantom{(+ k_2 \varepsilon)}}{\partial(y + k_2 \varepsilon)} \theta_n (x + k_1 \varepsilon, y + k_2 \varepsilon)
    } \right|_{\varepsilon = 0}
    ,
\end{align}
one can arrive to the following non-abelian analog of the 2d Toda field equation
\begin{align}
    \label{eq:2dToda_nc}
    \brackets{
    \theta_{n, x} \theta_n^{-1}
    }_y
    &= \theta_{n + 1} \theta_n^{-1}
    - \theta_n \theta_{n - 1}^{-1}
    .
\end{align}
We remark that derivatives $\partial_x$ and $\partial_y$ commute with each other.

\begin{prop}
\label{thm:2dToda_scalsys}
Consider the scalar linear system \eqref{eq:2ddToda_scalsys} and the change of variables \eqref{eq:2ddto2d_th} supplemented~by
\begin{align}
    \label{eq:2ddto2d_psi}
    \psi_n (x, y)
    &= \varepsilon^l \psi_{l, m, n}.
\end{align}
Then, the limit $\varepsilon \to 0$ leads to the system
\begin{gather}
    \label{eq:2dToda_scalsys}
    \left\{
    \begin{array}{rcl}
         \psi_{n, x}
         &=& \psi_{n + 1}
         + a_n 
         \psi_{n},
         \\[2mm]
         \psi_{n + 1, y}
         &=& - b_n \psi_{n},
    \end{array}
    \right.
    \\[3mm]
    \label{eq:2dToda_abdef}
    \begin{aligned}
    a_n
    &= \theta_{n, x} \theta_n^{-1},
    &&&
    b_n
    &= \theta_{n + 1} \theta_n^{-1},
    \end{aligned}
\end{gather}
from which the non-abelian 2d Toda field equation \eqref{eq:2dToda_nc} follows.
\end{prop}

\begin{proof}
We first consider the continuous limit and then verify the compatibility condition.

\medskip
\textbf{\textbullet \,\, Continuous limit.}
The change of variables defined by \eqref{eq:2ddto2d_th} and \eqref{eq:2ddto2d_psi} brings system \eqref{eq:2ddToda_scalsys} to the form
\begin{gather}
    \left\{
    \begin{array}{rcl}
         \psi_{n} (x + \varepsilon, y)
         &=& \varepsilon \psi_{n + 1} (x, y)
         + a_{n} (x, y) \psi_{n} (x, y)
         ,
         \\[3mm]
         \psi_{n + 1} (x, y - \varepsilon)
         &=& \psi_{n + 1} (x, y)
         + \varepsilon b_{n} (x, y) \psi_{n} (x, y)
         ,
    \end{array}
    \right.
\end{gather}
where 
\begin{align}
    &&
    a_{n} (x, y)
    &= \theta_{n} (x + \varepsilon, y) \theta_{n} (x, y)^{-1},
    &
    b_{n} (x, y)
    &= \theta_{n + 1} (x, y - \varepsilon) \theta_{n} (x, y)^{-1}.
    &&
\end{align}
Taking the series near $\varepsilon = 0$, this system turns into
\begin{gather}
    \left\{
    \begin{array}{rcl}
         \psi_{n, x}
         &=& \psi_{n + 1}
         + \theta_{n, x} \theta_n^{-1} \psi_{n}
         + O(\varepsilon)
         ,
         \\[2mm]
         \psi_{n + 1, y}
         &=& - \theta_{n + 1} \theta_n^{-1} \psi_{n}
         + O(\varepsilon)
    \end{array}
    \right.
\end{gather}
and coincides with \eqref{eq:2dToda_scalsys} after the limit $\varepsilon \to 0$.

\medskip
\textbf{\textbullet \,\, Compatibility condition.}
The compatibility condition of system \eqref{eq:2dToda_scalsys} means 
\begin{align}
    \label{eq:compcond_2d}
    \brackets{
    \psi_{n, x}
    }_y
    &= \brackets{
    \psi_{n, y}
    }_x.
\end{align}
For the l.h.s., we have the following chain of identities
\begin{align}
    \brackets{
    \psi_{n, x}
    }_y
    &\overset{\,\,\eqref{eq:2dToda_scalsys}_1}{=} 
    \psi_{n + 1, y}
    + \brackets{
    \theta_{n, x} \theta_n^{-1}
    }_y \psi_n
    + \theta_{n, x} \theta_n^{-1} \psi_{n, y}
    \\[2mm]
    &\overset{\,\,\eqref{eq:2dToda_scalsys}_2}{=}  - \theta_{n + 1} \theta_n^{-1} \psi_n
    + \brackets{
    \theta_{n, x} \theta_n^{-1}
    }_y \psi_n
    + \theta_{n, x} \theta_n^{-1} \brackets{
    - \theta_{n} \theta_{n - 1}^{-1} \psi_{n - 1}
    }
    \\[2mm]
    &\overset{\,\,\phantom{\eqref{eq:2dToda_scalsys}_1}}{=}  - \theta_{n + 1} \theta_n^{-1} \psi_n
    + \brackets{
    \theta_{n, x} \theta_n^{-1}
    }_y \psi_n
    - \theta_{n, x} \theta_{n - 1}^{-1} \psi_{n - 1}.
\end{align}
The r.h.s. can be rewritten in a similar way:
\begin{align}
    \brackets{
    \psi_{n, y}
    }_x
    &\overset{\,\,\eqref{eq:2dToda_scalsys}_2}{=} 
    - \theta_{n, x} \theta_{n - 1}^{-1} \psi_{n - 1}
    + \theta_n \theta_{n - 1}^{-1} \theta_{n - 1, x} \theta_{n - 1}^{-1} \psi_n
    - \theta_n \theta_{n - 1}^{-1} \psi_{n - 1, x}
    \\[2mm]
    &\overset{\,\,\eqref{eq:2dToda_scalsys}_1}{=}  
    - \theta_{n, x} \theta_{n - 1}^{-1} \psi_{n - 1}
    + \theta_n \theta_{n - 1}^{-1} \theta_{n - 1, x} \theta_{n - 1}^{-1} \psi_n
    - \theta_n \theta_{n - 1}^{-1} \brackets{
    \psi_n 
    + \theta_{n - 1, x} \theta_{n - 1}^{-1} \psi_{n - 1}
    }
    \\[2mm]
    &\overset{\,\,\phantom{\eqref{eq:2dToda_scalsys}_1}}{=}  
    - \theta_{n, x} \theta_{n - 1}^{-1} \psi_{n - 1}
    - \theta_n \theta_{n - 1}^{-1} \psi_n
    .
\end{align}
Thus, the condition \eqref{eq:compcond_2d} yields \eqref{eq:2dToda_nc}.
\end{proof}

A continuous analog of Proposition \ref{thm:2ddToda_matsys} can be obtained by passing to a  continuous limit and is given by the following
\begin{prop}
\label{thm:2dToda_matsys}
Let $\partial:= \partial_x - \partial_y$ and $2 \times 2$ matrices $\mathcal{L}_{n}$ and $\mathcal{M}_{n}$ be defined as
\begin{align}
    \label{eq:2dToda_pair}
    &&
    \mathcal{L}_{n}
    &= 
    \begin{pmatrix}
        \partial 
        - \theta_{n, x} \theta_{n}^{-1}
        & 
        - \theta_{n}
        \\[2mm]
        \theta_{n}^{-1} 
        & 
        0
    \end{pmatrix},
    &
    \mathcal{M}_{n}
    &= 
    \begin{pmatrix}
        0
        & 
        - \theta_{n}
        \\[2mm]
        \theta_{n - 1}^{-1} 
        & 
        - \partial
    \end{pmatrix}
    .
    &&
\end{align}
Then, the matrix linear system
\begin{align}
    \label{eq:2dToda_matsys}
    &&
    &
    \left\{
    \begin{array}{rcl}
         \Psi_{n + 1}
         &=& \mathcal{L}_{n} \Psi_{n},  
         \\[2mm]
         \Psi_{n, y}
         &=& \mathcal{M}_{n} \Psi_{n},  
    \end{array}
    \right.
    &
    \Psi_{n}
    &= 
    \begin{pmatrix}
        \psi_{n} & \psi_{n - 1}
    \end{pmatrix}^{T},
    &&
\end{align}
is compatible, i.e.
\begin{align}
    \label{eq:zrc_2d}
    \mathcal{L}_{n, y} 
    &= \mathcal{M}_{n + 1}
    \mathcal{L}_{n}
    - \mathcal{L}_n \mathcal{M}_n,
\end{align}
if function $\theta_{n}$ satisfies the non-abelian 2d Toda field equation \eqref{eq:2dToda_nc}.
\end{prop}
\begin{proof}
In the proof, we will take a continuous limit and then verify the Zakharov-Shabat equation.

\medskip
\textbf{\textbullet \,\, Continuous limit.}
Let us first make a rescaling of the variable $\theta_{l, m, n}$ that brings the matrices \eqref{eq:2ddToda_pair} to the form
\begin{align}
    \begin{aligned}
    \mathcal{L}_{l, m, n}
    &= 
    \begin{pmatrix}
        T_{1, -2} 
        - \varepsilon \theta_{l + 1, m, n} \theta_{l + 1, m + 1, n - 1}^{-1}
        - \varepsilon^{-1} \theta_{l + 1, m + 1, n} \theta_{l, m + 1, n}^{-1}
        & 
        - \varepsilon^{- l - m - 1} \theta_{l + 1, m, n}
        \\[2mm]
        \varepsilon^{l + m + 1} \theta_{l, m + 1, n}^{-1} 
        & 
        0
    \end{pmatrix},
    \\[3mm]
    \mathcal{M}_{l, m, n}
    &= 
    \begin{pmatrix}
        1
        & 
        \varepsilon^{- l - m} \theta_{l, m, n}
        \\[2mm]
        - \varepsilon^{l + m + 1} \theta_{l + 1, m, n - 1}^{-1} 
        & 
        T_{1, - 2} - \varepsilon \theta_{l + 1, m, n - 1}^{-1} \theta_{l, m, n}
    \end{pmatrix}
    ,
    \end{aligned}
\end{align}
or, after a conjugation with the help of the matrix $G_{l, m} = \diag\brackets{1, \varepsilon^{l + m + 1}}$
\begin{align}
    \begin{aligned}
    \mathcal{L}_{l, m, n}
    &= 
    \begin{pmatrix}
        T_{1, -2} 
        - \varepsilon \theta_{l + 1, m, n} \theta_{l + 1, m + 1, n - 1}^{-1}
        - \varepsilon^{-1} \theta_{l + 1, m + 1, n} \theta_{l, m + 1, n}^{-1}
        & 
        - \theta_{l + 1, m, n}
        \\[2mm]
        \theta_{l, m + 1, n}^{-1} 
        & 
        0
    \end{pmatrix},
    \\[3mm]
    \mathcal{M}_{l, m, n}
    &= 
    \begin{pmatrix}
        1
        & 
        \varepsilon \theta_{l, m, n}
        \\[2mm]
        - \varepsilon \theta_{l + 1, m, n - 1}^{-1} 
        & 
        \varepsilon T_{1, - 2} - \varepsilon^2 \theta_{l + 1, m, n - 1}^{-1} \theta_{l, m, n}
    \end{pmatrix}
    .
    \end{aligned}
\end{align}
Then, replacing the shift operator $T_{1, -2}$ by $\partial$, making the transformation \eqref{eq:2ddto2d_th}, and taking the series near $\varepsilon = 0$, the matrices become 
\begin{align}
    \mathcal{L}_{n}
    &= 
    \begin{pmatrix}
        \partial 
        - \varepsilon^{-1}
        - \theta_{n, x} \theta_n^{-1}
        & 
        - \theta_{n}
        \\[2mm]
        \theta_{n}^{-1} 
        & 
        0
    \end{pmatrix}
    + O(\varepsilon)
    ,
    &
    \mathcal{M}_{n}
    &= 
    \begin{pmatrix}
        1
        & 
        \varepsilon \theta_{n}
        \\[2mm]
        - \varepsilon \theta_{n - 1}^{-1} 
        & 
        \varepsilon \partial
    \end{pmatrix}
    + O (\varepsilon^2)
    ,
\end{align}
or, after the shift $\partial \mapsto \partial - \varepsilon^{-1}$,
\begin{align}
    \mathcal{L}_{n}
    &= 
    \begin{pmatrix}
        \partial 
        - \theta_{n, x} \theta_n^{-1}
        & 
        - \theta_{n}
        \\[2mm]
        \theta_{n}^{-1} 
        & 
        0
    \end{pmatrix}
    + O(\varepsilon)
    = \hat{\mathcal{L}}_n
    + O(\varepsilon)
    ,
    \\[3mm]
    \mathcal{M}_{n}
    &= \mathbf{I}
    + \varepsilon
    \begin{pmatrix}
        0
        & 
        \theta_{n}
        \\[2mm]
        - \theta_{n - 1}^{-1} 
        & 
        \partial
    \end{pmatrix}
    + O (\varepsilon^2)
    = \mathbf{I}
    + \varepsilon \hat{\mathcal{M}}_n
    + O (\varepsilon^2)
    ,
\end{align}
where $\mathbf{I}$ is the identity matrix. 

The matrix system \eqref{eq:2ddToda_matsys} transforms into
\begin{align}
    \left\{
    \begin{array}{rcl}
         \Psi_{n + 1}
         &=& \brackets{
         \hat{\mathcal{L}}_n 
         + O (\varepsilon)
         }
         \Psi_{n},  
         \\[2mm]
         \Psi_n 
         - \varepsilon \Psi_{n, y}
         + O (\varepsilon^2)
         &=& \brackets{
         \mathbf{I}
         + \varepsilon \hat{\mathcal{M}}_n
         + O (\varepsilon^2)
         }
         \Psi_{n}.
    \end{array}
    \right.
\end{align}
Taking the limit $\varepsilon \to 0$ and making the rescaling $\hat{\mathcal{M}}_n \mapsto - \hat{\mathcal{M}}_n$, it becomes
\begin{align}
    \left\{
    \begin{array}{rcl}
         \Psi_{n + 1}
         &=& \hat{\mathcal{L}}_n
         \Psi_{n},  
         \\[2mm]
         \Psi_{n, y}
         &=& \hat{\mathcal{M}}_n
         \Psi_{n}.
    \end{array}
    \right.
\end{align}
Omitting the signs $\hat{\phantom{u}}$, we arrive to the system \eqref{eq:2dToda_matsys} with matrices given by \eqref{eq:2dToda_pair}.

\medskip
\textbf{\textbullet \,\, Zero-curvature condition.}
Regarding the auxiliary function $f = f(x)$,  condition \eqref{eq:zrc_2d} can be written as follows:
\begin{align}
    l.h.s. \eqref{eq:zrc_2d}
    &\equiv \partial_y \brackets{\mathcal{L}_n (f)} 
    \\[2mm]
    &= \partial_y
    \begin{pmatrix}
        f' - \theta_{n, x} \theta_n^{-1} f
        &
        - \theta_n f
        \\[1mm]
        \theta_n^{-1} f
        &
        0
    \end{pmatrix}
    = 
    \begin{pmatrix}
        - \brackets{\theta_{n, x} \theta_n^{-1}}_y f
        &
        - \theta_{n, y} f
        \\[1mm]
        \brackets{\theta_n^{-1}}_y f
        & 
        0
    \end{pmatrix},
    \\[4mm]
    r.h.s.\eqref{eq:zrc_2d}
    &\equiv \mathcal{M}_{n + 1} \mathcal{L}_n (f)
    - \mathcal{L}_n \mathcal{M}_n (f)
    \\[2mm]
    &= 
    \begin{pmatrix}
        0 & - \theta_{n + 1}
        \\[1mm]
        \theta_n^{-1} & - \partial
    \end{pmatrix}
    \begin{pmatrix}
        f' - \theta_{n, x} \theta_n^{-1} f
        &
        - \theta_n f
        \\[1mm]
        \theta_n^{-1} f
        & 
        0
    \end{pmatrix}
    -
    \begin{pmatrix}
        \partial - \theta_{n, x} \theta_n^{-1}
        & 
        - \theta_n
        \\[1mm]
        \theta_n^{-1}
        &
        0
    \end{pmatrix}
    \begin{pmatrix}
        0 & - \theta_n f
        \\[1mm]
        \theta_{n - 1}^{-1} f
        &
        - f'
    \end{pmatrix}
    \\[2mm]
    &= 
    \begin{pmatrix}
        - \theta_{n + 1} \theta_n^{-1} f & 0
        \\[1mm]
        \brackets{\theta_n^{-1}}_y f & - f
    \end{pmatrix}
    - 
    \begin{pmatrix}
        - \theta_n \theta_{n - 1}^{-1} f
        &
        \theta_{n, y} f
        \\[1mm]
        0 & - f
    \end{pmatrix},
\end{align}
and, finally,
\begin{align}
    l.h.s.\eqref{eq:zrc_2d}
    - r.h.s.\eqref{eq:zrc_2d}
    &= \brackets{
    - \brackets{\theta_{n, x} \theta_n^{-1}}_y
    + \theta_{n + 1} \theta_n^{-1}
    - \theta_n \theta_{n - 1}^{-1}
    } f E_{1,1}, 
\end{align}
where $E_{i,j}$ is a standard notation for the unit matrix. Since the function $f$ is arbitrary, the condition \eqref{eq:zrc_2d} is satisfied if $\theta_n = \theta_n (x, y)$ solves the non-abelian 2d Toda field equation \eqref{eq:2dToda_nc}.
\end{proof}

Following the same computations of the continuous limit discussed in the proof of Proposition \ref{thm:2dToda_matsys}, we derive a semi-infinite formalism for the non-abelian 2d Toda field equation \eqref{eq:2dToda_nc}.

\begin{prop}
\label{thm:2dToda_matsys_n}
Let matrices $L$ and $M$ read as
\begin{align}
    \label{eq:2dToda_pair_n}
    L
    &= 
    \begin{pmatrix}
        a_{1}
        &
        1
        &&&
        \\[1mm]
        b_{1}
        &
        a_{2}
        &
        1
        &
        &
        \\[1mm]
        &
        b_{2}
        &
        a_{3}
        &
        \ddots
        &
        \\[1mm]
        &
        &
        \ddots
        &
        \ddots
    \end{pmatrix}
    ,
    &
    M
    &= 
    \begin{pmatrix}
        0
        &
        &&&
        \\[1mm]
        - b_{1}
        &
        0
        &
        &
        &
        \\[1mm]
        &
        - b_{2}
        &
        0
        &
        &
        \\[1mm]
        &
        &
        \ddots
        &
        \ddots
    \end{pmatrix},
\end{align}
where unfilled matrix entries are equal to zero. 
Then, non-abelian semi-infinite 2d discrete Toda equation~\eqref{eq:2dToda_nc} follows from the system
\begin{align}
    \label{eq:2dToda_matsys_n}
    &\left\{
    \begin{array}{rcl}
         L \Psi
         &=& \partial \Psi
         ,
         \\[2mm]
         \partial_y \Psi
         &=& M \Psi,
    \end{array}
    \right.
    &
    \Psi
    &= 
    \begin{pmatrix}
        \psi_{1}
        &
        \psi_{2}
        &
        \psi_{3}
        &
        \dots
    \end{pmatrix}^T
    .
\end{align}
\end{prop}
\begin{proof}
The statement of the proposition can be achieved just by reduction data defined in \eqref{eq:2ddto2d_th}. 
Indeed, the matrices \eqref{eq:2ddToda_pair_n} after  change of variables and series expansion read as
\begin{align}
    L
    = \varepsilon^{-1} \mathbf{I}
    + \sum_{k \geq 1} a_k \,  E_{k, k}
    + \sum_{k \geq 1} E_{k, k + 1}
    + O (\varepsilon)
    = \varepsilon^{-1} \mathbf{I}
    + \hat{L}
    + O(\varepsilon)
    ,
    \hspace{6cm}
    \\
    M
    = \mathbf{I}
    + \varepsilon \sum_{k \geq 1} b_k E_{k, k}
    + O(\varepsilon^2)
    = \mathbf{I}
    + \varepsilon \hat M
    + O(\varepsilon^2)
    ,
\end{align}
while the system \eqref{eq:2ddToda_matsys_n} turns into
\begin{align}
    &\left\{
    \begin{array}{rcl}
         L \Psi
         &=& \partial \Psi
         ,
         \\[2mm]
         \Psi - \varepsilon \partial_y \Psi
         + O(\varepsilon^2)
         &=& M \Psi,
    \end{array}
    \right.
    &\Leftrightarrow
    &
    &\left\{
    \begin{array}{rcl}
         \brackets{
         \varepsilon^{-1} \mathbf{I}
         + \hat L
         + O(\varepsilon)
         } \Psi
         &=& \partial \Psi
         ,
         \\[2mm]
         \Psi - \varepsilon \partial_y \Psi
         + O(\varepsilon^2)
         &=& \brackets{
         \mathbf{I}
         + \varepsilon \hat M
         + O(\varepsilon^2)
         } \Psi.
    \end{array}
    \right.
\end{align}
After the shift $\partial \mapsto \partial - \varepsilon^{-1}$ (cf. with the continuous limit in the proof of Proposition \ref{thm:2dToda_matsys}) and the rescaling $\hat{M} \mapsto - \hat{M}$, we can take the limit $\varepsilon \to 0$ and, as a result, arrive at system \eqref{eq:2dToda_matsys_n} with matrices \eqref{eq:2dToda_pair_n} (the signs $\hat{\phantom{u}}$ are omitted).

\medskip
The system \eqref{eq:2dToda_matsys_n} implies the condition
\begin{align}
    \label{eq:lax_2d}
    \partial_y L
    - \partial M
    &= M L - L M
\end{align}
that is equivalent to \eqref{eq:2dToda_nc}.
\end{proof}

\begin{rem}
Let us stress again that for finite lattices the boundary conditions must be assumed.
\end{rem}

Due to the specialisation noticed in \eqref{eq:spec_2dd}, the continuous limit leads to quasideterminant solutions of \eqref{eq:2dToda_nc} described in the proposition below.
\begin{prop}
\label{thm:2dToda_sol}
Let us define the bi-directional Wronskian
\begin{align}
    \label{eq:Theta_def_2d}
    \Theta_{n} (x, y)
    := \brackets{
    \partial_{x^{i - 1}} \, \partial_{y^{j - 1}}
    \varphi(x, y)
    }_{1 \leq i, j \leq n}
\end{align}
and its quasideterminant $\theta_n = |\Theta_n|_{n,n}$. 
Then, $\theta_n = \theta_n (x, y)$ is a solution of \eqref{eq:2dToda_nc}.
\end{prop}
\begin{proof}
Consider the solution of \eqref{eq:2ddToda_nc} given by \eqref{eq:spec_2dd} with $\varepsilon_1 = \varepsilon_2 = \varepsilon$:
\begin{align}
    \varphi_{l + k_1, m + k_2}
    &= \varepsilon^{- k_1 - k_2}
    \sum_{j_1 = 0}^{k_1} \sum_{j_2 = 0}^{k_2} (-1)^{j_1 + j_2} 
    \binom{k_1}{j_1} \binom{k_2}{j_2}
    \phi_{l - k_1, m - k_2},
    &
    0 \leq k_1, k_2 \leq n - 1.
\end{align}
Recall that this expression is a definition of the discrete derivation. Then, after the change
\begin{align}
    \label{eq:spec_2d}
    &&
    \phi (x, y)
    &= \phi_{l, m},
    &&&
    x
    &= \varepsilon l,
    &
    y
    &= \varepsilon m
    &&
\end{align}
and the limit $\varepsilon \to 0$, the solution \eqref{eq:spec_2dd} takes the form
\begin{align}
    \varphi_{l + k_1, m + k_2}
    = \partial_{x^{k_1}} \partial_{y^{k_2}} \phi (x, y).
\end{align}
So, we arrive at the statement of the proposition.
\end{proof}

\section{Non-Abelian 1d Toda equations}
\label{sec:main12}

\subsection{Discrete case}
\label{sec:1ddToda}

The map
\begin{align}
    \label{eq:2ddto1dd}
    &&
    \theta_{l, m, n}
    &\mapsto \theta_{l + m, n},
    &
    T_{1, -2} \brackets{\psi_{l, m, n}}
    &= \lambda \psi_{l + m, n}
    &&
\end{align}
brings the 2d discrete Toda lattice \eqref{eq:2ddToda_nc} to the non-abelian 1d discrete Toda equation
\begin{align}
    \label{eq:1ddToda_nc}
    \theta_{m + 2, n}
    &= \theta_{m, n + 1}
    + \theta_{m + 1, n} \brackets{
    \theta_{l, m, n}^{-1}
    - \theta_{m + 2, n - 1}^{-1}
    } \theta_{m + 1, n}.
\end{align}
The commutative parameter $\lambda$ will play role of the spectral parameter. Notice that the equation \eqref{eq:1ddToda_nc} can be introduced in a similar way as the 2d discrete Toda equation \eqref{eq:2ddToda_nc} by using the non-abelian Jacobi identity \eqref{eq:Jacobi_nc} and homological relations \eqref{eq:homrel_col}.

\begin{prop}
\label{thm:1ddToda_scalsys}
The non-abelian 1d discrete Toda equation \eqref{eq:1ddToda_nc} follows from the system
\begin{gather}
    \label{eq:1ddToda_scalsys}
    \left\{
    \begin{array}{rcl}
         \lambda \psi_{m, n}
         &=& \psi_{m, n + 1}
         + (a_{m - 1, n} + b_{m, n}) \psi_{m, n}
         + a_{m - 1, n} b_{m, n - 1} \psi_{m, n - 1}
         ,
         \\[2mm]
         \psi_{m - 1, n}
         &=& \psi_{m, n}
         + b_{m, n - 1} \psi_{m, n - 1},
    \end{array}
    \right.
    \\[2mm]
    \label{eq:1ddToda_abdef}
    \begin{aligned}
    a_{m, n}
    &= \theta_{m + 1, n} \theta_{m, n}^{-1},
    &&&&&
    b_{m, n}
    &= \theta_{m - 1, n + 1} \theta_{m, n}^{-1}.
    \end{aligned}
\end{gather}
\end{prop}
\begin{proof}
The proof consists of two stages: reduction and verification of the compatibility condition.

\medskip
\textbf{\textbullet \,\, Reduction.}
The system \eqref{eq:2ddToda_scalsys} can be rewritten in the form
\begin{gather}
    \left\{
    \begin{array}{rcl}
         \psi_{l + 1, m - 1, n}
         &=& \psi_{l, m - 1, n + 1}
         + a_{l, m - 1, n} \psi_{l, m - 1, n},
         \\[2mm]
         \psi_{l, m - 1, n + 1}
         &=& \psi_{l, m, n + 1}
         + b_{l, m, n} \psi_{l, m, n}.
    \end{array}
    \right.
\end{gather}
By using \eqref{eq:2ddto1dd}, the system becomes
\begin{gather}
    \left\{
    \begin{array}{rcl}
         \lambda \psi_{m, n}
         &=& \psi_{m - 1, n + 1}
         + a_{m - 1, n} \psi_{m - 1, n},
         \\[2mm]
         \psi_{m - 1, n + 1}
         &=& \psi_{m, n + 1}
         + b_{m, n} \psi_{m, n},
    \end{array}
    \right.
\end{gather}
where $a_{m, n}$ and $b_{m, n}$ are defined by \eqref{eq:1ddToda_abdef}. Replacing $\psi_{m - 1, n + 1}$ and $\psi_{m - 1, n}$ with the corresponding expressions given by the second equation of the system, we arrive at system \eqref{eq:1ddToda_scalsys}.

\medskip
\textbf{\textbullet \,\, Compatibility condition.}
Compatibility condition of the spectral problem \eqref{eq:1ddToda_scalsys}
\begin{align}
    \label{eq:compcond_1dd}
    T_{-1} \brackets{
    \lambda \psi_{m, n}
    }
    &= \lambda \brackets{
    \psi_{m - 1, n}
    }
\end{align}
can be expanded as follows:
\begin{align}
    l.h.s.\eqref{eq:compcond_1dd}
    &\overset{\,\,\phantom{\eqref{eq:2dToda_scalsys}_2}}{\equiv}  
    \lambda \psi_{m - 1, n}
    \overset{\,\,\eqref{eq:1ddToda_scalsys}_1}{=}  \psi_{m - 1, n + 1}
    + \brackets{
    a_{m - 2, n} + b_{m - 1, n}
    } \psi_{m - 1, n}
    + a_{m - 2, n} b_{m - 1, n - 1} \psi_{m - 1, n - 1}
    \\[2mm]
    &\overset{\,\,\eqref{eq:1ddToda_scalsys}_2}{=} 
    \brackets{
    \psi_{m, n + 1} + b_{m, n} \psi_{m, n}
    }
    + \brackets{
    a_{m - 2, n} + b_{m - 1, n}
    } \brackets{
    \psi_{m, n} + b_{m, n - 1} \psi_{m, n - 1}
    }
    \\[1mm]
    &\hspace{6.6cm}
    + a_{m - 2, n} b_{m - 1, n - 1} \brackets{
    \psi_{m, n - 1} + b_{m, n - 2} \psi_{m, n - 2}
    },
    \\[2mm]
    r.h.s.\eqref{eq:compcond_1dd}
    &\overset{\,\,\phantom{\eqref{eq:2dToda_scalsys}_2}}{\equiv} \lambda \psi_{m - 1, n}
    \overset{\,\,\eqref{eq:1ddToda_scalsys}_2}{=}  
    \lambda \psi_{m, n}
    + b_{m, n - 1} \brackets{
    \lambda \psi_{m, n - 1}
    }
    \\[2mm]
    &\overset{\,\,\eqref{eq:1ddToda_scalsys}_1}{=}  
    \psi_{m, n + 1}
    + \brackets{
    a_{m - 1, n} + b_{m, n}
    } \psi_{m, n}
    + a_{m - 1, n} b_{m, n - 1} \psi_{m, n - 1}
    \\[1mm]
    &\hspace{2.3cm}
    + b_{m, n - 1} \brackets{
    \psi_{m, n}
    + \brackets{
    a_{m - 1, n - 1} + b_{m, n - 1}
    } \psi_{m, n - 1}
    + a_{m - 1, n - 1} b_{m, n - 2} \psi_{m, n - 2}
    }.
\end{align}
Since $l.h.s.\eqref{eq:compcond_1dd} - r.h.s.\eqref{eq:compcond_1dd} = 0$, we obtain the system for $a_{m, n}$ and $b_{m, n}$
\begin{align}
    &\left\{
    \begin{array}{rcl}
         a_{m - 2, n} + b_{m - 1, n}
         &=& a_{m - 1, n} + b_{m, n - 1},
         \\[2mm]
         a_{m - 2, n} b_{m - 1, n - 1}
         &=& b_{m, n - 1} a_{m - 1, n - 1},
         \\[2mm]
         a_{m - 2, n} b_{m, n - 1}
         + b_{m - 1, n} b_{m, n - 1}
         + a_{m - 2, n} b_{m - 1, n - 1}
         &=& 
         a_{m - 1, n} b_{m, n - 1}
         + b_{m, n - 1} a_{m - 1, n - 1}
         + b_{m, n - 1}^2,
    \end{array}
    \right.
\end{align}
that can be reduced to
\begin{align}
    \label{eq:1ddToda_scalsysab}
    &\left\{
    \begin{array}{rcl}
         a_{m - 1, n} + b_{m, n - 1}
         &=& a_{m - 2, n} + b_{m - 1, n}
         \\[2mm]
         b_{m, n - 1} a_{m - 1, n - 1}
         &=& a_{m - 2, n} b_{m - 1, n - 1}.
    \end{array}
    \right.
\end{align}
Substituting \eqref{eq:1ddToda_abdef} into \eqref{eq:1ddToda_scalsysab},
\begin{align}
    &\left\{
    \begin{array}{rcl}
         \theta_{m, n} \theta_{m - 1, n}^{-1}
         + \theta_{m - 1, n} \theta_{m, n - 1}^{-1}
         &=& \theta_{m - 1, n} \theta_{m - 2, n}^{-1}
         + \theta_{m - 2, n + 1} \theta_{m - 1, n}^{-1}
         \\[2mm]
         \theta_{m - 1, n} \theta_{m, n - 1}^{-1} \cdot
         \theta_{m, n - 1} \theta_{m - 1, n - 1}^{-1}
         &=& \theta_{m - 1, n} \theta_{m - 2, n}^{-1} \cdot
         \theta_{m - 2, n} \theta_{m - 1, n - 1}^{-1},
    \end{array}
    \right.
\end{align}
the second equation \eqref{eq:1ddToda_scalsysab}$_2$ becomes an identity, while \eqref{eq:1ddToda_scalsysab}$_1$ coincides with \eqref{eq:1ddToda_nc} shifted by $m \mapsto m - 2$.
\end{proof}

Similarly to Proposition \ref{thm:2ddToda_matsys}, one can introduce a two-component vector-function in order to rewrite system \eqref{eq:1ddToda_scalsys} in the matrix form. On the other hand, the same result can be achieved by reduction data \eqref{eq:2ddto1dd}. Thus, we arrive at the following
\begin{prop}
\label{thm:1ddToda_matsys}
Let functions $\theta_{m, n}$ be solutions of the non-abelian 1d discrete Toda lattice \eqref{eq:1ddToda_nc}. Then the following matrix linear system
\begin{align}
    \label{eq:1ddToda_matsys}
    &&
    &
    \left\{
    \begin{array}{rcl}
         \Psi_{m, n + 1}
         &=& \mathcal{L}_{m, n} \Psi_{m, n},  
         \\[2mm]
         \Psi_{m - 1, n}
         &=& \mathcal{M}_{m, n} \Psi_{m, n},  
    \end{array}
    \right.
    &
    \Psi_{m, n}
    &= 
    \begin{pmatrix}
        \psi_{m, n} & \psi_{m, n - 1}
    \end{pmatrix}^{T},
    &&
\end{align}
where $\mathcal{L}_{m, n}$ and $\mathcal{M}_{m, n}$ are $2\times2$ matrices given by
\begin{align}
    \label{eq:1ddToda_pair}
    \begin{aligned}
    \mathcal{L}_{m, n}
    &= 
    \begin{pmatrix}
        \lambda 
        - \theta_{m + 1, n} \theta_{m + 2, n - 1}^{-1}
        - \theta_{m + 2, n} \theta_{m + 1, n}^{-1}
        & 
        - \theta_{m + 1, n}
        \\[1mm]
        \theta_{m + 1, n}^{-1} 
        & 
        0
    \end{pmatrix},
    \\[3mm]
    \mathcal{M}_{m, n}
    &= 
    \begin{pmatrix}
        1
        & 
        \theta_{m, n}
        \\[1mm]
        - \theta_{m + 1, n - 1}^{-1} 
        & 
        \lambda - \theta_{m + 1, n - 1}^{-1} \theta_{m, n}
    \end{pmatrix}
    ,
    \end{aligned}
\end{align}
is compatible, i.e. a discrete analog of the zero-curvature condition holds{\rm:}
\begin{align}
    \label{eq:zrc_1dd}
    \mathcal{L}_{m - 1, n} 
    \mathcal{M}_{m, n}
    &= \mathcal{M}_{m, n + 1}
    \mathcal{L}_{m, n}\ .
\end{align}
\end{prop}
\begin{proof}
Matrices $\mathcal{L}_{m, n}$, $\mathcal{M}_{m, n}$ can be derived either straightforwardly from \eqref{eq:1ddToda_scalsys} (similarly to the proof of Proposition \ref{thm:2ddToda_matsys}) or just by considering the reduction data defined by \eqref{eq:2ddto1dd}. Expanding the compatibility condition \eqref{eq:zrc_1dd} of system \eqref{eq:1ddToda_matsys}, 
\begin{align}
    l.h.s.\eqref{eq:zrc_1dd}
    &- r.h.s.\eqref{eq:zrc_1dd}
    \\[2mm]
    &= 
    \begin{pmatrix}
        \lambda 
        - \theta_{m, n} \theta_{m + 1, n - 1}^{-1}
        - \theta_{m + 1, n} \theta_{m, n}^{-1}
        &
        - \theta_{m, n}
        \\[1mm]
        \theta_{m, n}^{-1}
        &
        0
    \end{pmatrix}
    \begin{pmatrix}
        1 & \theta_{m, n}
        \\[1mm]
        - \theta_{m + 1, n - 1}^{-1}
        &
        \lambda 
        - \theta_{m + 1, n - 1}^{-1} \theta_{m, n}
    \end{pmatrix}
    \\[1mm]
    &\qquad
    -
    \begin{pmatrix}
        1 & \theta_{m, n + 1}
        \\[1mm]
        - \theta_{m + 1, n}^{-1}
        &
        \lambda 
        - \theta_{m + 1, n}^{-1} \theta_{m, n + 1}
    \end{pmatrix}
    \begin{pmatrix}
        \lambda 
        - \theta_{m + 1, n} \theta_{m + 2, n - 1}^{-1}
        - \theta_{m + 2, n} \theta_{m + 1, n}^{-1}
        & 
        - \theta_{m + 1, n}
        \\[1mm]
        \theta_{m + 1, n}^{-1} 
        & 
        0
    \end{pmatrix}
    \\[3mm]
    &= 
    \begin{pmatrix}
        \lambda - \theta_{m + 1, n} \theta_{m, n}^{-1}
        &
        - \theta_{m + 1, n}
        \\[1mm]
        \theta_{m, n}^{-1}
        &
        1
    \end{pmatrix}
    \\[1mm]
    &\hspace{3.1cm}
    -
    \begin{pmatrix}
        \lambda 
        - \theta_{m + 1, n} \theta_{m + 2, n - 1}^{-1}
        - \theta_{m + 2, n} \theta_{m + 1, n}^{-1}
        + \theta_{m, n + 1} \theta_{m + 1, n}^{-1}
        &
        - \theta_{m + 1, n}
        \\[1mm]
        \theta_{m + 2, n - 1}^{-1}
        + \theta_{m + 1, n}^{-1} \theta_{m + 2, n} \theta_{m + 1, n}^{-1}
        - \theta_{m + 1, n}^{-1} \theta_{m, n + 1} \theta_{m + 1, n}^{-1}
        &
        1
    \end{pmatrix}
    ,
\end{align}
we arrive at the statement of the proposition.
\end{proof}

An additional benefit of the spectral problem \eqref{eq:1ddToda_scalsys} is that besides the two-component formalism one is also able to introduce semi-infinite matrices. As a result, the discrete Lax equation can be obtained. However, thanks to the reduction, the matrix problem \eqref{eq:2ddToda_matsys_n} -- \eqref{eq:2ddToda_pair_n} reduces to those presented in the proposition below.

\begin{prop}
\label{thm:1ddToda_matsys_n}
Let the matrices $L_{m}$ and $M_m$ be defined as
\begin{align}
    \label{eq:1ddToda_pair_n}
    \begin{aligned}
    L_m
    &= 
    \begin{pmatrix}
        a_{m - 1, 1} + b_{m, 1}
        &
        1
        &&&
        \\[1mm]
        a_{m - 1, 2} b_{m, 1}
        &
        a_{m - 1, 2} + b_{m, 2}
        &
        1
        &
        &
        \\[1mm]
        &
        a_{m - 1, 3} b_{m, 2}
        &
        a_{m - 1, 3} + b_{m, 3}
        &
        \ddots
        &
        \\[1mm]
        &
        &
        \ddots
        &
        \ddots
    \end{pmatrix}
    \\[1mm]
    &= \sum_{k \geq 1} E_{k, k + 1}
    + \sum_{k \geq 1} \brackets{
    a_{m - 1, k} + b_{m, k}
    } E_{k, k}
    + \sum_{k \geq 1} a_{m - 1, k + 1} b_{m, k} E_{k + 1, k}
    ,
    \\[3mm]
    M_m
    &= 
    \begin{pmatrix}
        1
        &
        &&&
        \\[1mm]
        b_{m, 1}
        &
        1
        &
        &
        &
        \\[1mm]
        &
        b_{m, 2}
        &
        1
        &
        &
        \\[1mm]
        &
        &
        \ddots
        &
        \ddots
    \end{pmatrix}
    = \sum_{k \geq 1} E_{k, k}
    + \sum_{k \geq 1} b_{m, k} E_{k + 1, k}
    .
    \end{aligned}
\end{align}
Then, the discrete matrix system
\begin{align}
    \label{eq:1ddToda_matsys_n}
    &\left\{
    \begin{array}{rcl}
         L_m \Psi_m 
         &=& \lambda \Psi_m,
         \\[2mm]
         \Psi_{m - 1}
         &=& M_m \Psi_m,
    \end{array}
    \right.
    &
    \Psi_m
    &= 
    \begin{pmatrix}
        \psi_{m, 1}
        &
        \psi_{m, 2}
        &
        \psi_{m, 3}
        &
        \dots
    \end{pmatrix}^T
\end{align}
yields the non-abelian semi-infinite 1d discrete Toda lattice \eqref{eq:1ddToda_nc}.
\end{prop}

\begin{proof}
The system and matrices reduce from their 2d analogs \eqref{eq:2ddToda_matsys_n} -- \eqref{eq:2ddToda_pair_n} given in Proposition \ref{thm:2ddToda_matsys_n} by using \eqref{eq:2ddto1dd}. 
The system \eqref{eq:1ddToda_matsys_n} leads to the discrete analog of the Lax equation
\begin{align}
    \label{eq:lax_1dd}
    L_{m - 1} M_m
    &= M_m L_m,
\end{align}
that is equivalent to the semi-infinite 1d discrete Toda equation \eqref{eq:1ddToda_nc}. 
\end{proof}

\begin{rem}
Thanks to Proposition \ref{thm:1ddToda_matsys_n}, it is easy to obtain non-abelian analogs of the 1d discrete Toda lattices with the open-end and periodic boundary conditions. 
\end{rem}

Finally, we are going to derive quasideterminant solutions for \eqref{eq:1ddToda_nc}. 

\begin{prop}
\label{thm:1ddToda_sol}
Let 
\begin{align}
    \label{eq:1ddToda_sol}
    \Theta_{m, n} = \brackets{\varphi_{m + i + j - 1}}_{1 \leq i, j \leq n}
\end{align}
and $\theta_{m, n} = \abs{\Theta_{m, n}}_{n,n}$. Then, 

\begin{itemize}
\item[\rm(a)]
$\theta_{m, n}$ is a solution of the non-abelian 1d discrete Toda equation \eqref{eq:1ddToda_nc} and, in particular,

\vspace{2mm}
\item[\rm(b)]
the functions $\varphi_{m + k}$ can be written as
\begin{align}
    \label{eq:spec_1dd}
    &&
    \varphi_{m + k}
    &= \varepsilon^{- k} 
    \sum_{j = 0}^k (-1)^{j}
    \binom{k}{j} \phi_{m - j},
    &
    0 
    &\leq k \leq 2 n - 2.
    &&
\end{align}
\end{itemize}
\end{prop}
\begin{proof}
The statement follows from the definition of $\theta_{l, m, n}$, formula \eqref{eq:spec_2dd}, and reduction data \eqref{eq:2ddto1dd}. 
\end{proof}

Note that in the commutative case there exists an infinite 1d discrete Toda lattice whose determinant solutions were discovered in the paper \cite{kajiwara1999determinant}. Regarding the semi-infinite case, a connection with a non-abelian analog of  Theorem 2.3 from \cite{kajiwara1999determinant} is given by the lemma below.
\begin{lem}
\label{thm:gen_1dd}
Let $c_{m, k}$, $k \geq 0$ be defined recursively as
\begin{align}
    \label{eq:rec_1dd}
    &&
    c_{m, k}
    &= \varepsilon^{-k} \brackets{
    c_{m, k - 1}
    - c_{m - 1, k - 1}
    },
    &
    c_{m, 0}
    &= \phi_m.
    &&
\end{align}
Then, $c_{m, k}$ are given by the general formula
\begin{align}
    \label{eq:gen_1dd}
    c_{m, k}
    &= \varepsilon^{-k}
    \sum_{j = 0}^k (-1)^{j} \binom{k}{j} \phi_{m - j}.
\end{align}
\end{lem}
\begin{proof}
One can verify the statement just by using the mathematical induction. For the base $k = 1$, the recurrence \eqref{eq:rec_1dd} leads to
\begin{align}
    c_{m, 0}
    &= \varepsilon^{-1} \brackets{
    c_{m, 0} - c_{m - 1, 0}
    }
    = \varepsilon^{-1} \brackets{
    \phi_m - \phi_{m - 1}
    }
\end{align}
and coincide with \eqref{eq:gen_1dd} with $k = 1$. 
Let the formula \eqref{eq:gen_1dd} be true for $k = n$. Then, for $k = n + 1$, we have
\begin{align}
    \varepsilon^{n + 1} \, c_{m, n + 1}
    &= \varepsilon^{n} \brackets{
    c_{m, n} - c_{m - 1, n}
    }
    = \sum_{j = 0}^n (-1)^j \binom{n}{j} \phi_{m - j}
    - \sum_{j = 0}^n (-1)^j \binom{n}{j} \phi_{m - 1 - j}
    \\[1mm]
    &= \sum_{j = 0}^n (-1)^j \binom{n}{j} \phi_{m - j}
    - \sum_{j = 1}^{n + 1} 
    (-1)^{j - 1} \binom{n}{j - 1} \phi_{m - j}
    \\[1mm]
    &= \phi_{m} 
    + (-1)^{n + 1} \phi_{m - (n + 1)}
    + \sum_{j = 1}^n (-1)^j 
    \brackets{
    \binom{n}{j}
    + \binom{n}{j - 1}
    } \phi_{m - j}
    \\[1mm]
    &= \phi_{m} 
    + (-1)^{n + 1} \phi_{m - (n + 1)}
    + \sum_{j = 1}^n (-1)^j 
    \binom{n + 1}{j} \phi_{m - j}
    = \sum_{j = 1}^{n + 1} (-1)^j 
    \binom{n + 1}{j} \phi_{m - j},
\end{align}
where the binomial recurrence $\displaystyle \binom{n}{j} + \binom{n}{j - 1} = \binom{n + 1}{j}$ was used.
\end{proof}

\subsection{Continuous case}
\label{sec:1dToda}

A continuous 1d Toda lattice can be obtained either by a continuous limit of its discrete analog or by a reduction of its 2d generalization. Below we will consider the second method of deriving a continuous Toda lattice and in the end of the section will discuss the first method. 

Let us consider the following reduction data for the non-abelian 2d Toda field equation \eqref{eq:2dToda_nc}:
\begin{align}
    \label{eq:2dto1d}
    &&
    \theta_n (t)
    &= \theta_n (x + y),
    &
    x
    &= \tfrac12 (t + z),
    &
    y
    &= \tfrac12 (t - z).
    &&
\end{align}
Then, the operators $\partial_x + \partial_y$ and $\partial_x - \partial_y$ become $\partial_t$ and $\partial_z$, respectively. 
If we change the coordinates to $(t, z)$, then $\partial_z \theta_n = 0$ and the left hand side of the equation is equal to $\partial_t (\theta_{n}' \theta_n^{-1})$.
On the other hand, it is easy to see that we will get the same formula, if we put $\partial_z \theta_n = \lambda \, \theta_n$, $\lambda \neq 0$. Regarding this reduction, the equation \eqref{eq:2dToda_nc} turns into the non-abelian 1d Toda equation 
\begin{align}
    \label{eq:1dToda_nc}
    \brackets{
    \theta_{n}' \theta_n^{-1}
    }'
    &= \theta_{n + 1} \theta_n^{-1}
    - \theta_n \theta_{n - 1}^{-1}
    .
\end{align}

Similarly to Proposition \ref{thm:1ddToda_scalsys}, we obtain a scalar spectral problem for \eqref{eq:1dToda_nc}.
\begin{prop}
The non-abelian 1d Toda equation \eqref{eq:1dToda_nc} follows from the system
\begin{gather}
    \label{eq:1dToda_scalsys}
    \left\{
    \begin{array}{rcl}
         \lambda \psi_{n}
         &=& \psi_{n + 1}
         + a_n \psi_{n}
         + b_n \psi_{n - 1}
         ,
         \\[2mm]
         \psi_{n}'
         &=& - b_n \psi_{n - 1},
    \end{array}
    \right.
    \\[2mm]
    \label{eq:1dToda_abdef}
    \begin{aligned}
    a_{n}
    &= \theta_{n}' \theta_{n}^{-1},
    &&&&&
    b_{n}
    &= \theta_{n} \theta_{n - 1}^{-1}.
    \end{aligned}
\end{gather}
\end{prop}
\begin{proof}
Below we again split the proof into two stages. We first consider a reduction and then verify the compatibility condition for system \eqref{eq:1dToda_scalsys}.

\medskip
\textbf{\textbullet \,\, Reduction.}
Note that after the change \eqref{eq:2dto1d}
the operator $\partial \equiv \partial_x - \partial_y$ becomes just $\lambda$. Consider the following system that is equivalent to \eqref{eq:2dToda_scalsys}:
\begin{gather}
    \left\{
    \begin{array}{rcl}
         \partial (\psi_{n})
         &=& \psi_{n + 1}
         + \theta_{n, x} \theta_n^{-1} 
         \psi_{n}
         + \theta_n \theta_{n - 1}^{-1} \psi_{n - 1}
         ,
         \\[2mm]
         \psi_{n, y}
         &=& - \theta_{n} \theta_{n - 1}^{-1} \psi_{n - 1}.
    \end{array}
    \right.
\end{gather}
Then, after the reduction \eqref{eq:2dto1d} supplemented by the change
\begin{align}
    \label{eq:2dto1d_psi}
    \psi_n (t)
    &= \psi_n (x + y),
\end{align}
it becomes \eqref{eq:1dToda_scalsys}.

\medskip
\textbf{\textbullet \,\, Compatibility condition.}
The system \eqref{eq:1dToda_scalsys} is compatible if
\begin{align}
    \label{eq:comp_cond_1d}
    \brackets{\lambda \psi_n}'
    &= \lambda \brackets{\psi_n'}.
\end{align}
Expanding this condition, we obtain the chains of identities:
\begin{align}
    l.h.s.\eqref{eq:comp_cond_1d}
    &\overset{\,\,\phantom{\eqref{eq:1dToda_scalsys}_1}}{\equiv}  
    \lambda \psi_{n}'
    \overset{\,\,\eqref{eq:1dToda_scalsys}_1}{=}
    \psi_{n + 1}' 
    + a_n' \psi_n
    + a_n \psi_n'
    + b_n' \psi_{n - 1}
    + b_n \psi_{n - 1}'
    \\[2mm]
    &\overset{\,\,\eqref{eq:1dToda_scalsys}_2}{=}
    - b_{n + 1} \psi_n
    + a_n' \psi_n
    - a_n b_n \psi_{n - 1}
    + b_n' \psi_{n - 1}
    - b_n b_{n - 1} \psi_{n - 2},
    \\[3mm]
    r.h.s.\eqref{eq:comp_cond_1d}
    &\overset{\,\,\phantom{\eqref{eq:1dToda_scalsys}_1}}{\equiv}  
    \lambda \psi_{n}'
    \overset{\,\,\eqref{eq:1dToda_scalsys}_2}{=}
    - b_n \brackets{
    \lambda \psi_{n - 1}
    }
    \overset{\,\,\eqref{eq:1dToda_scalsys}_1}{=}
    - b_n \brackets{
    \psi_n + a_{n - 1} \psi_{n - 1} + b_{n - 1} \psi_{n - 2}
    }.
\end{align}
Hence, the condition \eqref{eq:comp_cond_1d} leads to the system for $a_n$ and $b_n$ only: 
\begin{align}
    \label{eq:1dToda_scalsysab}
    \left\{
    \begin{array}{rcl}
         a_n'
         &=& b_{n + 1} - b_n
         ,
         \\[2mm]
         b_n'
         &=& a_n b_n - b_n a_{n - 1},
    \end{array}
    \right.
\end{align}
or, after the substitution of \eqref{eq:1dToda_abdef} into \eqref{eq:1dToda_scalsysab}, we have
\begin{align}
    \left\{
    \begin{array}{rcl}
         \brackets{\theta_n' \theta_n^{-1}}'
         &=& \theta_{n + 1} \theta_{n}^{-1} 
         - \theta_{n} \theta_{n - 1}^{-1}
         ,
         \\[2mm]
         \brackets{\theta_{n} \theta_{n - 1}^{-1}}'
         &=& \theta_n' \theta_n^{-1} \cdot
         \theta_{n} \theta_{n - 1}^{-1} 
         - \theta_{n} \theta_{n - 1}^{-1} \cdot 
         \theta_{n - 1}' \theta_{n - 1}^{-1}.
    \end{array}
    \right.
\end{align}
The equations \eqref{eq:1dToda_scalsysab}$_1$ and \eqref{eq:1dToda_scalsysab}$_2$ give rise to \eqref{eq:1dToda_nc} and the identity, respectively.
\end{proof}

As in the previous subsection, the scalar system \eqref{eq:1dToda_scalsys} can be rewritten in the matrix form by introducing either a two-component vector-function or a semi-infinite vector-function. In the first case, one can obtain the following
\begin{prop}
\label{thm:1dToda_matsys}
Let the matrices $\mathcal{L}_n$ and $\mathcal{M}_n$ be defined by 
\begin{align}
    \label{eq:1dToda_pair}
    \mathcal{L}_{n}
    &= 
    \begin{pmatrix}
        \lambda 
        - \theta_n' \theta_n^{-1}
        & 
        - \theta_{n}
        \\[1mm]
        \theta_{n}^{-1} 
        & 
        0
    \end{pmatrix},
    &
    \mathcal{M}_{n}
    &= 
    \begin{pmatrix}
        0
        & 
        - \theta_{n}
        \\[1mm]
        \theta_{n - 1}^{-1} 
        & 
        - \lambda
    \end{pmatrix}
    .
\end{align}
Then, the matrix linear system
\begin{align}
    \label{eq:1dToda_matsys}
    &&
    &
    \left\{
    \begin{array}{rcl}
         \Psi_{n + 1}
         &=& \mathcal{L}_{n} \Psi_{n},  
         \\[2mm]
         \Psi_{n}'
         &=& \mathcal{M}_{n} \Psi_{n},  
    \end{array}
    \right.
    &
    \Psi_{n}
    &= 
    \begin{pmatrix}
        \psi_{n} & \psi_{n - 1}
    \end{pmatrix}^{T}
    &&
\end{align}
is compatible, that is
\begin{align}
    \label{eq:zrc_1d}
    \mathcal{L}_n'
    &= \mathcal{M}_{n + 1}
    \mathcal{L}_{n}
    - \mathcal{L}_{n} 
    \mathcal{M}_{n}
    ,
\end{align}
if the function $\theta_{n}$ solves the non-abelian 1d Toda equation \eqref{eq:1dToda_nc}.
\end{prop}
\begin{proof}
The proof can be done either by a straightforward rewriting of the scalar system in the matrix form or by reduction data \eqref{eq:2dto1d}, \eqref{eq:2dto1d_psi} of the matrix system \eqref{eq:2dToda_matsys}, \eqref{eq:2dToda_pair}. The compatibility condition \eqref{eq:zrc_1d} reads as
\begin{align}
    l.h.s.\eqref{eq:zrc_1d}
    - r.h.s.\eqref{eq:zrc_1d}
    &= \brackets{
    - \brackets{\theta_n' \theta_n^{-1}}'
    + \theta_{n + 1} \theta_n^{-1} 
    - \theta_n \theta_{n - 1}^{-1}
    } E_{1,1}
\end{align}
and is satisfied, if $\theta_n$ is a solution of \eqref{eq:1dToda_nc}.
\end{proof}

Regarding the semi-infinite vector function, we arrive at
\begin{prop}
\label{thm:1dToda_matsys_n}
Consider the matrices $L$ and $M$ given by
\begin{align}
    \label{eq:1dToda_pair_n}
    L
    &= 
    \begin{pmatrix}
        a_{1}
        &
        1
        &&&
        \\[1mm]
        b_{1}
        &
        a_{2}
        &
        1
        &
        &
        \\[1mm]
        &
        b_{2}
        &
        a_{3}
        &
        \ddots
        &
        \\[1mm]
        &
        &
        \ddots
        &
        \ddots
    \end{pmatrix},
    &
    M
    &= 
    \begin{pmatrix}
        0
        &
        &&&
        \\[1mm]
        - b_{1}
        &
        0
        &
        &
        &
        \\[1mm]
        &
        - b_{2}
        &
        0
        &
        &
        \\[1mm]
        &
        &
        \ddots
        &
        \ddots
    \end{pmatrix},
\end{align}
where unfilled matrix entries are equal to zero. Then, the matrix system
\begin{align}
    \label{eq:1dToda_matsys_n}
    &\left\{
    \begin{array}{rcl}
         L \Psi
         &=& \lambda \Psi,
         \\[2mm]
         \Psi'
         &=& M \Psi,
    \end{array}
    \right.
    &
    \Psi
    &= 
    \begin{pmatrix}
        \psi_1
        &
        \psi_2
        &
        \psi_3
        &
        \dots
    \end{pmatrix}^T
\end{align}
implies the non-abelian 1d Toda equation \eqref{eq:1dToda_nc}.
\end{prop}

\begin{proof}
All formulas from the proposition follow from the corresponding formulas given in Proposition \ref{thm:2dToda_matsys_n} to which the reduction \eqref{eq:2dto1d} was applied. 
From the system \eqref{eq:1dToda_matsys_n} the Lax equation can be derived
\begin{align}
    \label{eq:lax_1d}
    L'
    &= M L - L M.
\end{align}
The latter is equivalent to the semi-infinite 1d Toda lattice \eqref{eq:1dToda_nc}. 
\end{proof}

\begin{rem}
It is obvious that the results from Proposition \ref{thm:1dToda_matsys_n} can be written for the 1d Toda lattices with the open-end or periodic boundary conditions. 
\end{rem}

Finally, we derive the quasideterminant solutions of \eqref{eq:1dToda_nc}.
\begin{prop}
\label{thm:1dToda_sol}
Let $\Theta_n = \Theta_n (t)$ be a Hankel matrix of the form
\begin{align}
    \label{eq:1dToda_sol}
    \Theta_{n} (t)
    := \brackets{
    \partial_{t^{i + j - 2}}
    \varphi(t)
    }_{1 \leq i, j \leq n}
\end{align}
and $\theta_n := |\Theta_{n}|_{n,n}$ be its $(n, n)$ quasideterminant. Then, $\theta_n$ is a solution of the non-abelian 1d Toda lattice~\eqref{eq:1dToda_nc}.
\end{prop}
\begin{proof}
The statement can be obtained from Proposition \ref{thm:2dToda_sol} just by the reduction data \eqref{eq:2dto1d} supplemented with the formula
\begin{align}
    \varphi_n(t)
    &= \varphi_n (x + y).
\end{align}
\end{proof}

To conclude this subsection, let us answer the following natural question. Does there exist a continuous limit that brings all structures of the 1d discrete Toda lattice \eqref{eq:1ddToda_nc} considered in Subsection \ref{sec:1ddToda} to the structures of the 1d Toda equation \eqref{eq:1dToda_nc} obtained in this subsection by a reduction of the 2d Toda field equation \eqref{eq:2dToda_nc}? The answer is given in the proposition below.

\begin{prop}
\label{thm:1ddto1d}
The continuous limit
\begin{align}
    \label{eq:1ddto1d}
    \theta_n (t)
    &= \varepsilon^{-m} \theta_{m, n}
    &
    \psi_n (t)
    &= \psi_{m, n}
    &
    \varphi (t)
    &= \varphi_{m},
    &
    t
    &= \varepsilon m
\end{align}
brings \eqref{eq:1ddToda_scalsys}, \eqref{eq:1ddToda_matsys} -- \eqref{eq:1ddToda_pair}, \eqref{eq:1ddToda_matsys_n} -- \eqref{eq:1ddToda_pair_n}, and \eqref{eq:1ddToda_sol} to \eqref{eq:1dToda_scalsys}, \eqref{eq:1dToda_matsys} -- \eqref{eq:1dToda_pair}, \eqref{eq:1dToda_matsys_n} -- \eqref{eq:1dToda_pair_n}, and \eqref{eq:1dToda_sol}, respectively.
\end{prop}

\begin{proof}
We will skip the detailed proofs for the scalar system, matrix $2\times2$ system and matrix semi-infinite system, since they are similar to the corresponding proofs that have already discussed in Subsection \ref{sec:2dToda}. Regarding the quasideterminant solutions, recall that the solution \eqref{eq:spec_1dd} is just a definition of the discrete derivative which under the suggested limit becomes the continuous derivative. Therefore, we obtain
\begin{align}
    \label{eq:spec_1d}
    &&
    \varphi_{m}
    &:= \phi^{(m)},
    &
    0 
    &\leq m \leq 2 n - 2
    &&
\end{align}
and the proof is done.
\end{proof}

\section{Reductions to some (1+1) discrete systems}
\label{sec:2ddTLtosys}

It is well-known that the two-dimensional discrete Toda equation \eqref{eq:2ddToda} can be regarded as the most general discrete equation which reduces to other integrable discrete systems. Note that after reduction some parameters in \eqref{eq:2ddToda_scalsys_al} cannot be removed by a scaling and, thus, the spectral parameter arises. Following this idea, we dedicate the section to different reductions of the non-abelian 2d discrete Toda equation \eqref{eq:2ddToda_nc}. As a result, we derive non-abelian analogs of discrete $(1 + 1)$-systems such as the discrete Boussinesq, Korteweg-de Vries, modified Korteweg-de Vries, and sine-Gordon equations. In Subsections \ref{sec:dB} -- \ref{sec:dmKdV}, the non-abelian analogs and their spectral problems are discussed. The reduction data are sometimes a little bit complicated and, thus, for a comparison and in addition, we give formulas for the Hirota equation \eqref{eq:Hirota}. At~the end of this section, we have discussed Miura transformations for some of the obtained systems.

\subsection{KdV-like reductions}
\label{sec:kdvred}

\subsubsection{The KdV equation}
\label{sec:dKdV}

The famous discrete Korteweg-de Vries lattice (Example 1 from~\cite{date1983method})
\begin{align}
    \label{eq:dKdV}
    \tau_{m + 1, n} \tau_{m - 1, n + 1}
    - \tau_{m - 1, n} \tau_{m + 1, n + 1}
    + \tau_{m, n} \tau_{m, n + 1}
    &= 0
\end{align}
is a result of the following reduction of the Hirota equation
\begin{align}
    (l, m, n)
    &\mapsto \left(
    m - l, n
    \right)
    = (
    \tilde{m}, \tilde{n}
    ),
\end{align}
where we omit the sign $\tilde{\phantom{\tau}}$. Let us combine this change with \eqref{eq:2ddTodatoHirota}:
\begin{align}
    \label{eq:2ddTodatodKdV}
    (l, m, n)
    &\mapsto (- m - l + 1, n + m)
    = (\tilde{m} + 1, \tilde{n}).
\end{align}
The mapping \eqref{eq:2ddTodatodKdV} brings the non-abelian 2d Toda equation \eqref{eq:2ddToda_nc} to the form
\begin{align}
    \theta_{m + 1, n + 1} \theta_{m, n + 1}^{-1}
    - \theta_{m, n}
    \theta_{m - 1, n}^{-1}
    &= \theta_{m - 1, n + 1} \theta_{m, n + 1}^{-1}
    - \theta_{m, n} \theta_{m + 1, n}^{-1}
\end{align}
Note that one can introduce new variable
$v_{m, n} = \theta_{m, n} \theta_{m - 1, n}^{-1}$ and a final form of the non-abelian Korteweg-de Vries equation reads as
\begin{align}
    \label{eq:dKdV_nc}
    v_{m + 1, n + 1}
    - v_{m, n}
    &= v_{m, n + 1}^{-1}
    - v_{m + 1, n}^{-1}.
\end{align}
One can derive the spectral problem for equation \eqref{eq:dKdV_nc} from the scalar Lax pair for \eqref{eq:2ddToda_nc}. In order to do it, we first make the scaling
$\psi_{l, m, n} \mapsto \lambda^{-l} \psi_{l, m, n}$ to introduce the spectral parameter $\lambda$ and, then, consider the reduction \eqref{eq:2ddTodatodKdV}. A resulting spectral problem for \eqref{eq:dKdV_nc} is
\begin{align}
    \label{eq:dKdV_scalsys}
    \left\{
    \begin{array}{lcl}
         \lambda \psi_{m, n}
         &=& \psi_{m + 1, n + 1} 
         + v_{m + 1, n}^{-1} \psi_{m + 1, n},
         \\[2mm]
         \psi_{m + 1, n}
         &=& \psi_{m, n + 1}
         + v_{m + 1, n} \psi_{m, n},
    \end{array}
    \right.
\end{align}
i.e. its compatibility condition gives rise to \eqref{eq:dKdV_nc}.

The system \eqref{eq:dKdV_scalsys} can be rewritten in the matrix form. Namely, following the method already discussed in the proof of Proposition \ref{thm:2ddToda_matsys}, we arrive at
\begin{prop}
The non-abelian discrete KdV equation \eqref{eq:dKdV_nc} is a consequence of the matrix linear problem
\begin{align}
    \label{eq:dKdV_matsys}
    &&
    &
    \left\{
    \begin{array}{rcl}
         \Psi_{m, n + 1}
         &=& \mathcal{L}_{m, n} \Psi_{m, n},  
         \\[2mm]
         \Psi_{m + 1, n}
         &=& \mathcal{M}_{m, n} \Psi_{m, n},  
    \end{array}
    \right.
    &
    \Psi_{m, n}
    &= 
    \begin{pmatrix}
        \psi_{m, n} & \psi_{m - 1, n}
    \end{pmatrix}^{T},
    &&
\end{align}
where $\mathcal{L}_{m, n}$ and $\mathcal{M}_{m, n}$ are $2\times2$ matrices given by
\begin{align}
    \label{eq:dKdV_pair}
    &&
    \mathcal{L}_{m, n}
    &= 
    \begin{pmatrix}
        - v_{m, n}^{-1}
        & 
        \lambda
        \\[1mm]
        1 
        & 
        - v_{m, n}
    \end{pmatrix},
    &
    \mathcal{M}_{m, n}
    &= 
    \begin{pmatrix}
        v_{m + 1, n}
        - v_{m, n}^{-1}
        & 
        \lambda
        \\[1mm]
        1
        & 
        0
    \end{pmatrix}
    .
    &&
\end{align}
\end{prop}
\begin{proof}
Regarding two-component formalism for system \eqref{eq:dKdV_scalsys}:
\begin{align}
    \left\{
    \begin{array}{lcl}
         \psi_{m + 1, n + 1}
         &=& - v_{m + 1, n}^{-1} \psi_{m + 1, n}
         + \lambda \psi_{m, n},
         \\[2mm]
         \psi_{m, n + 1}
         &=& \psi_{m + 1, n}
         - v_{m + 1, n} \psi_{m, n},
    \end{array}
    \right.
\end{align}
it is easy to get components for the vector-functions $\Psi_{m, n + 1}$, $\Psi_{m + 1, n}$. Indeed, we have the following chain of identities:
\begin{align}
    \psi_{m, n + 1}
    &= 
    \begin{pmatrix}
    - v_{m, n}^{-1} & \lambda    
    \end{pmatrix}
    \begin{pmatrix}
    \psi_{m, n} & \psi_{m - 1, n}
    \end{pmatrix}^{T},
    \\[2mm]
    \psi_{m - 1, n + 1}
    &= 
    \begin{pmatrix}
    1 & - v_{m, n}    
    \end{pmatrix}
    \begin{pmatrix}
    \psi_{m, n} & \psi_{m - 1, n}
    \end{pmatrix}^{T},
    \\[2mm]
    \psi_{m + 1, n}
    &= \psi_{m, n + 1}
    + v_{m + 1, n} \psi_{m, n}
    = \brackets{
    \lambda \psi_{m - 1, n} 
    - v_{m, n}^{-1} \psi_{m, n}
    }
    + v_{m + 1, n} \psi_{m, n}
    \\
    &=
    \begin{pmatrix}
    v_{m + 1, n} - v_{m, n}^{-1} 
    & \lambda    
    \end{pmatrix}
    \begin{pmatrix}
    \psi_{m, n} & \psi_{m - 1, n}
    \end{pmatrix}^{T},
    \\[2mm]
    \psi_{m, n}
    &= 
    \begin{pmatrix}
    1 & 0
    \end{pmatrix}
    \begin{pmatrix}
    \psi_{m, n} & \psi_{m - 1, n}
    \end{pmatrix}^{T},
\end{align}
that form the matrices $\mathcal{L}_{m, n}$, $\mathcal{M}_{m, n}$. 
The compatibility condition
\begin{align}
    \mathcal{L}_{m + 1, n} \mathcal{M}_{m, n}
    &= \mathcal{M}_{m, n + 1} \mathcal{L}_{m, n}
\end{align}
of system \eqref{eq:dKdV_matsys} leads to \eqref{eq:dKdV_nc}.
\end{proof}

\subsubsection{Boussinesq equation}
\label{sec:dB}

In this subsection we obtain a non-commutative analog of the hexagonal Boussinesq equation \cite{date1983method}, \cite{hietarinta2016discrete}
\begin{align}
    \label{eq:dB}
    \tau_{m + 1, n} \tau_{m - 1, n}
    - \tau_{m, n + 1} \tau_{m, n - 1}
    + \tau_{m + 1, n + 1} \tau_{m - 1, n - 1}
    &= 0.
\end{align}
The hexagonal Boussinesq equation \eqref{eq:dB} is a result of the following change of variables in \eqref{eq:Hirota}
\begin{align}
    (l, m, n)
    &\mapsto
    \brackets{
    l - n,
    m - n
    }
    = (\tilde{m}, \tilde{n}), 
\end{align}
where we omitted the sign $\tilde{\phantom{\tau}}$. Combining this map with \eqref{eq:2ddTodatoHirota}, i.e. making the given transformation in~\eqref{eq:2ddToda_nc}
\begin{align}
    \label{eq:2ddTodatodB}
    (l, m, n)
    &\mapsto
    \brackets{
    l - (n + m),
    - m - (n + m) + 1
    }
    = (\tilde{m}, \tilde{n} + 1), 
\end{align}
we arrive to a non-abelian discrete Boussinesq equation
\begin{align}
    \label{eq:dB_nc}
    \theta_{m, n - 1} \theta_{m - 1, n - 1}^{-1}
    - \theta_{m - 1, n} \theta_{m - 1, n - 1}^{-1}
    &= \theta_{m + 1, n + 1} \theta_{m, n + 1}^{-1}
    - \theta_{m + 1, n + 1} \theta_{m + 1, n}^{-1}
    .
\end{align}
To obtain a spectral problem for \eqref{eq:dB_nc}, we first make the scaling
$\psi_{l, m, n} \mapsto \lambda^{m - l} \psi_{l, m, n}$
in \eqref{eq:2ddToda_scalsys} to implement the spectral parameter $\lambda$ and then consider reduction data \eqref{eq:2ddTodatodB}. Thus, the spectral problem for \eqref{eq:dB_nc} reads
\begin{align}
    \label{eq:dB_scalsys}
    &
    \left\{
    \begin{array}{lcl}
         \lambda \psi_{m + 1, n}
         &=& \psi_{m - 1, n - 1}
         + \theta_{m + 1, n} \theta_{m, n}^{-1}
         \, \psi_{m, n}
         ,
         \\[2mm]
         \lambda \psi_{m, n + 1}
         &=& \psi_{m - 1, n - 1}
         + \theta_{m, n + 1} \theta_{m, n}^{-1}
         \, \psi_{m, n}
         .
    \end{array}
    \right.
\end{align}
The compatibility condition gives rise to the non-abelian discrete Boussinesq equation \eqref{eq:dB_nc}. The obtained equation and system generalize well-known result from \cite{date1983method} (Example 4 therein) to the non-abelian case.

\subsection{Periodic reductions}
\label{sec:perred}

\subsubsection{Sine-Gordon equation}
\label{sec:dsG}

The sine-Gordon equation \cite{date1983method}
\begin{align}
    \label{eq:dsG}
    \left\{
    \begin{array}{rcl}
         \tau_{l + 1, m, 0} \tau_{l, m + 1, 1}
         - \tau_{l, m + 1, 0} \tau_{l + 1, m, 1}
         + \tau_{l + 1, m + 1, 0} \tau_{l, m, 1}
         &=& 0,
         \\[2mm]
         \tau_{l + 1, m, 1} \tau_{l, m + 1, 0}
         - \tau_{l, m + 1, 1} \tau_{l + 1, m, 0}
         + \tau_{l + 1, m + 1, 1} \tau_{l, m, 0}
         &=& 0
    \end{array}
    \right.
\end{align}
is the periodic $(n \mod 2)$ version of the discrete Hirota equation \eqref{eq:Hirota}. 
Note that these equations can be rewritten in terms of the function $\theta_{l, m, n} = \tau_{l, m, n} \tau_{l, m, n - 1}^{-1}$ and, then, the sine-Gordon equation takes the form
\begin{align}
    \label{eq:dsG_th}
    \theta_{l + 1, m} \theta_{l, m}^{-1}
    -  \theta_{l + 1, m + 1} \theta_{l, m + 1}^{-1}
    &= \theta_{l, m}^{-1} \theta_{l, m + 1}^{-1} 
    - \theta_{l + 1, m} \theta_{l + 1, m + 1},
\end{align}
since 
\begin{align}
    &&
    \theta_{l, m}
    := \theta_{l, m, 1}
    &= \tau_{l, m, 1} \tau_{l, m, 0}^{-1},
    &
    \theta_{l, m, 2}
    &= \tau_{l, m, 0} \tau_{l, m, 1}^{-1}
    = \theta_{l, m}^{-1},
    &
    \theta_{l, m, 3}
    &= \tau_{l, m, 1} \tau_{l, m, 0}^{-1}
    = \theta_{l, m},
    &&
\end{align}
and so on. 

Let us turn to the periodic reduction $(n \mod 2)$ for \eqref{eq:2ddToda_nc}. Then, as in the commutative case, for even and odd $n$ we consider $\theta_{l, m}$ and $\theta_{l, m}^{-1}$, respectively. Therefore, the 2d discrete Toda equation \eqref{eq:2ddToda_nc} becomes
\begin{align}
    \label{eq:dsG_nc}
    \theta_{l + 1, m + 1}
    \theta_{l, m + 1}^{-1}
    - \theta_{l, m}^{-1} \theta_{l, m + 1}^{-1}
    &= \theta_{l + 1, m} \theta_{l, m}^{-1}
    - \theta_{l + 1, m}
    \theta_{l + 1, m + 1}.
\end{align}
In order to derive a scalar linear problem with the spectral parameter $\lambda$, we first rescale the auxiliary function $\psi_{l, m, n}$ as $\psi_{l, m, n} \mapsto \lambda^{-n} \psi_{l, m, n}$ and, then, make the periodic reduction. The system \eqref{eq:2ddToda_scalsys} splits into two systems for $n = 0$ and $n = 1$, respectively:
\begin{align}
    \label{eq:dsG_scalsys}
    &
    \left\{
    \begin{array}{lcl}
         \psi_{l + 1, m, 0}
         &=& \lambda \psi_{l, m, 1}
         + \theta_{l + 1, m} \theta_{l, m}^{-1} \psi_{l, m, 0},
         \\[2mm]
         \psi_{l, m - 1, 1}
         &=& \psi_{l, m, 1}
         + \lambda^{-1} \theta_{l, m - 1}^{-1} \theta_{l, m}^{-1} \psi_{l, m, 0}
         ;
    \end{array}
    \right.
    &
    &
    \left\{
    \begin{array}{lcl}
         \psi_{l + 1, m, 1}
         &=& \lambda \psi_{l, m, 0}
         + \theta_{l + 1, m}^{-1} \theta_{l, m} \psi_{l, m, 1},
         \\[2mm]
         \psi_{l, m - 1, 0}
         &=& \psi_{l, m, 0}
         + \lambda^{-1} \theta_{l, m - 1} \theta_{l, m}
         \psi_{l, m, 1}.
    \end{array}
    \right.
\end{align}
Introducing the vector-function $\Psi_{l, m} = \begin{pmatrix} \psi_{l, m, 0} & \psi_{l, m, 1} \end{pmatrix}^T$, they can be easily rewritten in the matrix form
\begin{align}
    \label{eq:dsG_matsys}
    &
    \left\{
    \begin{array}{lcl}
         \Psi_{l, m - 1}
         &=& \mathcal{L}_{l, m}
         \Psi_{l, m},
         \\[2mm]
         \Psi_{l + 1, m}
         &=& \mathcal{M}_{l, m}
         \Psi_{l, m}
         ,
    \end{array}
    \right.
    &
    \Psi_{l, m} 
    &= \begin{pmatrix} \psi_{l, m, 0} & \psi_{l, m, 1} \end{pmatrix}^T
    ,
\end{align}
with $2\times2$ matrices $\mathcal{L}_{l, m}$ and $\mathcal{M}_{l, m}$ given by
\begin{align}
    \label{eq:dsG_pair}
    \mathcal{L}_{l, m}
    &= 
    \begin{pmatrix}
        1 
        & 
        \lambda^{-1} 
        \theta_{l, m - 1} \theta_{l, m}
        \\[0.9mm]
        \lambda^{-1}
        \theta_{l, m - 1}^{-1}
        \theta_{l, m}^{-1}
        &
        1
    \end{pmatrix}
    ,
    &
    \mathcal{M}_{l, m}
    &= 
    \begin{pmatrix}
        \theta_{l + 1, m}
        \theta_{l, m}^{-1}
        &
        \lambda
        \\[0.9mm]
        \lambda
        &
        \theta_{l + 1, m}^{-1}
        \theta_{l, m}
    \end{pmatrix}
    .
\end{align}
The compatibility condition
\begin{align}
    \label{eq:comp_dsG}
    \mathcal{L}_{l + 1, m} \mathcal{M}_{l, m}
    &= \mathcal{M}_{l, m - 1} \mathcal{L}_{l, m}
\end{align}
leads to the non-abelian sine-Gordon equation \eqref{eq:dsG_nc}. Under the commutative reduction, the resulting equation and system coincide with those presented in Example 2 from the paper \cite{date1983method}.

\subsubsection{The mKdV equation}
\label{sec:dmKdV}

The discrete modified Korteweg-de Vries equation is also a periodic reduction of \eqref{eq:2ddToda_nc}, but with some changes. Taking the map \eqref{eq:2ddTodatoHirota}, i.e. considering the Hirota equation \eqref{eq:Hirota_nc}, supplemented with a transformation of the form
\begin{align}
    \label{eq:2ddTodatodmKdV}
    &&
    \theta_{l, m, n}
    &\mapsto
    \left\{
    \begin{array}{lr}
         \theta_{l, m},
         &  \text{even } n,
         \\[2mm]
         \theta_{l, m}^{-1},
         &  \text{odd } n;
    \end{array}
    \right.
    &&&
    \psi_{l, m, n}
    &\mapsto
    \left\{
    \begin{array}{lr}
         \psi_{l, m, 0},
         &  \text{even } n,
         \\[2mm]
         \psi_{l, m, 1}^{-1},
         &  \text{odd } n,
    \end{array}
    \right.
    &&&&
\end{align}
the 2d discrete Toda equation \eqref{eq:2ddToda_nc} changes as
\begin{align}
    \label{eq:dmKdV_nc_1}
    \theta_{l + 1, m + 1}^{-1} 
    \theta_{l + 1, m}
    - \theta_{l + 1, m + 1}^{-1} \theta_{l, m + 1}
    &= \theta_{l, m + 1}
    \theta_{l, m}^{-1} 
    - \theta_{l + 1, m}
    \theta_{l, m}^{-1}.
\end{align}
After the described transformations and rescaling $\psi_{l, m, n} \mapsto \lambda^{-n} \psi_{l, m, n}$, the scalar Lax system \eqref{eq:2ddToda_scalsys} turns~into
\begin{align}
    &
    \left\{
    \begin{array}{lcl}
         \psi_{l + 1, m, 0}
         &=& \lambda \psi_{l, m, 1}^{-1}
         + \theta_{l + 1, m}
         \theta_{l, m}^{-1}
         \psi_{l, m, 0}
         ,
         \\[2mm]
         \psi_{l, m + 1, 1}^{-1}
         &=& \lambda \psi_{l, m, 0}
         + \theta_{l, m + 1}^{-1} \theta_{l, m}
         \psi_{l, m, 1}^{-1}
         ;
    \end{array}
    \right.
    &
    &
    \left\{
    \begin{array}{lcl}
         \psi_{l + 1, m, 1}^{-1}
         &=& \lambda \psi_{l, m, 0}
         + \theta_{l + 1, m}^{-1} 
         \theta_{l, m}
         \psi_{l, m, 1}^{-1},
         \\[2mm]
         \psi_{l, m + 1, 0}
         &=& \lambda \psi_{l, m, 1}^{-1}
         + \theta_{l, m + 1} \theta_{l, m}^{-1}
         \psi_{l, m, 0},
    \end{array}
    \right.
\end{align}
This system can be easily rewritten in the matrix form. So, one can verify the following proposition.
\begin{prop}
\label{thm:dmKdV_matsys}
The non-abelian mKdV equation \eqref{eq:dmKdV_nc_1} follows from the matrix linear system
\begin{align}
    \label{eq:dmKdV_matsys}
    &&
    &
    \left\{
    \begin{array}{rcl}
         \Psi_{l, m + 1}
         &=& \mathcal{L}_{l, m} \Psi_{l, m},  
         \\[2mm]
         \Psi_{l + 1, m}
         &=& \mathcal{M}_{l, m} \Psi_{l, m},  
    \end{array}
    \right.
    &
    \Psi_{l, m}
    &= 
    \begin{pmatrix}
        \psi_{l, m, 0}
        & \psi_{l, m, 1}^{-1}
    \end{pmatrix}^{T},
    &&
\end{align}
where matrices $\mathcal{L}_{l, m}$, $\mathcal{M}_{l, m}$ read as
\begin{align}
    \label{eq:dmKdV_pair}
    &&
    \mathcal{L}_{l, m}
    &= 
    \begin{pmatrix}
        \theta_{l, m + 1} 
        \theta_{l, m}^{-1}
        & 
        \lambda
        \\[1mm]
        \lambda
        & 
        \theta_{l, m + 1}^{-1} \theta_{l, m}
    \end{pmatrix},
    &
    \mathcal{M}_{l, m}
    &= 
    \begin{pmatrix}
        \theta_{l + 1, m} 
        \theta_{l, m}^{-1}
        & 
        \lambda
        \\[1mm]
        \lambda
        & 
        \theta_{l + 1, m}^{-1} \theta_{l, m}
    \end{pmatrix}
    .
    &&
\end{align}
\end{prop}

Note that in the commutative case the sine-Gordon equation is related to the modified Korteweg-de Vries equation by the change \cite{quispel1991integrable}
\begin{align}
    \label{eq:dsGtodmKdV}
    &&
    \theta_{l, m}
    &\mapsto
    \left\{
    \begin{array}{lr}
         \theta_{l, m}^{-1},
         &  \text{even } m,
         \\[2mm]
         \theta_{l, m},
         &  \text{odd } m;
    \end{array}
    \right.
    &&&
    \psi_{l, m, i}
    &\mapsto
    \left\{
    \begin{array}{lr}
         \psi_{l, m, i}^{-1},
         &  \text{even } m,
         \\[2mm]
         \psi_{l, m, i},
         &  \text{odd } m,
    \end{array}
    \right.
    &
    i
    &= 0, 1.
    &&&&
\end{align}
Using the same change in the non-abelian setting, one can obtain from \eqref{eq:dmKdV_nc_1} the following non-abelian discrete sine-Gordon equation
\begin{align}
    \label{eq:dsG_nc_}
    \theta_{l + 1, m + 1}
    \theta_{l + 1, m} 
    - \theta_{l + 1, m + 1} 
    \theta_{l, m + 1}^{-1}
    &= \theta_{l, m + 1}^{-1} 
    \theta_{l, m}^{-1}
    - \theta_{l + 1, m}
    \theta_{l, m}^{-1}
    .
\end{align}
The latter is equivalent to \eqref{eq:dsG_nc}. The equivalence can be achieved by the transformation
$\theta_{l, m} \mapsto \theta_{l, 1 - m}^{-1}$.

\subsection{Remarks on Miura transformations}

\subsubsection{Volterra and Toda lattices}

Non-abelian Volterra lattice reads as \cite{sall1982darboux}
\begin{align}
    \label{eq:1dV_nc}
    \gamma_{n}'
    &= \gamma_{n + 1} \gamma_{n}
    - \gamma_{n} \gamma_{n - 1}
\end{align}
and is eqivalent to a zero-curvature condition of the form \eqref{eq:zrc_1d} with matrices
\begin{align}
    \mathcal{L}_n
    &= 
    \begin{pmatrix}
        \lambda & \lambda \gamma_n
        \\[0.9mm]
        -1 & 0
    \end{pmatrix}
    ,
    &
    \mathcal{M}_n
    &= 
    \begin{pmatrix}
        \lambda + \gamma_n & 
        \lambda \gamma_n
        \\[0.9mm]
        -1 & \gamma_{n - 1}
    \end{pmatrix}
    .
\end{align}
This equation can be rewritten as a system, by introducing variables
\begin{align}
    &&
    u_k
    &= \gamma_{2k + 1},
    & 
    v_k
    &= \gamma_{2k}.
    &&
\end{align}
Thus, the resulting system is
\begin{align}
    \label{eq:1dV_uv_nc}
    \left\{
    \begin{array}{lcl}
         u_n'
         &=& v_{n + 1} u_n
         - u_n v_n,
         \\[2mm]
         v_n'
         &=& u_n v_n 
         - v_n u_{n - 1}.
    \end{array}
    \right.
\end{align}
This system transforms to system \eqref{eq:1dToda_scalsysab} that is equivalent to the non-abelian 1d Toda equation \eqref{eq:1dToda_nc}. Namely, we have the following statement that generalizes the well-known Miura transformation from the paper \cite{yamilov1994construction} to the non-abelian case.
\begin{prop}
\label{thm:1dVLto1dToda}
The non-abelian Volterra lattice reduces to non-abelian 1d Toda equation by the discrete Miura transformation
\begin{align}
    \label{eq:1dVLto1dToda_nc}
    &&
    a_n
    &= u_n + v_n,
    &
    b_n
    &= v_n u_{n - 1}.
    &&
\end{align}
\end{prop}
\begin{proof}
The proof is given by a direct computation. Namely, using system \eqref{eq:1dV_uv_nc}, we have
\begin{align}
    a_n'
    &= u_n' + v_n'
    = (v_{n + 1} u_n
    - u_n v_n)
    + (u_n v_n 
    - v_n u_{n - 1})
    = b_{n + 1} - b_n,
    \\[2mm]
    b_n'
    &= v_n' u_{n - 1}
    + v_n u_{n - 1}'
    = (u_n v_n 
    - v_n u_{n - 1}) u_{n - 1}
    + v_n (
    v_n u_{n - 1}
    - u_{n - 1} v_{n - 1}
    )
    \\[1mm]
    &= (
    u_n + v_n
    ) v_n u_{n - 1}
    - v_n u_{n - 1} (
    u_{n - 1}
    + v_{n - 1}
    )
    = a_n b_n - b_n a_{n - 1}.
\end{align}
Therefore, we arrive at \eqref{eq:1dToda_scalsysab}.
\end{proof}

\subsubsection{KdV and mKdV lattices}
\label{sec:pkdv}
Recall that in the commutative case we have the scheme
\begin{align}
    &&&&
    {\rm dKdV}
    &&
    \longleftrightarrow
    &&
    {\rm dpKdV}
    &&
    \longleftrightarrow
    &&
    {\rm dmKdV}
    ,
    &&&&
\end{align}
where the dpKdV is a discrete potential KdV equation. 
Below we will consider its generalization to the non-abelian case.

The commutative potential discrete KdV equation
\begin{align}
    \brackets{
    u_{m, n} - u_{m + 1, n + 1}
    } \brackets{
    u_{m, n + 1} - u_{m + 1, n}
    }
    &= p \, q
\end{align}
is related to the discrete KdV equation 
\begin{align}
    p \brackets{
    v_{m, n} - v_{m + 1, n + 1}
    }
    &= q \brackets{
    v_{m, n + 1}^{-1} - v_{m + 1, n}^{-1}
    },
\end{align}
where $p$ and $q$ are arbitrary constants. Note that they can be reduced to one by an appropriate scaling depending on the indices $m$, $n$.
The relation can be found, for instance, in the paper \cite{shi2013darboux} and easily generalizes to the non-abelian case. Namely, we have 
\begin{thm}
\label{thm:dKdVtodpKdV_nc}
Let $p$, $q$ be arbitrary abelian constants. Then the system 
\begin{align}
    \label{eq:dKdVtodpKdV_nc}
    \left\{
    \begin{array}{lcl}
         p \, v_{m, n + 1}
         &=& u_{m, n + 1} - u_{m + 1, n},  
         \\[2mm]
         q \, v_{m, n}^{-1}
         &=& u_{m, n - 1} - u_{m + 1, n}
    \end{array}
    \right.
\end{align}
establishes a relation between the non-abelian discrete KdV equation \eqref{eq:dKdV_nc} and its potential analogs:
\begin{align}
    \label{eq:dpKdV_nc_1}
    \brackets{
    u_{m, n + 1} - u_{m + 1, n}
    } \brackets{
    u_{m, n} - u_{m + 1, n + 1}
    }
    &= p \, q,
    \\[2mm]
    \label{eq:dpKdV_nc_2}
    \brackets{
    u_{m, n} - u_{m + 1, n + 1}
    } \brackets{
    u_{m, n + 1} - u_{m + 1, n}
    }
    &= p \, q.
\end{align}
\end{thm}
\begin{rem}
The latter two equations are equivalent up to the transposition $\tau$ given in Definition~\ref{def:tau}.
\end{rem}
\begin{proof}
Our aim is to eliminate the variable $u$ or $v$ from system \eqref{eq:dKdVtodpKdV_nc}.

\medskip
\textbf{\textbullet \,\, The pKdV equations.} 
In order to obtain \eqref{eq:dpKdV_nc_1}, one needs to consider $T_2 \eqref{eq:dKdVtodpKdV_nc}_2 \cdot \eqref{eq:dKdVtodpKdV_nc}$:
\begin{align}
    p \, q \,\, v_{m, n + 1}^{-1} v_{m, n + 1}
    &= \brackets{
    u_{m, n} - u_{m + 1, n + 1}
    } \brackets{
    u_{m, n + 1} - u_{m + 1, n}
    }. 
\end{align}
Similarly, expression $\eqref{eq:dKdVtodpKdV_nc} \cdot T_2 \eqref{eq:dKdVtodpKdV_nc}_2$ leads to \eqref{eq:dpKdV_nc_2}.

\medskip
\textbf{\textbullet \,\, The KdV equation.} 
Let us take $T_2 \eqref{eq:dKdVtodpKdV_nc}_2 - T_1 \eqref{eq:dKdVtodpKdV_nc}_2$: 
\begin{align}
    q \brackets{
    v_{m, n + 1}^{-1} - v_{m + 1, n}^{-1}
    }
    &= u_{m, n} - u_{m + 1, n + 1} - u_{m + 1, n - 1} + u_{m + 2, n}.
\end{align}
Note that 
\begin{align}
    T_{- 2} \eqref{eq:dKdVtodpKdV_nc}_1
    - T_{1} \eqref{eq:dKdVtodpKdV_nc}_1
    &= p \brackets{
    v_{m, n} - v_{m + 1, n + 1}
    }
    = u_{m, n} - u_{m + 1, n - 1} - u_{m + 1, n + 1} + u_{m + 2, n}.
\end{align}
Therefore, taking the difference of the latter two expressions, we arrive at the equation
\begin{align}
    p \brackets{
    v_{m, n} - v_{m + 1, n + 1}
    }
    &= q \brackets{
    v_{m, n + 1}^{-1} - v_{m + 1, n}^{-1}
    },
\end{align}
which coincides with \eqref{eq:dKdV_nc} up to a scaling of $v_{m, n}$.
\end{proof}

One can verify that the linear system
\begin{align}
    \label{eq:dpKdV_matsys}
    &
    \left\{
    \begin{array}{rcl}
         \Psi_{m, n + 1}
         &=& \mathcal{L}_{m, n} \Psi_{m, n},  
         \\[2mm]
         \Psi_{m + 1, n}
         &=& \mathcal{M}_{m, n} \Psi_{m, n}
    \end{array}
    \right.
\end{align}
with matrices $\mathcal{L}_{m, n}$, $\mathcal{M}_{m, n}$ of the form
\begin{align}
    \label{eq:dpKdV_pair}
    &&
    \mathcal{L}_{m, n}
    &= 
    \begin{pmatrix}
        u_{m, n}
        & 
        \lambda + p \, q 
        - u_{m, n} u_{m, n + 1}
        \\[1mm]
        1
        & 
        - u_{m, n + 1}
    \end{pmatrix},
    &
    \mathcal{M}_{m, n}
    &= 
    \begin{pmatrix}
        u_{m, n}
        & 
        \lambda
        - u_{m, n} u_{m + 1, n}
        \\[1mm]
        1
        & 
        - u_{m + 1, n}
    \end{pmatrix}
    &&
\end{align}
is equivalent to the non-abelian pKdV equation \eqref{eq:dpKdV_nc_1}. Indeed, there are only three non-zero matrix entries of the compatibility condition 
\begin{align}
    \mathcal{L}_{m + 1, n} \mathcal{M}_{m, n}
    &= \mathcal{M}_{m, n + 1}
    \mathcal{L}_{m, n}.
\end{align}
The $(1,1)$ and $(2,2)$ elements
coincide with \eqref{eq:dpKdV_nc_1} and \eqref{eq:dpKdV_nc_2}, respectively, while the $(1,2)$ is equal to their linear combination: $\eqref{eq:dpKdV_nc_1} \, \cdot \, u_{m, n + 1} + u_{m + 1, n} \, \cdot \, \eqref{eq:dpKdV_nc_2}$.

There is a discrete Miura transformation between potential KdV and modified KdV lattices \cite{nijhoff1995discrete}. It has a non-abelian analog that brings the potential KdV equation \eqref{eq:dpKdV_nc_1} to the following mKdV equation different from \eqref{eq:dmKdV_nc_1}:
\begin{align}
    \label{eq:dmKdV_nc_2}
    \theta_{m, n + 1} \theta_{m + 1, n + 1}^{-1}
    - \theta_{m + 1, n} \theta_{m + 1, n + 1}^{-1}
    &= \theta_{m + 1, n} \theta_{m, n}^{-1}
    - \theta_{m, n + 1} \theta_{m, n}^{-1}.
\end{align}
Namely, we have 
\begin{prop}
\label{thm:pKdVtomKdV}
Let $p = a_1 - a_2$ and $q = a_1 + a_2$. Then the non-abelian pKdV \eqref{eq:dpKdV_nc_1} and mKdV \eqref{eq:dmKdV_nc_2} lattices are connected by the Miura map
\begin{align}
    \label{eq:pKdVtomKdV}
    \left\{
    \begin{array}{lcl}
         u_{m, n + 1} - u_{m + 1, n}
         &=& a_1 \, \theta_{m, n + 1} \theta_{m + 1, n + 1}^{-1}
         - a_2 \, \theta_{m + 1, n} \theta_{m + 1, n + 1}^{-1},
         \\[2mm]
         u_{m, n} - u_{m + 1, n + 1}
         &=& a_1 \, \theta_{m, n} \theta_{m + 1, n}^{-1}
         + a_2 \, \theta_{m + 1, n + 1} \theta_{m + 1, n}^{-1}.
    \end{array}
    \right.
\end{align}
\end{prop}

\begin{proof}
The proof is given by a straightforward computation. Substitution of the change \eqref{eq:pKdVtomKdV} into \eqref{eq:dpKdV_nc_1} leads to the following chain of identities: 
\begin{align}
    \brackets{
    a_1 \, \theta_{m, n + 1} \theta_{m + 1, n + 1}^{-1}
    - a_2 \, \theta_{m + 1, n} \theta_{m + 1, n + 1}^{-1}
    } \brackets{
    a_1 \, \theta_{m, n} \theta_{m + 1, n}^{-1}
    + a_2 \, \theta_{m + 1, n + 1} \theta_{m + 1, n}^{-1}
    }
    = a_1^2 - a_2^2,
    \\[2mm]
    a_1 \, \theta_{m, n + 1} \theta_{m + 1, n + 1}^{-1} \,\,
    \theta_{m, n} \theta_{m + 1, n}^{-1}
    + a_2 \, \theta_{m, n + 1} \theta_{m + 1, n}^{-1}
    - a_2 \, \theta_{m + 1, n} \theta_{m + 1, n + 1}^{-1} \,\, 
    \theta_{m, n} \theta_{m + 1, n}^{-1}
    = a_1,
    \\[2mm]
    \label{eq:mkdv}
    a_1 \, \theta_{m, n + 1} \theta_{m + 1, n + 1}^{-1}
    - a_2 \, \theta_{m + 1, n} \theta_{m + 1, n + 1}^{-1}
    = a_1 \, \theta_{m + 1, n} \theta_{m, n}^{-1}
    - a_2 \, \theta_{m, n + 1} \theta_{m, n}^{-1}
    .
\end{align}
Note that the constants $a_1$, $a_2$ can be absorbed by a proper scaling of $\theta_{m, n}$. These transformations are invertible. So, starting from the rescaled mKdV lattice \eqref{eq:mkdv} and using system \eqref{eq:pKdVtomKdV}, one is able to arrive at the pKdV equation \eqref{eq:dpKdV_nc_1}. 
\end{proof}

\begin{rem}
According to \eqref{eq:dsGtodmKdV}, this mKdV equation turns into the following sine-Gordon lattice:
\begin{align}
    \label{eq:dsG_nc_2}
    \theta_{m, n + 1}^{-1} \theta_{m + 1, n + 1}
    - \theta_{m + 1, n} \theta_{m + 1, n + 1}
    = \theta_{m + 1, n} \theta_{m, n}^{-1}
    - \theta_{m, n + 1}^{-1} \theta_{m, n}^{-1}
    .
\end{align}
\end{rem}

In order to obtain a Lax pair for \eqref{eq:dmKdV_nc_2}, we first conjugate the pair \eqref{eq:dpKdV_pair} by the matrix
\begin{align}
    G_{m, n}
    &= 
    \begin{pmatrix}
    1 & u_{m + 1, n}
    \\[0.9mm]
    0 & 1
    \end{pmatrix}
\end{align}
to bring it to the form
\begin{multline}
    \mathcal{L}_{m, n}
    = 
    \begin{pmatrix}
        u_{m, n} - u_{m + 2, n}
        & 
        \lambda + p \, q 
        - \brackets{
        u_{m, n} - u_{m + 2, n}
        } \, \brackets{
        u_{m, n + 1} - u_{m + 1, n}
        }
        \\[1mm]
        1
        & 
        u_{m + 1, n} - u_{m, n + 1}
    \end{pmatrix},
    \\
    \mathcal{M}_{m, n}
    = 
    \begin{pmatrix}
        u_{m, n} - u_{m + 1, n + 1}
        & 
        \lambda
        \\[1mm]
        1
        & 
        0
    \end{pmatrix}
\end{multline}
and then make the change \eqref{eq:pKdVtomKdV}. Note that $u_{m, n} - u_{m + 2, n} = \eqref{eq:pKdVtomKdV}_2 + T_1 \eqref{eq:pKdVtomKdV}_1$.

A similar Miura map does not lead to \eqref{eq:dmKdV_nc_1}. It happens due to the following factorization of \eqref{eq:dmKdV_nc_1}:
\begin{multline}
    a_1 \, \brackets{
    a_1 \, \theta_{l + 1, m + 1}^{-1} \theta_{l, m + 1}
    - a_2 \, \theta_{l + 1, m + 1}^{-1} \theta_{l + 1, m}
    } \, \theta_{l, m} \theta_{l + 1, m}^{-1} 
    \\
    + a_2 \, \brackets{
    a_1 \, \theta_{l, m + 1} \theta_{l + 1, m + 1}^{-1}
    - a_2 \, \theta_{l + 1, m} \theta_{l + 1, m + 1}^{-1}
    } \, \theta_{l + 1, m + 1} \theta_{l + 1, m}^{-1}
    = a_1^2 - a_2^2
    .
\end{multline}
In particular, the authors do not know how to reduce the non-abelian 2d discrete Toda equation \eqref{eq:2ddToda_nc} to \eqref{eq:dmKdV_nc_2}. Perhaps, the mKdV equation \eqref{eq:dmKdV_nc_2} follows from another non-abelian analog of the 2d discrete Toda lattice. 

\section{Somos-like sequences and 
\texorpdfstring{$q$}{q}-\Painleve equations}
\label{sec:Somos_dPainleve}

Going further, one can reduce the 2d discrete Toda equation to integrable sequences. Thus, in Section \ref{sec:Somos_dPainleve}, we consider a reduction to the Somos-like sequence that in special case can be rewritten to the $q$-\Painleve~I and II equations and their hierarchies as it was done in the paper \cite{hone2014discrete}.

\subsection{Somos-like sequences}
\label{sec:ncsomos}
Several authors have mentioned relations between  discrete Hirota equation and the Somos-like sequences. It can be achieved by a plane-wave reduction (see, e.g., \cite{hone2017reductions}), that we apply in the non-commutative case. This reduction leads to a one-dimensional discrete equation and, thus, the resulting scalar Lax pair will have two spectral parameters. Let us stress that autonomous one-dimensional systems (discrete or continuous) admit the spectral Lax pairs only, while in the non-autonomous case they possess Lax pair of the Zakharov-Shabat type (or, in other words, isomonodromic Lax pairs). 

Starting from the non-abelian two dimensional Toda equation \eqref{eq:2ddToda_nc} and its scalar Lax pair \eqref{eq:2ddToda_scalsys}, let us first implement additional constant parameters. Namely, the scaling (cf. with \eqref{eq:scalrule})
\begin{align}
    \label{eq:scalrule_al}
    &&
    \theta_{l, m, n}
    &\mapsto \theta_{l, m, n} \, \alpha^{l + m},
    &
    \psi_{l, m, n}
    &\mapsto \lambda^{m} \, \mu^{- (m + n)} \, \psi_{l, m, n}
    &&
\end{align}
gives the 2d discrete Toda equation \eqref{eq:2ddToda_nc} of the form
\begin{align}
    \label{eq:2ddToda_nc_al}
    \theta_{l + 1, m + 1, n}
    &= \alpha^{2} \, \theta_{l, m, n + 1}
    + \theta_{l + 1, m, n} \brackets{
    \theta_{l, m, n}^{-1}
    - \alpha^{2} \, \theta_{l + 1, m + 1, n - 1}^{-1}
    } \theta_{l, m + 1, n}
\end{align}
and the corresponding scalar Lax system
\begin{gather}
    \label{eq:2ddToda_scalsys_al_nc}
    \left\{
    \begin{array}{rcl}
         \psi_{l + 1, m, n}
         &=& \mu \, \psi_{l, m, n + 1} 
         + a_{l, m, n} \, \psi_{l, m, n},
         \\[2mm]
         \lambda \, \psi_{l, m - 1, n + 1}
         &=& \mu \, \psi_{l, m, n + 1}
         + b_{l, m, n} \, \psi_{l, m, n},
    \end{array}
    \right.
    \\[3mm]
    \begin{aligned}
        a_{l, m, n}
        &= \alpha^{-1} \,\, \theta_{l + 1, m, n} \, \theta_{l, m, n}^{-1},
        &&&
        b_{l, m, n}
        &= \alpha \,\, \theta_{l, m - 1, n} \, \theta_{l, m, n}^{-1}.
    \end{aligned}
\end{gather}
The parameters $\lambda$ and $\mu$ will play a role of the spectral parameters in a reduced one-dimensional system.

\begin{prop}
\label{thm:TodatoSomosN}
Let the functions $\theta_{l, m, n}$ and $\psi_{l, m, n}$ satisfy \eqref{eq:2ddToda_nc_al} and \eqref{eq:2ddToda_scalsys_al_nc}. Then, the plane-wave reduction
\begin{align}
    &&
    y_{M} 
    &:= \theta_{l, m, n},
    &
    \varphi_{M}
    &:= \psi_{l, m, n}
    &
    M 
    &= (N - s) \, l + s \, m + r \, n,
    &&
\end{align}
changes \eqref{eq:2ddToda_nc_al} into the non-commutative version of the Somos-$N$ type equation
\begin{gather}
    \label{eq:SomosN_nc}
    \begin{aligned}
    y_{M + N}
    &= \alpha^2 \, y_{M + r}
    + y_{M + N - s} \, 
    \brackets{
    y_M^{-1}
    - \alpha^2 \, y_{M + N - r}^{-1} 
    } \,
    y_{M + s},
    &&&
    N 
    &> 3,
    &
    1 
    &\leq r < s \leq \left[\tfrac{N}{2}\right],
    \end{aligned}
\end{gather}
while the system \eqref{eq:2ddToda_scalsys_al_nc} becomes
\begin{gather}
    \label{eq:SomosN_scalsys_nc}
    \left\{
    \begin{array}{rcl}
         \varphi_{M + N - s}
         &=& \mu \, \varphi_{M + r} 
         + a_{M} \, \varphi_{M},
         \\[2mm]
         \lambda \, \varphi_{M}
         &=& \mu \, \varphi_{M + s}
         + b_{M} \, \varphi_{M + s - r},
    \end{array}
    \right.
    \\[3mm]
    \begin{aligned}
        a_{M}
        &= \alpha^{-1} \,\, y_{M + N - s} \, y_{M}^{-1},
        &&&&&
        b_{M}
        &= \alpha \,\, y_M \, y_{M + s - r}^{-1}.
    \end{aligned}
\end{gather}
\end{prop}
\begin{proof}
The proof is just a straightforward computation. 
\end{proof}

\begin{rem}
The condition for integer parameters $N$, $s$, and $r$ is not necessary for the reduction and is motivated by the commutative Somos-$N$ equation. 
\end{rem}

\begin{rem}
The parameter $\alpha$ may be chosen as a non-abelian constant parameter. Note that the order of the multipliers in \eqref{eq:scalrule_al} is important and must be used as written. Then, \eqref{eq:2ddToda_nc_al} takes the form
\begin{align}
    \label{eq:2ddToda_nc_al_nc}
    \theta_{l + 1, m + 1, n}
    &= \theta_{l, m, n + 1} \, \alpha^{2} 
    + \theta_{l + 1, m, n} \, \alpha \, \brackets{
    \alpha^{-2} \, \theta_{l, m, n}^{-1}
    - \theta_{l + 1, m + 1, n - 1}^{-1}
    } \theta_{l, m + 1, n} \, \alpha
\end{align}
and the functions $a_{l, m, n}$, $b_{l, m, n}$ in \eqref{eq:2ddToda_scalsys_al_nc} become
\begin{align}
    &&
    a_{l, m, n}
    &= \theta_{l + 1, m, n} \, \alpha^{-1} \, \theta_{l, m, n}^{-1},
    &
    b_{l, m, n}
    &= \theta_{l, m - 1, n + 1} \, \alpha \, \theta_{l, m, n}^{-1}.
    &&
\end{align}
It is easy to obtain a similar to Proposition \ref{thm:TodatoSomosN} statement.
\end{rem}

In order to get the Somos-$N$ equation \eqref{eq:SomosN_nc} from the system \eqref{eq:SomosN_scalsys_nc}, one needs to rewrite the latter in matrix form. Namely, introducing the vector-function 
$\Phi_M = \begin{pmatrix} \varphi_{M + N - s - 1} & \dots & \varphi_{M + 1} & \varphi_{M} \end{pmatrix}^T$, the system \eqref{eq:SomosN_scalsys_nc} can be rewritten as
\begin{align}
    \label{eq:SomosN_matsys}
    \left\{
    \begin{array}{rcl}
         \mathcal{L}_M \, \Phi_M
         &=& \lambda \, \Phi_{M},  
         \\[2mm]
         \Phi_{M + 1}
         &=& \mathcal{M}_{M} \Phi_{M},  
    \end{array}
    \right.
\end{align}
where $(N - s - 1)\times(N - s - 1)$-matrices $\mathcal{L}_M = \mathcal{L}_M (\mu)$ and $\mathcal{M}_M = \mathcal{M}_M (\mu)$ depends on the spectral parameter $\mu$ and $\lambda$ is an eigenvalue for the matrix $\mathcal{L}_M$. Then, the compatibility condition
\begin{align}
    \label{eq:compcond}
    \mathcal{L}_{M + 1} \, \mathcal{M}_M
    &= \mathcal{M}_M \, \mathcal{L}_M
\end{align}
of system \eqref{eq:SomosN_matsys} gives rise to \eqref{eq:SomosN_nc}. As in the commutative case (see Subsection \ref{sec:somoslikeseq}), here we avoid an explicit form of the matrices $\mathcal{L}_M$ and $\mathcal{M}_M$ and this fact follows from the reduced compatibility condition of system \eqref{eq:2ddToda_scalsys_al_nc}. But, for the reader's convenience, below we list several explicit examples.

\begin{rem}
Note that, unlike the commutative case, the compatibility condition does not lead to a non-autonomous parameter in the Somos-$N$ sequence.
\end{rem}

\begin{exmp}
Let us illustrate how Proposition \ref{thm:TodatoSomosN} works. Consider the Somos-$4$, Somos-$5$, and the special case of Somos-$7$ sequences (recall Remark \ref{rem:spSomos7}). The case of Somos-$4$ will contain a detailed discussion, for the remaining sequences we will omit computations and just present the resulting matrices for \eqref{eq:SomosN_matsys}.

\begin{itemize}
\item 
The following set of parameters
\begin{align}
    &&&&
    N
    &= 4,
    &
    s
    &= 2,
    &
    r
    &= 1
    &&&&
\end{align}
corresponds to Somos-$4$. Then, the system \eqref{eq:SomosN_scalsys_nc},
\begin{gather}
    \label{eq:Somos4_scalsys_nc}
    \left\{
    \begin{array}{rcl}
         \varphi_{M + 2}
         &=& \mu \, \varphi_{M + 1} 
         + a_{M} \, \varphi_{M},
         \\[2mm]
         \lambda \, \varphi_{M}
         &=& \mu \, \varphi_{M + 2}
         + b_{M} \, \varphi_{M + 1},
    \end{array}
    \right.
    \\[3mm]
    \begin{aligned}
        a_{M}
        &= \alpha^{-1} \,\, y_{M + 2} \, y_{M}^{-1},
        &&&&&
        b_{M}
        &= \alpha \,\, y_M \, y_{M + 1}^{-1},
    \end{aligned}
\end{gather}
turns into \eqref{eq:SomosN_matsys} as follows. Introduce the vector-function
$\Phi_M = \begin{pmatrix} \varphi_{M + 1} & \varphi_{M} \end{pmatrix}^T$. In order to get the matrix $\mathcal{M}_M$, one needs to use the first equation of system \eqref{eq:Somos4_scalsys_nc} to compute the components of $\Phi_{M + 1}$. So, they are given by
\begin{align}
    &&
    \varphi_{M + 2}
    &= \mu \, \varphi_{M + 1} + a_M \varphi_M
    = \begin{pmatrix} \mu & a_M \end{pmatrix} \, \Phi_M,
    &
    \varphi_{M + 1}
    &= \begin{pmatrix} 1 & 0 \end{pmatrix} \, \Phi_M.
    &&
\end{align}
Regarding the $\mathcal{L}_M$ matrix, we should use the both equations of system \eqref{eq:Somos4_scalsys_nc}. This leads to the following chains of identities
\begin{align}
    \lambda \, \varphi_M
    &= \mu \, \varphi_{M + 2} + b_M \, \varphi_{M + 1}
    = \mu \, \brackets{
    \mu \, \varphi_{M + 1} + a_M \, \varphi_M
    }
    + b_M \, \varphi_{M + 1}
    \\[1mm]
    &= 
    \begin{pmatrix} 
    \mu^2 + b_M 
    & 
    \mu \, a_M 
    \end{pmatrix} \, \Phi_M,
    \\[3mm]
    \lambda \, \varphi_{M + 1}
    &= \brackets{
    \mu^2 + b_{M + 1}
    } \, \varphi_{M + 2}
    + \mu \, a_{M + 1} \, \varphi_{M + 1}
    \\[1mm]
    &= \brackets{
    \mu^2 + b_{M + 1}
    } \, \brackets{
    \mu \, \varphi_{M + 1} + a_M \, \varphi_M
    }
    + \mu \, a_{M + 1} \, \varphi_{M + 1}
    \\[1mm]
    &= 
    \begin{pmatrix} 
    \mu \, \brackets{
    \mu^2 + b_{M + 1} + a_{M + 1}
    }
    & 
    \brackets{
    \mu^2 + b_{M + 1}
    } \, a_M 
    \end{pmatrix} \, \Phi_M.
\end{align}
Therefore, the resulting matrices are given by
\begin{align}
    &&
    \mathcal{L}_M
    &= 
    \begin{pmatrix}
    \mu \, \brackets{
    \mu^2 + b_{M + 1} + a_{M + 1}
    }
    & 
    \brackets{
    \mu^2 + b_{M + 1}
    } \, a_M 
    \\[0.9mm]
    \mu^2 + b_M 
    & 
    \mu \, a_M 
    \end{pmatrix}
    ,
    &
    \mathcal{M}_M
    &= 
    \begin{pmatrix}
        \mu & a_M
        \\[0.9mm]
        1 & 0
    \end{pmatrix}
    .
    &&
\end{align}
The compatibility condition \eqref{eq:compcond} is equivalent to the non-abelian Somos-$4$ equations:
\begin{align}
    y_{M + 4}
    &= \alpha^2 \, y_{M + 1}
    + y_{M + 2} \, 
    \brackets{
    y_M^{-1}
    - \alpha^2 \, y_{M + 3}^{-1} 
    } \,
    y_{M + 2}.
\end{align}

\item
The Somos-$5$ can be obtained by setting $N = 5$, $s = 2$, $r = 1$. Then, considering the vector-function $\Phi_M = \begin{pmatrix} \varphi_{M + 2} & \varphi_{M + 1} & \varphi_{M} \end{pmatrix}^T$, the corresponding system \eqref{eq:SomosN_scalsys_nc} turns into \eqref{eq:SomosN_matsys} with the matrices
\begin{align}
    \mathcal{L}_M
    &= 
    \begin{pmatrix}
    \mu^2 
    & \mu \, \brackets{b_{M + 2} + a_{M + 1}}
    & b_{M + 2} \, a_M
    \\[0.9mm]
    b_{M + 1} & \mu^2 & \mu \, a_M
    \\[0.9mm]
    \mu & b_M & 0
    \end{pmatrix},
    &
    \mathcal{M}_M
    &= 
    \begin{pmatrix}
    0 & \mu & a_M
    \\[0.9mm]
    1 & 0 & 0
    \\[0.9mm]
    0 & 1 & 0
    \end{pmatrix}
    .
\end{align}
Condition \eqref{eq:compcond} leads to
\begin{align}
    y_{M + 5}
    &= \alpha^2 \, y_{M + 1}
    + y_{M + 3} \, 
    \brackets{
    y_M^{-1}
    - \alpha^2 \, y_{M + 4}^{-1} 
    } \,
    y_{M + 2}.
\end{align}

\item
For the special case of the Somos-$7$ (see Remark \ref{rem:spSomos7}), let us set $N = 7$, $s = 3$, $r = 1$ in \eqref{eq:SomosN_scalsys_nc}. Using the vector-function $\Phi_M = \begin{pmatrix} \varphi_{M + 3} & \varphi_{M + 2} & \varphi_{M + 1} & \varphi_{M} \end{pmatrix}^T$, one can arrive at \eqref{eq:SomosN_matsys}, where the matrices $\mathcal{L}_M$ and $\mathcal{M}_M$ read
\begin{align}
    &&&&&&
    \mathcal{L}_M
    &= 
    \begin{pmatrix}
    \mu^2 
    & \mu \, \brackets{b_{M + 3} + a_{M + 2}}
    & b_{M + 3} \, a_{M + 1}
    & 0
    \\[0.9mm]
    0
    & \mu^2
    & \mu \, \brackets{b_{M + 2} + a_{M + 1}}
    & b_{M + 2} \, a_M
    \\[0.9mm]
    b_{M + 1}
    & 0
    & \mu^2
    & \mu \, a_M
    \\[0.9mm]
    \mu
    & b_M
    & 0
    & 0
    \end{pmatrix},
    &
    \mathcal{M}_M
    &= 
    \begin{pmatrix}
    0 & 0 & \mu & a_M
    \\[0.9mm]
    1 & 0 & 0 & 0
    \\[0.9mm]
    0 & 1 & 0 & 0 
    \\[0.9mm]
    0 & 0 & 1 & 0
    \end{pmatrix}
    .
\end{align}
A Somos-$7$ equation of the form
\begin{align}
    y_{M + 7}
    &= \alpha^2 \, y_{M + 1}
    + y_{M + 4} \, 
    \brackets{
    y_M^{-1}
    - \alpha^2 \, y_{M + 6}^{-1} 
    } \,
    y_{M + 3}
\end{align}
arise from \eqref{eq:compcond}. Setting $u_M = y_{M + 3} y_{M + 1}^{-1}$, the latter becomes (cf. with \cite{retakh2019noncommutative})
\begin{align}
    u_{M + 4} u_{M + 2}
    - u_{M + 1} u_{M - 1}
    &= \alpha^2 \brackets{
    u_{M}^{-1}
    - u_{M + 3}^{-1}
    }.
\end{align}
\end{itemize}
\end{exmp}

\begin{exmp}
When $N = 6$, $s$ and $r$ may belong to one of the following sets: $s = 3$, $r = 2$, $s = 3$, $r = 1$, and $s = 2$, $r = 1$. The corresponding matrices $\mathcal{L}_M$ and $\mathcal{M}_M$ are listed below.
\begin{itemize}
\item $N = 6$, $s = 3$, $r = 2$:
\begin{align}
    \mathcal{L}_M
    &= 
    \begin{pmatrix}
    \mu \brackets{
    \mu^3 + b_{M + 2} + a_{M + 2}
    }
    & \mu^2 \, a_{M + 1}
    & \brackets{
    \mu^3 + b_{M + 2}
    } \, a_{M}
    \\[0.9mm]
    \mu^3 + b_{M + 1} 
    & \mu \, a_{M + 1}
    & \mu^2 \, a_M
    \\[0.9mm]
    \mu^2 & b_M & \mu \, a_{M}
    \end{pmatrix},
    &
    \mathcal{M}_M
    &= 
    \begin{pmatrix}
    \mu & 0 & a_M
    \\[0.9mm]
    1 & 0 & 0
    \\[0.9mm]
    0 & 1 & 0
    \end{pmatrix}
    .
\end{align}

\item $N = 6$, $s = 3$, $r = 1$:
\begin{align}
    \mathcal{L}_M
    &= 
    \begin{pmatrix}
    \mu \, \brackets{b_{M + 2} + a_{M + 2}}
    & \mu^3 + b_{M + 2} \, a_{M + 1}
    & \mu^2 \, a_{M}
    \\[0.9mm]
    \mu^2 
    & \mu \, \brackets{
    b_{M + 1} + a_{M + 1}
    }
    & b_{M + 1} \, a_{M}
    \\[0.9mm]
    b_M 
    & \mu^2
    & \mu \, a_{M}
    \end{pmatrix},
    &
    \mathcal{M}_M
    &= 
    \begin{pmatrix}
    0 & \mu & a_M
    \\[0.9mm]
    1 & 0 & 0
    \\[0.9mm]
    0 & 1 & 0
    \end{pmatrix}
    .
\end{align}

\item $N = 6$, $s = 2$, $r = 1$:
\begin{align}
    \mathcal{L}_M
    &= 
    \begin{pmatrix}
    0
    & \mu^2 
    & \mu \, \brackets{b_{M + 3} + a_{M + 1}}
    & b_{M + 3} \, a_{M}
    \\[0.9mm]
    b_{M + 3}
    & 0
    & \mu^2
    & \mu \, a_M
    \\[0.9mm]
    \mu
    & b_{M + 1}
    & 0
    & 0
    \\[0.9mm]
    0
    & \mu
    & b_M
    & 0
    \end{pmatrix},
    &
    \mathcal{M}_M
    &= 
    \begin{pmatrix}
    0 & 0 & \mu & a_M
    \\[0.9mm]
    1 & 0 & 0 & 0
    \\[0.9mm]
    0 & 1 & 0 & 0 
    \\[0.9mm]
    0 & 0 & 1 & 0
    \end{pmatrix}
    .
\end{align}
\end{itemize}
\end{exmp}

\subsection{\texorpdfstring{$q$}{q}-\Painleve equations}
\label{sec:qPainleve}

First of all, we would like to introduce non-autonomous coefficients into the Somos-$N$ equation. Before the reduction given in Proposition \ref{thm:TodatoSomosN}, we need to implement it into the two-dimensional Toda lattice, since there is more freedom. Consider a transposed equation \eqref{eq:2ddToda_nc}. The scaling (cf.~with Remark \ref{rem:theta_qlmn})
\begin{align}
    \theta_{l, m, n}
    \mapsto 
    {\alpha}^{l} \, q^{\frac12 k_1 l (l - 1)} \, 
    q^{\frac12 k_2 m (m - 1)} \, 
    q^{\frac12 k_3 l (2 n + l - 1)} 
    \, 
    \theta_{l, m, n}
\end{align}
and transposition bring the non-abelian 2d Toda lattice \eqref{eq:2ddToda_nc} to the form
\begin{align}
    \label{eq:2ddToda_nc_alM}
    \begin{aligned}
    \theta_{l + 1, m + 1, n}
    = \alpha_{l, m, n} \theta_{l, m, n + 1}
    + \theta_{l, m + 1, n} \,\, \brackets{
    \theta_{l, m, n}^{-1}
    - \theta_{l + 1, m + 1, n - 1}^{-1}
    \alpha_{l, m, n - 1} 
    } \,
    \theta_{l + 1, m, n}
    ,
    \\
    \alpha_{l, m, n}
    = \alpha \, q^{k_1 l + k_2 m + k_3 n}.
    \end{aligned}
\end{align}
Here $\alpha$ is a non-abelian constant parameter and $k_i$, $q$ are commutative ones.

\begin{prop}
Let $\theta_{l, m, n}$ satisfy \eqref{eq:2ddToda_nc_alM} and $k_1 = N - s$, $k_2 = s$, $k_3 = r$. Then the function 
\begin{align}
    &&
    y_{M} 
    &:= \theta_{l, m, n},
    &
    M 
    &= (N - s) \, l + s \, m + r \, n
    &&
\end{align}
is a solution of the non-commutative version of the non-autonomous Somos-$N$ type equation:
\begin{gather}
    \label{eq:SomosN_nc_alM}
    \begin{aligned}
    \begin{aligned}
    y_{M + N}
    &= \alpha_M \, y_{M + r}
    + y_{M + s} \, 
    \brackets{
    y_M^{-1}
    - y_{M + N - r}^{-1} \, \alpha_{M - r} 
    } \,
    y_{M + N - s},
    &
    \alpha_M
    &= \alpha \, q^M,
    \end{aligned}
    \\
    \begin{aligned}
    N 
    &\in \mathbb{N}_{> 3},
    &
    1 
    &\leq r < s \leq \left[\tfrac{N}{2}\right].
    \end{aligned}
    \end{aligned}
\end{gather}
\end{prop}

Similarly to the paper \cite{hone2014discrete}, we will introduce the $q$PI and $q$PII hierarchies. Below we set
\begin{align}
    &&
    \prod_{k = 1}^n a_{k}
    &:= a_1 \, a_2 \, \dots \, a_n,
    &
    a_k
    &\in R.
    &&
\end{align}

\begin{prop}
\label{thm:qPIn_qPIIn_nc}
Let $y_M$ be a solution of \eqref{eq:SomosN_nc_alM} with $r = 1$ and $s = 2$. Then
\begin{itemize}
    \item[(a)] 
    for even $N \geq 4$, the change of variables
    \begin{align}
        \label{eq:Somos2qPI}
        u_M
        &= y_{M + 3} \, y_{M + 2}^{-1}
    \end{align}
    leads to the non-commutative $q$PI hierarchy
    \begin{align}
    \label{eq:qPIn_nc}
        u_{M + N - 3} \,
        u_{M + N - 4}
        - u_{M - 1} \, u_{M - 2}
        &= \prod_{k = 0}^{N - 4}
        \brackets{
        \alpha_M \, u_{M + k - 1}^{-1}
        - u_{M + k}^{-1} \, \alpha_{M - 1}
        };
    \end{align}
    
    \item[(b)] 
    for odd $N \geq 5$, the change of variables
    \begin{align}
        \label{eq:Somos2qPII}
        u_M
        &= y_{M + 4} \, y_{M + 2}^{-1}
    \end{align}
    leads to the non-commutative $q$PII hierarchy
    \begin{align}
    \label{eq:qPIIn_nc}
        \prod_{k = 1 - [\frac{N}{2}]}^0
        u_{M - 2 k - 1}
        - u_{M - 2} \, \prod_{k = 2 - [\frac{N}{2}]}^0
        u_{M - 2 k - 3}
        &= \alpha_M
        - \prod_{k = 0}^{[\frac{N}{2}] - 2}
        u_{M + 2k - 2}^{-1} \,\, 
        \alpha_{M - 1}
        \prod_{k = 2 - [\frac{N}{2}]}^{0}
        u_{M - 2k - 3}.
    \end{align}
\end{itemize}
\end{prop}

\begin{proof}
The proof is given by a straightforward computation. We first consider $N = 4$ and $N = 5$ to illustrate how the Somos equations can be rewritten before the change and, then, proceed to an arbitrary $N$.

\medskip
\textbullet \,\, \textbf{Case (a).} 
The Somos-$4$ equation can be written as
\begin{align}
    y_{M + 4} y_{M + 2}^{-1}
    - y_{M + 2} y_M^{-1}
    &= \alpha_M \, y_{M + 1} y_{M + 2}^{-1}
    - y_{M + 2} y_{M + 3}^{-1} \, \alpha_{M - 1}
    ,
\end{align}
or, after change \eqref{eq:Somos2qPI}, it becomes
\begin{align}
    u_{M + 1} u_{M}
    - u_{M - 1} u_{M - 2}
    &= \alpha_M \, u_{M - 1}^{-1}
    - u_M^{-1} \, \alpha_{M - 1}
    .
\end{align}
Similar considerations for $N \in 2\mathbb{N}_{> 3}$ lead to the following chain of identities:
\begin{align}
    y_{M + N} y_{M + N - 2}^{-1}
    - y_{M + 2} y_{M}^{-1}
    = \alpha_M \, y_{M + 1} y_{M + N - 2}^{-1}
    - y_{M + 2} y_{M + N - 1}^{-1} \, \alpha_{M - 1} 
    ,
    \\[2mm]
    \begin{aligned}
    y_{M + N} y_{M + N - 1}^{-1} 
    &y_{M + N - 1} y_{M + N - 2}^{-1}
    - y_{M + 2} y_{M + 1}^{-1} y_{M + 1} y_{M}^{-1}
    = \alpha_M \, y_{M + 1} y_{M + 2}^{-1} y_{M + 2} \dots
    \\[1mm]
    &\dots \, y_{M + N - 1}^{-1} y_{M + N - 1} y_{M + N - 2}^{-1}
    - y_{M + 2} y_{M + 3}^{-1} y_{M + 3} \dots y_{M + N - 2}^{-1} y_{M + N - 2} y_{M + N - 1}^{-1} \, \alpha_{M - 1} 
    ,
    \end{aligned}
    \\[2mm]
    u_{M + N - 3} u_{M + N - 4}
    - u_{M - 1} u_{M - 2}
    = \prod_{k = 0}^{N - 4} \brackets{
    \alpha_M \, u_{M + k - 1}^{-1}
    - u_{M + k}^{-1} \, \alpha_{M - 1}
    }.
\end{align}

\medskip
\textbullet \,\, \textbf{Case (b).} 
Making the change of variables \eqref{eq:Somos2qPII} in the Somos-$5$ equation of the form
\begin{align}
    y_{M + 5} y_{M + 1}^{-1}
    - y_{M + 2} y_M^{-1} y_{M + 3} y_{M + 1}^{-1}
    &= \alpha_M
    - y_{M + 2} y_{M + 4}^{-1} \, \alpha_{M - 1} \, y_{M + 3} y_{M + 1}^{-1}
    ,
\end{align}
we obtain 
\begin{align}
    u_{M + 1} u_{M - 1}
    - u_{M - 2} u_{M - 1}
    &= \alpha_M 
    - u_{M}^{-1} \, \alpha_{M - 1} \, u_{M - 1}
    .
\end{align}
In the case of $N \in 2\mathbb{N}_{> 3} + 1$, one can get:
\begin{align}
    y_{M + N} y_{M + 1}^{-1}
    - y_{M + 2} y_{M}^{-1} y_{M + N - 2} y_{M + 1}^{-1}
    = \alpha_M 
    - y_{M + 2} y_{M + N - 1}^{-1} \, \alpha_{M - 1} \, y_{M + N - 2} y_{M + 1}^{-1}
    ,
    \\[2mm]
    \begin{aligned}
    y_{M + N} y_{M + N - 2}^{-1} y_{M + N - 2} 
    &\dots
    y_{M + 3}^{-1} y_{M + 3} y_{M + 1}^{-1} 
    \\[1mm]
    & \hspace{1.6cm} 
    - y_{M + 2} y_{M}^{-1} \, 
    y_{M + N - 2} y_{M + N - 4}^{-1} y_{M + N - 4} \dots
    y_{M + 3}^{-1} y_{M + 3} y_{M + 1}^{-1}
    \\[1mm]
    &= \alpha_M 
    - y_{M + 2} y_{M + 4}^{-1} y_{M + 4} 
    \dots y_{M + N - 3}^{-1} y_{M + N - 3} y_{M + N - 1}^{-1} \, \alpha_{M - 1}
    \\[1mm]
    &\hspace{3.4cm}
    y_{M + N - 2} y_{M + N - 4}^{-1} y_{M + N - 4} \dots
    y_{M + 3}^{-1} y_{M + 3} y_{M + 1}^{-1}
    ,
    \end{aligned}
\end{align}
i.e. equation \eqref{eq:qPIIn_nc}.
\end{proof}

\begin{exmp}
Below we present first two members of the hierarchies from Proposition \ref{thm:qPIn_qPIIn_nc}.
\begin{itemize}
    \item[(a)] When $N = 4$, the \eqref{eq:qPIn_nc} is a non-commutative analog of the $q\PI$ equation:
    \begin{align}   
    \label{eq:qPI_nc}
    u_{M + 1} \, u_{M}
    - u_{M - 1} \, u_{M - 2}
    &= \alpha_M \, u_{M - 1}^{-1} - u_{M}^{-1} \, \alpha_{M - 1}.
    \end{align}
    For $N = 6$ we have:
    \begin{align}
        u_{M + 3} \, u_{M + 2}
        - u_{M - 1} \, u_{M - 2}
        &= \alpha_M \,\, u_{M - 1}^{-1} \, u_{M}^{-1} \, u_{M + 1}^{-1}
        - u_{M}^{-1} \, u_{M + 1}^{-1} \, u_{M + 2}^{-1} \,\, \alpha_{M - 1}.
    \end{align}

    \item[(b)] When $N = 5$, the \eqref{eq:qPIIn_nc} is a non-commutative analog of the $q\PPII$ equation:
    \begin{align}
    \label{eq:qPII_nc}
    u_{M + 1} \, u_{M - 1}
    - u_{M - 2} \, u_{M - 1} 
    &= \alpha_M
    - u_{M}^{-1} \, 
    \alpha_{M - 1} \, 
    u_{M - 1}.
    \end{align} 
    If $N = 7$, we obtain:
    \begin{align}
        u_{M + 3} \, u_{M + 1} \, u_{M - 1}
        - u_{M - 2} \, u_{M + 1} \, u_{M - 1}
        &= \alpha_M 
        - u_{M - 2}^{-1} \, u_{M}^{-1} \,\, \alpha_{M - 1} \,\,
        u_{M - 1} \, u_{M - 3}.
    \end{align}
\end{itemize}
\end{exmp}

\begin{rem}
In the commutative case, the equation \eqref{eq:qPI_nc} can be derived as follows. Let us take two equations \eqref{eq:qPI} with indices $M$ and $M - 1$:
\begin{align}
    &&
    u_{M + 1} \, u_{M}^{2} \, u_{M - 1}
    &= \beta + \alpha_M \, u_{M},
    &&&
    u_{M} \, u_{M - 1}^{2} \, u_{M - 2}
    &= \beta + \alpha_{M - 1} u_{M - 1},
    &&
\end{align}
or, equivalently,
\begin{align}
    &&
    u_{M + 1} \, u_{M}
    &= \beta \, 
    u_{M}^{-1} \, u_{M - 1}^{-1} 
    + \alpha_M \, u_{M - 1}^{-1},
    &&&
    u_{M - 1} \, u_{M - 2}
    &= \beta \, u_{M}^{-1} \, u_{M - 1}^{-1} 
    + \alpha_{M - 1} \, u_{M}^{-1}.
    &&
\end{align}
Then their difference leads to
\begin{align}
    &&
    u_{M + 1} \, u_{M}
    - u_{M - 1} \, u_{M - 2}
    &= \alpha_{M} \, u_{M - 1}^{-1}
    - \alpha_{M - 1} \, u_{M}^{-1}.
    &&
\end{align}
Similarly, taking the difference of two equations \eqref{eq:qPII} with indices $M$ and $M - 1$ in the form:
\begin{align}
    &&
    u_{M + 1} \, u_{M - 1}
    &= \beta \, u_{M}^{-1} + \alpha_M,
    &&&
    u_{M - 1} \, u_{M - 2}
    &= \beta \, u_{M}^{-1} 
    + \alpha_{M - 1} \, u_{M}^{-1} \, u_{M - 1},
    &&
\end{align}
we obtain commutative \eqref{eq:qPII_nc}:
\begin{align}
    &&
    u_{M + 1} \, u_{M - 1}
    - u_{M - 1} \, u_{M - 2}
    &= \alpha_M
    - \alpha_{M - 1} \, u_{M}^{-1} \, u_{M - 1}.
    &&
\end{align}
\end{rem}

\section{Discussion}

In the paper we consider continuous limits only for the Toda lattices. It would be interesting to study similar continuous limits for other discrete systems discussed in this paper. In particular, we expect an appearance of both non-abelian versions for  modified Korteweg-de Vries equations obtained in \cite{marchenko2012nonlinear} and \cite{khalilov1990matrix}, since we have derived two non-equivalent versions, \ref{eq:dmKdV_nc_1_0} and \ref{eq:dmKdV_nc_2_0}, of the discrete mKdV equation. Note that these limits may also be extended for the Lax pairs and quasideterminant solutions. 

We also expect that the considered non-autonomous one-dimensional systems possess discrete Lax pairs and quasideterminant solutions, since they result from the plane-wave reduction of the \ref{eq:2ddToda_0} lattice satisfying these properties. Note that the non-abelian analog of the autonomous Somos-$N$ type sequence should enjoy the non-commutative Laurent phenomena \cite{berenstein2011short} (see also \cite{usnich2010non}, \cite{russell2015noncommutative}), while an appropriate continuous limits for the non-abelian \ref{eq:qPIn_nc_0} and \ref{eq:qPIIn_nc_0} hierarchies should give rise to their continuous analogs generalizing the well-known hierarchies for the PI and PII equations \cite{kudryashov1997first} to the non-abelian case. As far as the authors know, this issue has not been studied yet in the commutative case. In addition, Lax pairs for the abelian discrete $q$-\Painleve hierarchies are still not known.

The commutative discrete integrable systems are related to cluster algebras \cite{hone2014discrete}, \cite{hone2019cluster}. Cluster algebras already have generalisations to the non-abelian case (see, e.g., \cite{berenstein2013noncommutative}, \cite{arthamonov2020noncommutative}). In the authors opinion, it would be interesting to study such a relation for the non-abelian discrete systems obtained in the present paper.
	
    \bibliographystyle{alpha}
    \bibliography{bib}
\end{document}